%% file: complete.tex
\documentclass[10.5pt,draftcls,onecolumn]{IEEEtran}
\usepackage{latexsym,url}
\usepackage{epsf,epsfig,verbatim,amssymb,amsmath,array,cite,amsthm,subfigure,multicol,multirow}
\usepackage{amsmath,amssymb,array,graphicx,palatino,verbatim}
\usepackage[active]{srcltx}
\usepackage{psfrag,bm,xspace}
\usepackage{setspace}
{\makeatletter
 \gdef\xxxmark{%
   \expandafter\ifx\csname @mpargs\endcsname\relax 
     \expandafter\ifx\csname @captype\endcsname\relax 
       \marginpar{xxx}
     \else
       xxx 
     \fi
   \else
     xxx 
   \fi}
 \gdef\xxx{\@ifnextchar[\xxx@lab\xxx@nolab}
 \long\gdef\xxx@lab[#1]#2{{\bf [\xxxmark #2 ---{\sc #1}]}}
 \long\gdef\xxx@nolab#1{{\bf [\xxxmark #1]}}
}

\let\realbfseries=\bfseries
\def\bfseries{\realbfseries\boldmath}

\newif\ifabstract
\abstracttrue
\newif\iffull
\ifabstract \fullfalse \else \fulltrue \fi

\let\epsilon=\varepsilon

\newtheorem{theorem}{Theorem}
\newtheorem{definition}{Definition}
\newtheorem{lemma}{Lemma}
\newtheorem{corollary}{Corollary}

\newtheorem{fact}{Fact}

\begin{document}

\bibliographystyle{ieeetr}
\newcommand{\X}{\mathcal{X}}
\newcommand{\Y}{\mathcal{Y}}
\newcommand{\Z}{\mathcal{Z}}
\renewcommand{\S}{\mathcal{S}}
\renewcommand{\P}{\mathcal{P}}
\newcommand{\Q}{\mathcal{Q}}
\newcommand{\T}{\mathcal{T}}
\newcommand{\A}{\mathcal{A}}
\newcommand{\E}{\mathcal{E}}
\newcommand{\V}{\mathcal{V}}
\newcommand{\W}{\mathcal{W}}
\newcommand{\U}{\mathcal{U}}
\newcommand{\C}{\mathcal{C}}

\title{\LARGE \bf
Error Exponent for Multiple-Access Channels:\\Lower Bounds}

\author{Ali Nazari, Achilleas Anastasopoulos and S. Sandeep Pradhan
\thanks{This work was supported by NSF grants CCF 0427385 and CCF
  0448115. The material in this paper was presented in part at the
  Information Theory and Applications conference, San Diego, and also
  at the Conference on Information Sciences and Systems, the John
  Hopkins University, Baltimore, 2009.}, \\
Department of Electrical Engineering and Computer Science, \\
University of Michigan, Ann Arbor, MI 48109, USA \\
email: {\tt\small anazari@umich.edu, anastas@umich.edu,
pradhanv@eecs.umich.edu}}

\maketitle \thispagestyle{empty} \pagestyle{plain}


\vspace{-0.3in}
\begin{abstract}
A unified framework to obtain all known lower bounds (random coding,
typical random coding and expurgated bound) on the reliability
function of a point-to-point discrete memoryless channel (DMC) is
presented. By using a similar idea for a two-user discrete
memoryless (DM) multiple-access channel (MAC), three lower bounds on
the reliability function are derived. The first one (random coding)
is identical to the best known lower bound on the reliability
function of DM-MAC. It is shown that the random coding bound is the
performance of the average code in the constant composition code
ensemble. The second bound (Typical random coding) is the typical
performance of the constant composition code ensemble. To derive the
third bound (expurgated), we eliminate some of the codewords from
the codebook with larger rate. This is the first bound of this type
that explicitly uses the method of expurgation for MACs. It is shown
that the exponent of the typical random coding and the expurgated
bounds are greater than or equal to the exponent of the known random
coding bounds for all rate pairs. Moreover, an example is given
where the exponent of the expurgated bound is strictly larger. All
these bounds can be universally obtained for all discrete memoryless
MACs with given input and output alphabets.
\end{abstract}



\section{Introduction} \label{intro}

In this paper, we consider the problem of communication over a
multiple-access channel (MAC) without feedback in the discrete
memoryless setting.  In particular, we consider the  error exponents
for this channel model.  In this model, two transmitters wish to
communicate reliably two independent messages  to a single decoder.
A schematic is depicted in Figure \ref{fig:mac_prob}.
%

\begin{figure}[htp]
\center \epsfig{figure=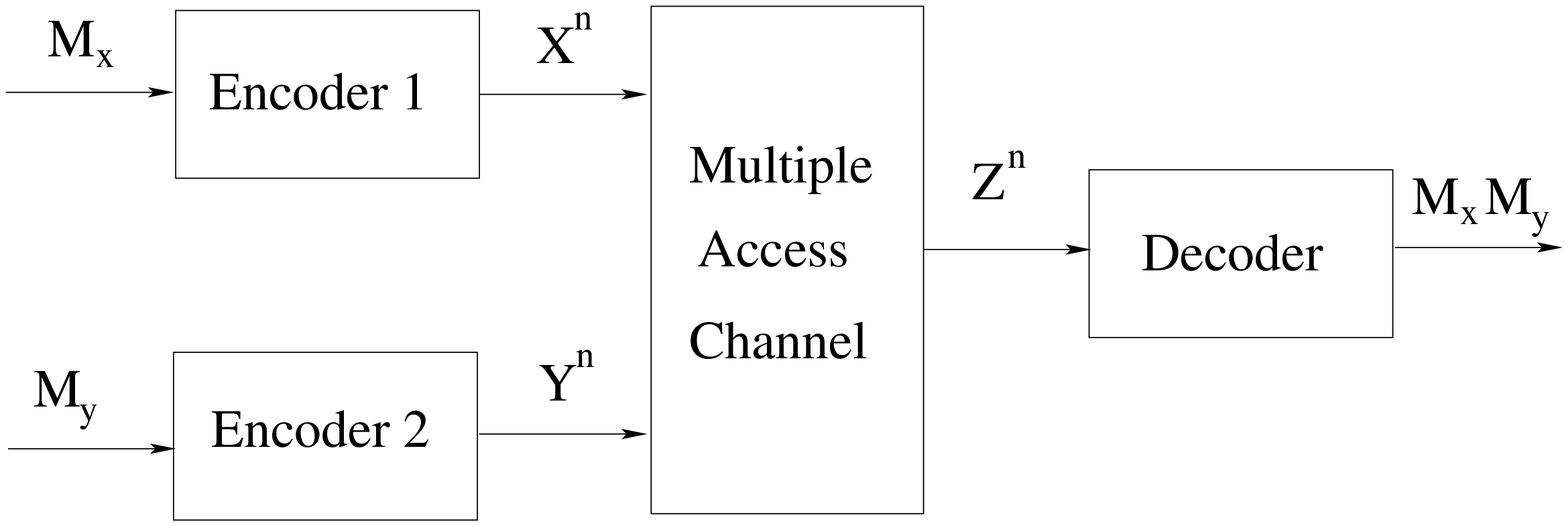, width = 0.6 \textwidth}
 \caption{{\small A schematic of
two-user multiple-access channel} } \label {fig:mac_prob}
\end{figure}

Error exponents have been meticulously studied for point to point
discrete memoryless channels (DMCs) in the
literature~\cite{Elias2,Dobrusin,FanoBook,Shan-Gall-Bere1,Shan-Gall-Bere2,Augustin,Feinstein}.
The optimum error exponent $E(R)$  at some fixed transmission rate
$R$ (also known as the channel reliability function) gives the
decoding error probability exponential rate of decay as a function
of block-length for the best sequence of codes. Lower and upper
bounds on the channel reliability function for the DMC are known. A
lower bound, known as the random coding exponent, was developed by
Fano~\cite{FanoBook} by upper-bounding the average error probability
over an ensemble of codes. This bound is loose at low rates.
Gallager~\cite{Gallager-Tightness} demonstrated that the random
coding bound is the true average error exponent for the random code
ensemble. This result illustrates that the weakness of the random
coding bound, at low rates, is not due to upper-bounding the
ensemble average. Rather, this weakness is due to the fact that the
best codes perform much better than the average, especially at low
rates. The random coding exponent is further improved at low rates
by the process of ``expurgation'' \cite{GalBook,GlBook,Cover}. The
expurgated bound coincides with the upper bound on the reliability
function at $R=0$~\cite[pg.~189]{Csiszarbook}.
Barg and Forney~\cite{Barg-RandomCode} investigated another lower
bound for the binary symmetric channel (BSC),  called the
``typical'' random coding bound. The authors showed that almost all
codes in the standard random coding ensemble exhibit a performance
that is as good as the one described by the typical random coding
bound. In addition, they showed that the typical error exponent is
larger than the random coding exponent and smaller than the
expurgated exponent at low rates.
Regarding discrete memoryless multiple-access channels (DM-MACs),
stronger versions of Ahlswede and Liao's coding
theorem~\cite{Ahlswede71,Liao}, giving exponential upper and lower
bounds for the error probability, were derived by several authors.
Slepian and Wolf~\cite{SlWo73}, Dyachkov~\cite{Dyachkov},
Gallager~\cite{Gallager-Multiaccess}, Pokorny and
Wallmeier~\cite{Pokorney}, and Liu and
Hughes~\cite{Liu-RandomCoding} studied upper bounds on the error
probability. Haroutunian~\cite{Haroutunian} and
Nazari~\cite{nazari08,nazari09,nazari09-Arx} studied lower bounds on
the error probability.

Comparing the state of the art in the study of error exponents for
DMCs and DM-MACs, we observe that the latter is much less advanced.
We believe the main difficulty in the study of error exponents for
DM-MACs is the fact that error performance in a DM-MAC depends on
the pair of codebooks (in the case of a two-user MAC) used by the
two transmitters, while at the same time, each transmitter can only
control its own codebook. This simple fact has important
consequences. For instance, expurgation has not been studied in MAC,
since by eliminating some of the ``bad'' codeword pairs, we may end
up with a set of correlated input sequences, which is hard to
analyze. In this paper, we develop two new lower bounds for the
reliability function of DM-MACs. These bound outperform the bounds
of~\cite{Pokorney,Liu-RandomCoding}.

Toward this goal, we first revisit  the point-to-point case and look
at the techniques that are used for obtaining the lower bounds on
the optimum error exponents. The techniques can be broadly
classified into three categories. The first is the Gallager
technique~\cite{Gallager-Tightness}. Although this yields
expressions for the error exponents that are computationally easier
to evaluate than others, the expressions themselves are harder to
interpret. The second is the Csiszar-Korner technique
\cite{Csiszarbook}. This technique gives more intuitive expressions
for the error exponents in terms of optimization of an objective
function involving information quantities over probability
distributions. This approach is more amenable to generalization to
multi-user channels. The third is the graph decomposition technique
using $\alpha$-decoding \cite{Csiszar-Graph}. $\alpha$-decoding is a
class of decoding procedures that includes maximum likelihood
decoding and minimum entropy decoding. Although this technique gives
a simpler derivation of the exponents, we believe that it is harder
to generalize this to multi-user channels. All three classes of
techniques give expressions for the random coding and expurgated
exponents. The expressions obtained by the three techniques appear
in different forms.

\begin{figure}[htp]\label{fig2}
\center \epsfig{figure=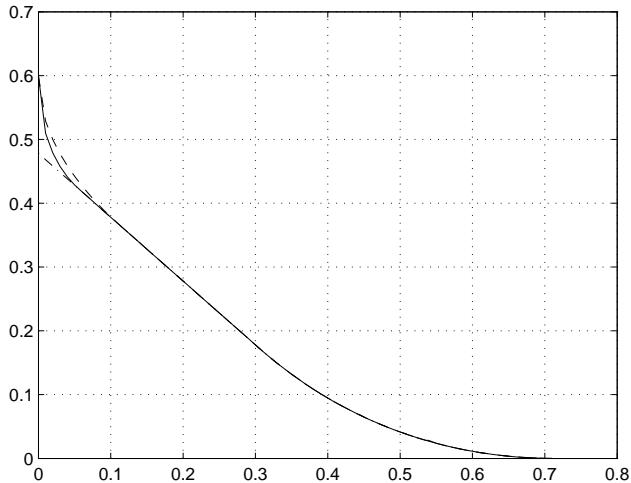, width = 0.6 \textwidth}
\caption{{\small Lower bounds on the reliability function for
point-to-point channel               (random coding~$- \cdot$,
typical random coding~$-$,         expurgated~$- -$)} }
\end{figure}

In developing our main result, we first develop a new simpler
technique for deriving the random coding and expurgated exponents
for the point-to-point channel using a constant composition code
ensemble with $\alpha$-decoding.  We present our results in the
format given in~\cite{Csiszar-Graph}. This technique also gives
upper bounds on the ensemble averages. As a bonus, we obtain the
typical random coding exponent for this channel. This gives an exact
characterization (lower and upper bounds that meet) of the error
exponent of almost all codes in the ensemble. When specialized to
the BSC, this reduces to the typical random coding bound of Barg and
Forney~\cite{Barg-RandomCode}\footnote{Barg and Forney gave only a
lower bound in~\cite{Barg-RandomCode}.}. Fig.~2 shows the random
coding, the typical random coding, and the expurgated bounds
 for a BSC with crossover probability $p =
0.05$, which is representative of the general case. All the three
lower bounds are expressed as minimizations of a single objective
function under different constraint sets. The reasons for looking at
typical performance are two-fold. The first is that the average
error exponent is in general smaller than the typical error exponent
at low rates, hence the latter gives a tighter characterization of
the optimum error exponent of the channel. For example, for the BSC,
although the average performance of the linear code ensemble is
given by the random coding exponent of the Gallager ensemble,  the
typical performance is given by the expurgated exponent of the
Gallager ensemble. In this direction, it was also noted recently
in~\cite{Como} that for the $8$-PSK Gaussian channel, the typical
performance of the ensemble of group codes over $\mathbb{Z}_8$
equals the expurgated exponent of the Gallager ensemble, whereas the
typical performance of the ensemble of binary coset codes (under any
mapping) is bounded away from the same. The second is that in some
cases, expurgation may not be possible or may not be desirable. For
example, (a) in the MAC, the standard expurgation is not possible,
and (b) if one is looking at the performance of the best linear code
for a  channel, then expurgation destroys the linear structure which
is not desirable. In the proposed technique we provide a unified way
to derive all the three lower bounds on the optimum error exponents,
and upper bounds on the ensemble average and the typical
performance. We wish to note that the bounds derived in this paper
are universal in nature. The proposed  approach appears to be more
amenable to generalization to multi-user channels.

A brief outline of the technique is given as follows. First, for a
given constant composition code, we define a pair of packing
functions that are independent of the channel. For an arbitrary
channel,  we relate the probability of error of a code with
$\alpha$-decoding   to  its packing functions. Packing functions
give pair-wise and triple-wise joint-type distributions of the code.
This is similar in spirit to the concept of distance distribution of
the code. Then we do random coding and obtain lower and upper bounds
on the expected value of the packing functions of the ensemble
without interfacing it with the channel.  That is, these bounds do
not depend on the channel. Finally, using the above relation between
the packing function and the probability of error, we get
single-letter expressions for the bounds on the optimum error
exponents for an arbitrary channel.

Toward extending  this technique to MACs, we follow a three-step
approach. We start with a constant conditional composition  ensemble
identical to~\cite{Liu-RandomCoding}. Then, we provide a new packing
lemma in which the resulting code has better properties in
comparison to the packing lemmas in~\cite{Pokorney}
and~\cite{Liu-RandomCoding}. This packing lemma is similar to
Pokorny's packing lemma, in the sense that the channel conditional
distribution does not appear in the inequalities. One of the
advantages of our methodology
is that it enables us to partially expurgate some of the codewords
and end up with a new code with stronger properties. In particular,
we do not eliminate pairs of codewords. Rather, we expurgate
codewords from only one of the codebooks and analyze the performance
of the expurgated code.

\noindent \emph{Contributions:} In summary the key contributions of
this work are
\begin{itemize}
\item  An exact characterization of the typical error exponent for the
  constant composition code ensemble for the DMC.

  \item A new lower bound on the optimum error exponent for the
  MAC.
\item An upper bound on the average error exponent of the constant
  composition code ensemble for the MAC.
\item  A characterization of the typical error exponent for the
  constant composition code ensemble for the MAC.
\end{itemize}


This paper is organized as follows: Section~\ref{sec:prelim}
introduces terminology, and Section~\ref{sec:p2p} unifies the
derivation of all lower bounds on the reliability function for a
point-to-point DMC. Our main results for the DM-MAC are introduced
in Section~\ref{sec:MACsection}. Some numerical results are
presented in Section~\ref{sec:num}, and Section~\ref{sec:conclusion}
concludes the paper. The proofs of some of these results are given
in the Appendix.

\section{Preliminaries}\label{sec:prelim}

We will follow the notation of \cite{Csiszarbook}. For any finite
alphabet $\mathcal{X}$, let $\mathcal{P(X)}$ denote the set of all
probability distributions on $\mathcal{X}$. For any sequence
$\mathbf{x} \in \mathcal{X}^n$, let $P_{\mathbf{x}}$ denote its
type. Let $T_P$ denote the type class of type $P$. Let
$\mathcal{P}_n(\mathcal{X})$ denote the set of all types on
$\mathcal{X}$. Let $T_V$ denote a V-shell, and $D(V \|W|P)$ denote
conditional I-divergence. In this paper, we consider channels
without feedback.

\begin{definition}
A discrete memoryless channel (DMC) is defined by a stochastic
matrix $W:\mathcal{X \rightarrow Y}$, where $\mathcal{X}$, the input
alphabet, and $\mathcal{Y}$, the output alphabet, are finite sets.
The channel transition probability for n-sequences is given by
\begin{eqnarray*}
W^n(\mathbf{y}|\mathbf{x}) \triangleq \prod_{i=1}^n W(y_i|x_i),
\end{eqnarray*}
where $\mathbf{x}\triangleq (x_1,...,x_n) \in \mathcal{X}^n$,
$\mathbf{y}\triangleq(y_1,...,y_n) \in \mathcal{Y}^n$. An $(n,M)$
code for a given DMC, $W$, is a set $C=\{(\mathbf{x}_i,D_{i}): 1
\leq i \leq M\}$ with (a) $\mathbf{x}_i \in \mathcal{X}^n$, $D_{i}
\subset \mathcal{Y}^n$ and (b) $D_{i} \cap D_{i'}=\varnothing$ for
$i \neq i'$.
\end{definition}
When message $i$ is transmitted, the conditional probability of
error of a code $C$ is given by
\begin{equation*}
e_{i}(C,W) \triangleq W^n(D^c_{i}|\mathbf{x}_i).
\end{equation*}
The average probability of error for this code is defined as
\begin{align}
e(C,W) \triangleq \frac{1}{M }\sum_{i=1}^{M}e_{i}(C,W).
\end{align}

\begin{definition}\label{def:typical}
For the DMC, $W:\X  \rightarrow \Y$, the average error exponent, at
rate $R$, is defined as:
\begin{align}
&E^*_{av}(R) \triangleq \limsup_{n\rightarrow \infty} \max_{
\substack{C \in \C}} -\frac{1}{n} \log{e(C,W)},
\end{align}
where $\C$ is the set of all codes of length $n$ and rate $R$. The
typical average error exponent of an ensemble $\mathcal{C}$, at rate
$R$, is defined as:
\begin{align}
&E^T_{av}(R) \triangleq \liminf_{\delta \rightarrow 0}
\limsup_{n\rightarrow \infty} \max_{ \tilde{\C}:
\mathbb{P}(\tilde{\C} ) > 1- \delta}\;\; \min_{ \substack{ C \in
\tilde{\C}}} -\frac{1}{n} \log{e(C,W)}.
\end{align}
where $\mathbb{P}$ is the uniform distribution over $\mathcal{C}$.
\end{definition}
The typical error exponent is basically the exponent of the average
error probability of the worst code belonging to the best high
probable collection of the ensemble.

\begin{definition}
A two-user DM-MAC is defined by a stochastic matrix $W:\mathcal{X
\times Y \; \rightarrow Z}$, where $\mathcal{X}$, $\mathcal{Y}$, the
input alphabets, and $\mathcal{Z}$, the output alphabet, are finite
sets. The channel transition probability for n-sequences is given by
\begin{eqnarray}
W^n(\mathbf{z}|\mathbf{x},\mathbf{y}) \triangleq \prod_{i=1}^n
W(z_i|x_i,y_i),
\end{eqnarray}
where $\mathbf{x}\triangleq (x_1,...,x_n) \in \mathcal{X}^n$,
$\mathbf{y}\triangleq(y_1,...,y_n) \in \mathcal{Y}^n$, and
$\mathbf{z}\triangleq (z_1,...,z_n) \in \mathcal{Z}^n$.\\An
$(n,M,N)$ multi-user code for a given MAC, $W$, is a set
$C=\{(\mathbf{x}_i,\mathbf{y}_j,D_{ij}): 1 \leq i \leq M, 1 \leq j
\leq N\}$ with
\begin{itemize}
\item $\mathbf{x}_i \in \mathcal{X}^n$, $\mathbf{y}_j \in
\mathcal{Y}^n$, $D_{ij} \subset \mathcal{Z}^n$
\item $D_{ij} \cap D_{i'j'}=\varnothing$ for  $(i,j) \neq
(i',j')$.
\end{itemize}
\end{definition}
When message $(i,j)$ is transmitted, the conditional probability of
error of the two-user code $C$ is given by
\begin{equation}
e_{ij}(C,W) \triangleq W^n(D^c_{ij}|\mathbf{x}_i,\mathbf{y}_j).
\end{equation}
The average probability of error for the two-user code, $C$, is
defined as
\begin{align}
e(C,W) \triangleq \frac{1}{M
N}\sum_{i=1}^{M}\sum_{j=1}^{N}e_{ij}(C,W).
\end{align}

\begin{definition}
For the MAC, $W:\X \times \Y \rightarrow \Z$, the average error
exponent at rate pair $(R_X,R_Y)$, is defined as:
\begin{align}
&E^*_{av}(R_X,R_Y) \triangleq \limsup_{n\rightarrow \infty}
\max_{\substack{C \in \C_{M}}} -\frac{1}{n}
\log{e(C,W)},
\end{align}
where $\C_{M}$ is the set of all codes of length $n$ and rate pair
$(R_X,R_Y)$. The typical average error exponent of an ensemble
$\mathcal{C}$, at rate pair $(R_X,R_Y)$, is defined as:
\begin{align}
&E^T_{av}(R_X,R_Y) \triangleq \liminf_{\delta \rightarrow 0}
\limsup_{n\rightarrow \infty} \max_{ \tilde{\C}\subset \mathcal{C}:
\mathbb{P}(\tilde{\C} ) > 1- \delta}\;\; \min_{ \substack{C \in
\tilde{\C}}} -\frac{1}{n}
\log{e(C,W)},
\end{align}
where $\mathbb{P}$ is the uniform distribution over $\mathcal{C}$.
\end{definition}


\section{Point to Point: Lower Bounds on reliability function}
\label{sec:p2p}
\subsection{Packing functions}
Consider the class of DMCs with input alphabet $\X$ and output
alphabet $\Y$.  In the following, we introduce a unified way to
derive all known lower bounds on the reliability function of such a
channel. We will follow the random coding approach. First, we choose
a constant composition code ensemble. Then, we define a packing
function, $\pi: \C \times \P(\X \times \X) \rightarrow \mathbb{R}$,
on all codebooks in the ensemble.  The packing function that we use
is the average number of codeword pairs sharing a particular joint
type, $V_{X\tilde{X}}$. Specifically, for $P \in \P_n(\X)$,
$V_{X\tilde{X}} \in \P_n(\X \times \X)$, and any code
$C=\{\mathbf{x}_1, \mathbf{x}_2,...,\mathbf{x}_M\} \subset T_P$, the
packing function is defined as:
\begin{equation}
\pi(C,V_{X\tilde{X}}) = \frac{1}{M} \sum_{i=1}^{M} \sum_{j\neq i}
1_{T_{V_{X\tilde{X}}}}(\mathbf{x}_i,\mathbf{x}_j).
\end{equation}
We call this the first order packing function. Using this packing
function, we prove three different packing lemmas, each of which
shows the existence of a code with some desired properties.

In the first packing lemma, tight upper and lower bounds on the
expectation of the packing function over the ensemble are derived.
By using this packing lemma, upper and lower bounds on the
expectation of the average probability of error over the ensemble
are derived. These bounds meet  for all transmission rates below the
critical rate\footnote{This
  is essentially a re-derivation of the upper and lower bounds on the
  average probability of error obtained by Gallager in a different
  form. The present results are for constant composition codes.}. In the
second packing lemma, by using the expectation and the variance of
the packing function, we prove that for almost all codes in the
constant composition code ensemble, the bounds in the first packing
lemma are still valid. By using this tight bound on the performance
of almost every code in the ensemble, we provide a tighter bound on
the error exponent which we call the ``typical'' random coding
bound. As we see later in the paper, the typical random coding bound
is indeed the typical performance of the constant composition code
ensemble. In the third packing lemma, we use one of the typical
codes and eliminate some of its ``bad'' codewords. The resulting
code satisfies some stronger constraints in addition to  all the
previous properties. By using this packing lemma and an efficient
decoding rule, we re-derive the well-known expurgated bound.

To provide upper bounds on the average error exponents, such as
those given below in  Fact~\ref{fc:p2p_random_coding} and
Theorem~\ref{th:p2p_typical}, for every $V_{X\tilde{X}\hat{X}} \in
\P_n\left(\X \times \X \times \X \right)$,  we define a second
packing function $\lambda : \C  \times \P(\X \times \X \times
\X)\rightarrow \mathbb{R}$ on all codes in the constant composition
code ensemble as follows:
\begin{align}
\lambda (C, V_{X\tilde{X}\hat{X}}) \triangleq
\frac{1}{M}\sum_{i=1}^{M} \sum_{j \neq i} \sum_{k \neq i,j}
1_{T_{V_{X\tilde{X}\hat{X}}}}(\mathbf{x}_i,\mathbf{x}_j,\mathbf{x}_k
).
\end{align}
We call this the second order packing function. As it is clear from
the definition, this quantity is the average number of codeword
triplets
 sharing a common joint distribution in code $C$.

\subsection{Relation between packing function and probability of
  error}

First, we consider the decoding rule at the receiver, and secondly
we relate the average probability of error to the packing function.

\noindent \textbf{Decoding Rule:} In our derivation, error
probability bounds using maximum-likelihood and minimum-entropy
decoding rules will be obtained in a unified way. The reason is that
both can be given in terms of a real-valued function on the set of
distributions on $\X \times \Y$. This type of decoding rule was
introduced in~\cite{Csiszar-Graph} as the $\alpha-decoding$ rule.
For a given real-valued  function $\alpha$, a given code $C$, and
for a received sequence $\mathbf{y} \in \Y^n$, the $\alpha-decoder$
accepts the codeword $\hat{\mathbf{x}} \in C$ for which the joint
type of $\hat{\mathbf{x}}$ and $\mathbf{y}$ minimizes the function
$\alpha$, i.e., the decoder accepts $\hat{\mathbf{x}}$ if
\begin{equation}
\hat{\mathbf{x}} = \arg\min_{\mathbf{x} \in C} \alpha(P \cdot
V_{\mathbf{y}|\mathbf{x}}).
\end{equation}
It was shown in~\cite{Csiszar-Graph} that for  fixed composition
codes, maximum-likelihood and minimum-entropy are special cases of
this decoding rule. In particular, for maximum-likelihood decoding,
\begin{equation}
\alpha(P \cdot V) = D(V||W|P) + H(V|P),
\end{equation}
and for minimum entropy decoding,
\begin{equation}
\alpha(P \cdot V) = H(V|P),
\end{equation}
where $P$ is the fixed composition of the codebook, and $V$ is the
conditional type of $\mathbf{y}$ given $\mathbf{x}$.

\noindent \textbf{Relation between probability of error and packing
function:} Next, for a given channel, we derive an upper bound and a
lower bound on the average probability of error of an arbitrary
constant composition code in terms of its first order and second
order packing functions. The rest of the paper is built on this
crucial derivation.  Consider the following argument about the
average probability of error of a code $C$ used on a channel $W$.
\begin{align}
e(C,W)  &=\frac{1}{M} \sum_{i=1}^{M}
W^n(D_i^c|\mathbf{x}_i)\nonumber = \frac{1}{M} \sum_{i=1}^{M}
W^n\left(\left\{\mathbf{y}: \alpha(P \cdot
V_{\mathbf{y}|\mathbf{x}_i}) \geq \alpha(P \cdot
V_{\mathbf{y}|\mathbf{x}_j}) \text{ for some } j \neq i
\right\}|\mathbf{x}_i\right)\nonumber\\
 &= \sum_{\substack{V_{X\tilde{X}Y} \in \P_n^r}}\left(
2^{-n[D(V_{Y|X}||W|P)+H_V(Y|X)]} \left[\frac{1}{M}\sum_{i=1}^{M}
A_i\left(V_{X\tilde{X}Y},C\right)\right]\right)\label{P2P-thm-Main-Formula1},
\end{align}
where $\P_n^r$ and $A_i\left(V_{X\tilde{X}Y},C\right)$ are defined
as follows
\begin{align}
&\P_n^r \triangleq \left\{V_{X\tilde{X}Y} \in \P_n(\X \times \X
\times \Y): V_X=V_{\tilde{X}}=P \;,\;\alpha(P \cdot V_{Y|\tilde{X}})
\leq
\alpha(P ,V_{Y|X})\right\},\\
&A_i\left(V_{X\tilde{X}Y},C\right) \triangleq
\left|\left\{\mathbf{y}: (\mathbf{x}_i,\mathbf{x}_j,\mathbf{y}) \in
T_{V_{X\tilde{X}Y}} \text{ for some } j \neq i \right\}\right|.
\end{align}
From the inclusion-exclusion principle, it follows that
 $A_i(V_{X\tilde{X}Y},C)$ satisfies
\begin{align}
B_i(V_{X\tilde{X}Y},C) - C_i(V_{X\tilde{X}Y},C) \leq
A_i(V_{X\tilde{X}Y},C) \leq B_i(V_{X\tilde{X}Y},C),\label{ABC}
\end{align}
where
\begin{align}
& B_i(V_{X\tilde{X}Y},C) \triangleq \sum_{j \neq i}
1_{T_{V_{X\tilde{X}}}}(\mathbf{x}_i,\mathbf{x}_j)\left|\left\{\mathbf{y}:
\mathbf{y} \in T_{V_{Y|X\tilde{X}}}(\mathbf{x}_i,\mathbf{x}_j)
\right\}\right|, \\
& C_i(V_{X\tilde{X}Y},C) \triangleq \sum_{j \neq i} \sum_{k \neq
i,j} 1_{T_{V_{X\tilde{X}}}}(\mathbf{x}_i,\mathbf{x}_j)
1_{T_{V_{X\tilde{X}}}}(\mathbf{x}_i,\mathbf{x}_k )
\left|\left\{\mathbf{y}: \mathbf{y} \in
T_{V_{Y|X\tilde{X}}}(\mathbf{x}_i,\mathbf{x}_j) \cap
T_{V_{Y|X\tilde{X}}}(\mathbf{x}_i,\mathbf{x}_k) \right\}\right|.
\end{align}

Next, we provide an upper bound on the second term on the right hand
side of (\ref{P2P-thm-Main-Formula1}) as follows.
\begin{subequations}
\begin{align}
\frac{1}{M}\sum_{i=1}^{M} A_i\left(V_{X\tilde{X}Y},C\right) &\leq
\frac{1}{M}\sum_{i=1}^{M}
B_i\left(V_{X\tilde{X}Y},C\right) \\
&= \frac{1}{M}\sum_{i=1}^{M}  \sum_{j \neq i}
1_{T_{V_{X\tilde{X}}}}(\mathbf{x}_i,\mathbf{x}_j)\left|\left\{\mathbf{y}:
\mathbf{y} \in T_{V_{Y|X\tilde{X}}}(\mathbf{x}_i,\mathbf{x}_j)
\right\}\right|  \\
&\leq \frac{1}{M}\sum_{i=1}^{M}  \sum_{j \neq i}
1_{T_{V_{X\tilde{X}}}}(\mathbf{x}_i,\mathbf{x}_j)  2^{nH(Y|X\tilde{X})} \\
&= \pi(C,V_{X\tilde{X}}) 2^{nH(Y|X\tilde{X})}
\end{align}
\end{subequations}
On the other hand
\begin{align}
\left\{\mathbf{y}: (\mathbf{x}_i,\mathbf{x}_j,\mathbf{y}) \in
T_{V_{X\tilde{X}Y}} \text{ for some } j \neq i \right\} \subset
T_{V_{Y|X}}(\mathbf{x}_i),
\end{align}
so we can conclude that
\begin{align}
\frac{1}{M}\sum_{i=1}^{M} A_i(V_{X\tilde{X}Y}, C)  \leq
2^{nH_V(Y|X)}.
\end{align}
Combining the above with (\ref{P2P-thm-Main-Formula1}), we have an
upper bound on the probability of error in terms of the first order
packing function as follows.
\begin{equation}
e(C,W) \leq \sum_{\substack{V_{X\tilde{X}Y} \in \P_n^r}}
2^{-n[D(V_{Y|X}||W|P)]} \min\left\{2^{-nI_V(\tilde{X} \wedge Y|X)}
\pi(C,V_{X\tilde{X}}),1\right\}  \label{upperboundA}
\end{equation}

Next, we consider the lower bound. For that, we provide a lower
bound on $B_i$ and upper bound on $C_i$ as follows.
\begin{align}
\frac{1}{M} \sum_{i=1}^{M} B_i(V_{X\tilde{X}Y},C) &= \frac{1}{M}
\sum_{i=1}^{M}\sum_{j \neq i}
1_{T_{V_{X\tilde{X}}}}(\mathbf{x}_i,\mathbf{x}_j)\left|\{\mathbf{y}:
\mathbf{y} \in T_{V_{Y|X\tilde{X}}}(\mathbf{x}_i,\mathbf{x}_j)
\}\right| \nonumber\\
&\geq  \pi(C,V_{X\tilde{X}}) 2^{n[ H(Y|X\tilde{X})-\delta]},
\label{intermmediate1}
\end{align}
and
\begin{align}
&\frac{1}{M} \sum_{i=1}^{M}
C_i(V_{X\tilde{X}Y},C) \nonumber\\
&= \frac{1}{M} \sum_{i=1}^{M} \sum_{j \neq i} \sum_{k \neq i,j}
1_{T_{V_{X\tilde{X}}}}(\mathbf{x}_i,\mathbf{x}_j)
1_{T_{V_{X\tilde{X}}}}(\mathbf{x}_i,\mathbf{x}_k )
\left|\left\{\mathbf{y}: \mathbf{y} \in
T_{V_{Y|X\tilde{X}}}(\mathbf{x}_i,\mathbf{x}_j) \cap
T_{V_{Y|X\tilde{X}}}(\mathbf{x}_i,\mathbf{x}_k)
\right\}\right|\nonumber\\
&= \sum_{\substack{V_{X\tilde{X}\hat{X}Y}:\\V_{X\hat{X}Y}
=V_{X\tilde{X}Y}}} \frac{1}{M} \sum_{i=1}^{M}\sum_{j \neq i} \sum_{k
\neq i,j}
1_{T_{V_{X\tilde{X}\hat{X}}}}(\mathbf{x}_i,\mathbf{x}_j,\mathbf{x}_k
) \left|\left\{\mathbf{y}: \mathbf{y} \in
T_{V_{Y|X\tilde{X}\hat{X}}}(\mathbf{x}_i,\mathbf{x}_j, \mathbf{x}_k)
\right\}\right|\nonumber\\
&\leq  \sum_{\substack{V_{X\tilde{X}\hat{X}Y}:\\V_{X\hat{X}Y}
=V_{X\tilde{X}Y}}} 2^{nH(Y|X\tilde{X}\hat{X})}
\lambda(C,V_{X\tilde{X}\hat{X}}) \label{intermmediate2}
\end{align}
Combining (\ref{P2P-thm-Main-Formula1}), (\ref{intermmediate1}), and
(\ref{intermmediate2}) we have the following lower bound on the
average probability of error.
\begin{equation}
e(C,W) \geq \sum_{\substack{V_{X\tilde{X}Y} \in \P_n^r}}
2^{-n[D(V_{Y|X}||W|P)+I_V(\tilde{X} \wedge Y|X)+\delta]} \left|
\pi(C,V_{X\tilde{X}}) -
\sum_{\substack{V_{X\tilde{X}\hat{X}Y}:\\V_{X\hat{X}Y}
=V_{X\tilde{X}Y}}} 2^{-n[I_V(\hat{X} \wedge Y|X \tilde{X})]}
 \lambda(C,V_{X\tilde{X}\hat{X}}) \right|^+ \label{lowerboundA}
\end{equation}
Observe that these upper and lower bounds apply for every code $C$.
We have accomplished the task of relating the average probability of
error to the two packing functions. The key results of this
subsection are given by (\ref{upperboundA}) and (\ref{lowerboundA}).
Next we use the packing lemmas to derive the bounds on the error
exponents.

\subsection{Random Coding Packing Lemmas}

\begin{lemma}\textbf{(Random Coding Packing
Lemma)}\label{P2P-random-Packing} Fix $R > 0$, $\delta > 0$, a
sufficient large $n$ and any type $P$ of sequences in
$\mathcal{X}^n$ satisfying $H(P)
> R$. For any $V_{X\tilde{X}} \in \P_n(\X \times \X)$, the expectation of the first order packing function over the constant
composition code ensemble is bounded by
\begin{equation}
2^{n(R-I_V(X \wedge \tilde{X})-\delta)} \leq \mathbb{E}
\left(\pi(X^M,V_{X\tilde{X}})\right) \leq 2^{n(R-I_V(X \wedge
\tilde{X})+\delta)},
\end{equation}
where $X^M\triangleq(X_1,X_2,...,X_M) \subset T_P$ are independent
and $X_i$s are uniformly distributed on $T_{P}$, and
$2^{n(R-\delta)} \leq M \leq 2^{nR}$. Moreover, the following
inequality holds for the second order packing function:
\begin{equation}
\mathbb{E} \left(\lambda (X^M, V_{X\tilde{X}\hat{X}})\right) \leq
2^{n[2R - I_V(X \wedge \tilde{X})  - I_V(\hat{X} \wedge X\tilde{X})
+4
 \delta]}\quad \quad \quad \text{for all } V_{X\tilde{X}\hat{X}}
\in \P_n(\X \times \X \times \X).
\end{equation}
\end{lemma}
\begin{proof}
The proof follows directly from the fact that two words drawn
independently from $T_P$ have a joint type $V_{X\tilde{X}}$ with
probability close to $2^{-nI(X \wedge \hat{X})}$.  The details are
provided in the Appendix.
\end{proof}

\begin{lemma}\textbf{(Typical Random Code Packing Lemma)}\label{lm:p2p_typical}
Fix $R > 0$, $\delta > 0$, a sufficient large $n$ and any type $P$
of sequences in $\mathcal{X}^n$ satisfying $H(P) > R$. Almost every
code, $C^t$, with $2^{n(R-\delta)} \leq M \leq 2^{nR}$ codewords, in
the constant composition code ensemble satisfies the following
inequalities
\begin{equation}\label{condition2-1}
2^{n[R-I_V(X \wedge \tilde{X})-2\delta]} \leq
\pi(C^t,V_{X\tilde{X}}) \leq 2^{n[R-I_V(X \wedge
\tilde{X})+2\delta]}\quad \quad \quad \text{for all } V_{X\tilde{X}}
\in \P_n(\X \times \X),
\end{equation}
and
\begin{equation}\label{condition2-2}
 \lambda (C^t, V_{X\tilde{X}\hat{X}}) \leq 2^{n[2R - I_V(X \wedge
\tilde{X})  - I_V(\hat{X} \wedge X\tilde{X}) +4
 \delta]}\quad \quad \quad \text{for all } V_{X\tilde{X}\hat{X}}
\in \P_n(\X \times \X \times \X).
\end{equation}

\end{lemma}
\begin{proof}
The proof is provided in the Appendix. In the proof, we evaluate the
variance of the packing function and use Chebyshev's inequality to
show that with high probability the packing function is close to its
expected value.
\end{proof}

\begin{lemma}\textbf{(Expurgated Packing Lemma)}\label{lm:p2p_expurgated}
For every sufficiently large $n$, every $R > 0$, $\delta > 0$ and
every type $P$ of sequences in $\mathcal{X}^n$ satisfying $H(P) > R$
, there exists a set of codewords $C^{ex}=\{\mathbf{x}_1,
\mathbf{x}_2,...,\mathbf{x}_{M^{*}}\} \subset T_P$ with $M^{*} \geq
\frac{2^{n(R-\delta)}}{2}$, such that for any $V_{X\tilde{X}} \in
\P_n(\X \times \X)$,
\begin{equation}\label{condition3-1}
\pi(C^{ex},V_{X\tilde{X}}) \leq 2^{n(R-I_V(X \wedge
\tilde{X})+2\delta)},
\end{equation}
and for every sequence $\mathbf{x}_i \in C^{ex}$,
\begin{equation}\label{condition3-2}
|T_{V_{\tilde{X}|X}}(\mathbf{x}_i) \cap C^{ex}| \leq  2^{n(R-I_V(X
\wedge \tilde{X})+2\delta)}.
\end{equation}
\end{lemma}
\begin{proof}
The Proof is provided in the Appendix. The basic idea of the proof
is simple. From Lemma 1 we know that for every $V_{X\tilde{X}}$,
there exists a code whose packing function is  upper bounded by a
number that is close to $2^{n(R-I_V(X \wedge
  \tilde{X}))}$. Since the packing function is an average over all
codewords in the code, we infer that  for at least half of the
codewords, the corresponding property~\eqref{condition3-2} is
satisfied. In the Appendix, we show that there exists a single code
that works for every joint type.
\end{proof}

\subsection{Error Exponent Bounds}

Now, we obtain the bounds on the error exponents using the results
from the previous three subsections. We present three lower bounds
and two upper bounds. The lower bounds are the random coding
exponent, typical random coding exponent and expurgated exponent.
All the three lower bounds are expressed as minimization of the same
objective function under different constraint sets. Similar
structure is manifested in the case of upper bounds. For
completeness, we first rederive the well-known result of random
coding exponent.

\begin{fact}\textbf{(Random Coding Bound)}\label{fc:p2p_random_coding}
For every type $P$ of sequences in $\X^n$ and $0 \leq R \leq H(P)$,
$\delta >0$, every DMC, $W: \X \rightarrow \Y$, and $2^{n(R-\delta)}
\leq M \leq 2^{nR}$, the expectation of the average error
probability over the constant composition code ensemble with $M$
codewords of type $P$, can be bounded by
\begin{align}
2^{-n[E_{rL}(R,P,W) + 3\delta]} \leq \bar{P_e} \leq
2^{-n[E_r(R,P,W)- 2\delta]},
\end{align}
whenever $n \geq n_1(|\X|,|\Y|,\delta)$, where
\begin{eqnarray}
E_r(R,P,W) &\triangleq& \min_{\substack{V_{X\tilde{X}Y} \in \P^r}}
D(V_{Y|X} || W |P) + |I_V(\tilde{X} \wedge XY)-R|^+,
\label{P2P-Er-Def}\\
E_{rL}(R,P,W) &\triangleq& \min_{\substack{V_{X\tilde{X}Y} \in
\P^r:\\I_V(\tilde{X} \wedge XY) \geq R}} D(V_{Y|X} || W |P) +
I_V(\tilde{X} \wedge XY)-R,
\end{eqnarray}
and
\begin{eqnarray}
\P^r \triangleq \big\{V_{X\tilde{X}Y} \in \P(\X \times \X \times
\Y): V_X=V_{\tilde{X}}=P \;,\;\alpha(P ,V_{Y|\tilde{X}}) \leq
\alpha(P ,V_{Y|X})\big\}.\label{P^r-def}
\end{eqnarray}
In particular,  there exists a set of codewords
$C^{r}=\{\mathbf{x}_1, \mathbf{x}_2,...,\mathbf{x}_{M}\} \subset
T_P$, with $M \geq 2^{n(R-\delta)}$, such that for every DMC, $W: \X
\rightarrow \Y$,
\begin{equation}
e(C^r,W) \leq 2^{-n[E_r(R,P,W)- 3\delta]}.
\end{equation}
\end{fact}
\begin{proof}
The proof is straightforward and is outlined  in the Appendix.
\end{proof}
It is well known that for $R\geq R_{crit}$, the random coding error
exponent is equal to the sphere packing error exponent, and as a
result the random coding bound is a tight bound. In addition, the
following is true.
\begin{corollary}\label{cor:erl_eq_er}
For any $R \leq R_{crit}$,
\begin{align}
\max_{P \in \P(\X)} E_{rL}(R,P,W) = \max_{P \in \P(\X)} E_r(R,P,W).
\end{align}
\end{corollary}
\begin{proof}
The proof is provided in the Appendix.
\end{proof}

Next we have an exact characterization of the typical performance of
the constant composition code ensemble.
\begin{theorem}\textbf{(Typical random Coding Bound)}\label{th:p2p_typical}
For every type $P$ of sequences in $\X^n$, $\delta > 0$, and every
transmission rate satisfying $0 \leq R \leq H(P)$, almost all codes,
$C^t=\{\mathbf{x}_1, \mathbf{x}_2,...,\mathbf{x}_{M}\}$ with
$\mathbf{x}_i \in T_P$ for all $i$, $M \geq 2^{n(R-\delta)}$,
satisfy
\begin{equation}
2^{-n[E_{TL}(R,P,W)+ 4\delta]} \leq e(C^t,W) \leq 2^{-n[E_T(R,P,W)-
3\delta]},\label{p2p_typical-bounds}
\end{equation}
for every DMC, $W: \X \rightarrow \Y$, whenever $n \geq
n_1(|\X|,|\Y|,\delta)$. Here,
\begin{eqnarray}
E_T(R,P,W) &\triangleq& \min_{\substack{V_{X\tilde{X}Y}} \in \P^t}
D(V_{Y|X} || W |P) + |I_V(\tilde{X} \wedge XY)-R|^+,
\label{P2P-Typ-Def}\\
E_{TL}(R,P,W) &\triangleq& \min_{\substack{V_{X\tilde{X}Y} \in
\P^t:\\I_V(\tilde{X} \wedge XY) \geq R }} D(V_{Y|X} || W |P) +
I_V(\tilde{X} \wedge XY)-R, \label{P2P-Typ-Def}
\end{eqnarray}
where
\begin{eqnarray}
\P^t \triangleq \big\{V_{X\tilde{X}Y} \in \P(\X \times \X \times
\Y): \;V_X=V_{\tilde{X}}=P,\; I_V(X \wedge \tilde{X}) \leq 2R
\;,\;\alpha(P ,V_{Y|\tilde{X}}) \leq \alpha(P ,V_{Y|X})\big\}.
\end{eqnarray}
\end{theorem}
\begin{proof}
The proof is provided in the Appendix.
\end{proof}
In Theorem~\ref{th:p2p_typical}, we proved the existence of a high
probability (almost 1) collection of codes such that every code in
this collection satisfies~\eqref{p2p_typical-bounds}. This provides
a lower bound on the typical average error exponent for the constant
composition code ensemble as defined in
Definition~\ref{def:typical}. In the following, we show that the
typical performance of the best high-probability collection cannot
be better than that given in Theorem~\ref{th:p2p_typical}.
%
\begin{corollary}\label{P2P-typical-Cor1}
For every type $P$ of sequences in $\X^n$, $\delta > 0$, and every
transmission rate satisfying $0 \leq R \leq H(P)$,
\begin{align}
E_T(R,P,W) \leq E_{av}^T (R) \leq E_{TL}(R,P,W),
\end{align}
for the constant composition code ensemble.
\end{corollary}
\begin{proof}
The proof is provided in the Appendix.
\end{proof}
Clearly, since the random coding bound is tight for $R\geq
R_{crit}$, the same is true for the typical random coding bound. For
$R\leq R_{crit}$ we have the following result.
\begin{corollary}\label{cor:et1_eq_et}
For any $R \leq R_{crit}$,
\begin{align}
\max_{P \in \P(\X)} E_{TL}(R,P,W) = \max_{P \in \P(\X)} E_T(R,P,W).
\end{align}
\end{corollary}
\begin{proof}
The proof is very similar to that of Corollary~\ref{cor:erl_eq_er}
and is omitted.
\end{proof}
It can be seen that the typical random coding bound is the true
error exponent for almost all codes, with $M$ codewords, in the
constant composition code ensemble. A similar lower bound on the
typical random coding bound was derived by Barg and
Forney~\cite{Barg-RandomCode} for the binary symmetric channel.
Although the approach used here is completely different from the one
in~\cite{Barg-RandomCode}, in the following corollary we show that
these two bounds coincide for binary symmetric channels.

\begin{corollary}\label{P2P-typical-Fact}
For a binary symmetric channel with crossover probability $p$, and
for $0 \leq R \leq R_{crit}$
\begin{align}
E_T(R,P,W) = E_{TRC}(R),
\end{align}
where $E_{TRC}$ is the lower bound for the error exponent of a
typical random code in~\cite{Barg-RandomCode}.
\end{corollary}

Finally, we re-derive the well-known  expurgated error exponent in a
rather straightforward way.


\begin{fact}\textbf{(Expurgated Bound)}\label{fc:p2p_expurgated}
For every type P of sequences in $\X^n$ and $0 \leq R \leq H(P)$,
$\delta >0$, there exists a set of codewords $C^{ex}=\{\mathbf{x}_1,
\mathbf{x}_2,...,\mathbf{x}_{M^*}\} \subset  T_P$ with $M^* \geq
\frac{2^{n(R-\delta)}}{2}$, such that for every DMC, $W: \X
\rightarrow \Y$,
\begin{equation}
e(C^{ex},W) \leq 2^{-n[E_{ex}(R,P,W)- 3\delta]}
\end{equation}
whenever $n \geq n_1(|\X|,|\Y|,\delta)$, where
\begin{equation}
E_{ex}(R,P,W) \triangleq \min_{\substack{V_{X\tilde{X}Y}} \in
\P^{ex}} D(V_{Y|X} || W |P) + |I_V(\tilde{X} \wedge XY)-R|^+
\label{P2P-Exp-Def}
\end{equation}
where
\begin{eqnarray}
\P^{ex} \triangleq \big\{V_{X\tilde{X}Y} \in \P(\X \times \X \times
\Y): \;V_X=V_{\tilde{X}}=P,\; \;\; I_V(X \wedge \tilde{X}) \leq R
\;,\;\alpha(P ,V_{Y|\tilde{X}}) \leq \alpha(P ,V_{Y|X})\big\}
\end{eqnarray}
\end{fact}
\begin{proof}
The proof is provided in the Appendix.
\end{proof}
Note that none of the mentioned three bounds have their
``traditional format'' as found
in~\cite{Csiszarbook},~\cite{GalBook}, but rather the format
introduced in~\cite{Csiszar-Graph} by Csiszar and Korner. It was
shown in~\cite{Csiszar-Graph} that the new random coding bound is
equivalent to the original one for maximum likelihood and minimum
entropy decoding rule. Furthermore, the new format for the
expurgated bound is equivalent to the traditional one for maximum
likelihood-decoding and it results in a bound that is the maximum of
the traditional expurgated and random coding bounds.


\section{MAC: Lower Bounds on reliability
function}\label{sec:MACsection}
Consider a DM-MAC, $W$, with input alphabets $\X$ and $\Y$, and
output alphabet $\Z$. In this section, we present three achievable
lower bounds on the reliability function (upper bound on the average
error probability) for this channel. The method we are using is very
similar to the point-to-point case. Again, the goal is first proving
the existence of a good code and then analyzing its performance. The
first step is choosing the ensemble. The ensemble, $\C$, we are
using is similar to the ensemble in~\cite{Liu-RandomCoding}. For a
fixed distribution, $P_U P_{X|U} P_{Y|U}$, the codewords of each
code in the ensemble are chosen from $T_{P_{X|U}}(\mathbf{u})$ and
$T_{P_{Y|U}}(\mathbf{u})$ for some sequence $\mathbf{u} \in
T_{P_U}$. Intuitively, we expect that the codewords in a ``good''
code must be far from each other. In accordance with the ideas of
Csiszar and Korner~\cite{Csiszarbook}, we use conditional types to
quantify this statement. We select a prescribed number of sequences
in $\X^n$ and $\Y^n$ so that the shells around each pair have small
intersections with the shells around other sequences. In general,
two types of packing lemmas have been studied in the literature
based on whether the shells are defined on the channel input space
or channel output space. The packing lemma in~\cite{Pokorney}
belongs to the first type, and the one in~\cite{Liu-RandomCoding}
belongs to the second type. All the inequalities in the first type
depend only on the channel input sequences. However, in the second
type, the lemma incorporates the channel output into the packing
inequalities. In this work, we use the first type. In the following,
we follow a four step procedure to arrive at the error exponent
bounds. In step one, we define first-order and second-order packing
functions. These functions are independent of the channel
statistics. Next, in step two, for any constant composition code and
any DM-MAC, we provide upper and lower bounds on the probability of
decoding error in terms of these packing functions. In step three,
by using a random coding argument on the constant composition code
ensemble, we show the existence of codes whose packing functions
satisfy certain conditions. Finally, in step four, by connecting the
results in step two and three, we provide lower and upper bounds on
the error exponents. Our results include a new tighter lower bound
on the error exponent for DM-MAC using a new partial expurgation
method for multi-user codes. We also give a tight characterization
of the typical performance of the constant composition code
ensemble. Both the expurgated bound as well as the typical bound
outperform the random coding bound of~\cite{Liu-RandomCoding}, which
is derived as special case of our methodology.

\subsection{Definition of Packing Functions}
Let $C_X=\{\mathbf{x}_1,\mathbf{x}_2,...,\mathbf{x}_{M_X}\}$ and
$C_Y=\{\mathbf{y}_1,\mathbf{y}_2,...,\mathbf{y}_{M_Y}\} $ be
constant composition codebooks with $ \mathbf{x}_i \in
T_{P_{X|U}}(\mathbf{u})$ and $\mathbf{y}_j \in
T_{P_{Y|U}}(\mathbf{u})$, for some $\mathbf{u} \in T_{P_{U}}$. In
the following, for a two-user code $C=C_X \times C_Y$, we define the
following quantities that we will use later in this section.
\begin{definition}
Fix a finite set $\U$, and a joint type $V_{UXY\tilde{X}\tilde{Y}}
\in \P_n(\U \times (\X \times \Y)^2)$. For code $C$, the first-order
packing functions are defined as follows:
\begin{subequations}
\begin{align}
N_{U}(C,V_{UXY}) &\triangleq \frac{1}{M_X M_Y} \sum_{i=1}^{M_X}
\sum_{j=1}^{M_Y}
1_{T_{V_{UXY}}}(\mathbf{u}, \mathbf{x}_i, \mathbf{y}_j),\label{NU-def}\\
N_{X}(C,V_{UXY\tilde{X}}) &\triangleq \frac{1}{M_X M_Y}
\sum_{i=1}^{M_X} \sum_{j=1}^{M_Y} \sum_{k \neq i}
1_{T_{V_{UXY\tilde{X}}}}(\mathbf{u}, \mathbf{x}_i, \mathbf{y}_j,\mathbf{x}_k),\\
N_{Y}(C,V_{UXY\tilde{Y}}) &\triangleq \frac{1}{M_X M_Y}
\sum_{i=1}^{M_X} \sum_{j=1}^{M_Y} \sum_{l \neq j}
1_{T_{V_{UXY\tilde{Y}}}}(\mathbf{u}, \mathbf{x}_i, \mathbf{y}_j , \mathbf{y}_l),\\
N_{XY}(C,V_{UXY\tilde{X}\tilde{Y}}) &\triangleq \frac{1}{M_X M_Y}
\sum_{i=1}^{M_X} \sum_{j=1}^{M_Y} \sum_{k \neq i} \sum_{l \neq j}
1_{T_{V_{UXY\tilde{X}\tilde{Y}}}}(\mathbf{u}, \mathbf{x}_i,
\mathbf{y}_j, \mathbf{x}_k, \mathbf{y}_l).\label{NXY-def}
\end{align}
\end{subequations}

Moreover, for any $V_{UXY\tilde{X}\tilde{Y}\hat{X}\hat{Y}} \in
\P_n\left(\U \times (\X \times \Y)^3\right)$, we define a set of
second-order packing functions as follows:
\begin{subequations}
\begin{align}
\Lambda_X (C,V_{UXY\tilde{X}\hat{X}}) &\triangleq \frac{1}{M_X M_Y}
\sum_{i,j} \sum_{\substack{k \neq i}}\sum_{\substack{k' \neq i,k}}
1_{T_{V_{UXY\tilde{X}\hat{X}}}}(\mathbf{u},\mathbf{x}_i,\mathbf{y}_j,\mathbf{x}_k,
\mathbf{x}_{k'}),\\
\Lambda_Y (C,V_{UXY\tilde{Y}\hat{Y}}) &\triangleq \frac{1}{M_X M_Y}
\sum_{i,j} \sum_{\substack{l \neq j}}\sum_{\substack{l' \neq j,l}}
1_{T_{V_{UXY\tilde{Y}\hat{Y}}}}(\mathbf{u},\mathbf{x}_i,\mathbf{y}_j,\mathbf{y}_l,
\mathbf{y}_{l'}),\\
\Lambda_{XY} (C,V_{UXY\tilde{X}\tilde{Y}\hat{X}\hat{Y}}) &\triangleq
\frac{1}{M_X M_Y} \sum_{i,j} \sum_{\substack{k \neq i\\
l \neq j}}\sum_{\substack{k' \neq i,k\\ l' \neq j,l}}
1_{T_{V_{UXY\tilde{X}\tilde{Y}\hat{X}\hat{Y}}}}(\mathbf{u},\mathbf{x}_i,\mathbf{y}_j,\mathbf{x}_k,\mathbf{y}_l,\mathbf{x}_{k'},\mathbf{y}_{l'}).
\end{align}
\end{subequations}
\end{definition}
The second-order packing functions are used to prove the tightness
of the results of Theorem~\ref{randomcodingthm} and
Theorem~\ref{Typrandomcodingthm}.
%
%
Next we will obtain  upper and lower bounds on the probability of
decoding error for an arbitrary two-user code that depend on its
packing functions defined above.

\subsection{Relation between probability of error and packing functions}\label{MAC-relation}
 Consider the multiuser code $C$ as defined above, and a function $\alpha: \P(\U \times \X \times \Y \times \Z) \rightarrow \mathbb{R}$.
Taking into account the given $\mathbf{u}$, $\alpha$-decoding yields
the decoding sets
\begin{eqnarray}
D_{ij}= \left\{\mathbf{z}:
\alpha(P_{\mathbf{u},\mathbf{x}_i,\mathbf{y}_j, \mathbf{z}}) \leq
\alpha(P_{\mathbf{u},\mathbf{x}_k,\mathbf{y}_l, \mathbf{z}}) \text{
for all } (k,l) \neq (i,j)\right\}.
\end{eqnarray}
The average error probability of this multiuser code on DM-MAC $W$,
can be written as
\begin{align}
e(C,W) &\triangleq \frac{1}{M_X M_Y} \sum_{\substack{i,j}} W^n(D^c_{ij}|\mathbf{x}_i,\mathbf{y}_j)\nonumber\\
&=\frac{1}{M_X M_Y}\sum_{\substack{i,j}} W^n(\bigcup_{\substack{k
\neq i}}D_{kj}|\mathbf{x}_i,\mathbf{y}_j) +\frac{1}{M_X M_Y}
\sum_{\substack{i,j}} W^n(\bigcup_{\substack{l \neq
j}}D_{il}|\mathbf{x}_i,\mathbf{y}_j) +\frac{1}{M_X M_Y}
\sum_{\substack{i,j}} W^n(\bigcup_{\substack{k \neq i\\ l \neq
j}}D_{kl}|\mathbf{x}_i,\mathbf{y}_j).\label{proof-thm1-1}
\end{align}
The first term on the right side of~\eqref{proof-thm1-1} can be
written as
\begin{align}
&\frac{1}{M_X M_Y}\sum_{\substack{i,j}} W^n(\bigcup_{\substack{k
\neq i}}D_{kj}|\mathbf{x}_i,\mathbf{y}_j) \nonumber\\
&= \frac{1}{M_X M_Y} \sum_{\substack{i,j}}
W^n\Big(\left\{\mathbf{z}:\alpha(P_{\mathbf{u},\mathbf{x}_k,\mathbf{y}_j,
\mathbf{z}}) \leq \alpha(P_{\mathbf{u},\mathbf{x}_i,\mathbf{y}_j,
\mathbf{z}}), \text{for some } k  \neq i \right\}|\mathbf{u}, \mathbf{x}_i,\mathbf{y}_j \Big)\;\;\;\;\;\;\;\;\;\;\;\nonumber\\
&=  \frac{1}{M_X M_Y}
\sum_{\substack{i,j}} \sum_{\substack{\mathbf{z}:\\
\alpha(P_{\mathbf{u},\mathbf{x}_k,\mathbf{y}_j, \mathbf{z}}) \leq
\alpha(P_{\mathbf{u},\mathbf{x}_i,\mathbf{y}_j, \mathbf{z}})\\
\text{for
some } k  \neq i}}W^n\left(\mathbf{z}|\mathbf{u}, \mathbf{x}_i,\mathbf{y}_j \right)\nonumber\\
&= \frac{1}{M_X M_Y}
\sum_{\substack{i,j}}\sum_{\substack{V_{UXY\tilde{X}Z} \in
\V^r_{X,n}}}
\sum_{\substack{\mathbf{z}:\\
\alpha(P_{\mathbf{u},\mathbf{x}_k,\mathbf{y}_j, \mathbf{z}}) \leq
\alpha(P_{\mathbf{u},\mathbf{x}_i,\mathbf{y}_j, \mathbf{z}})\\
\text{for some } k  \neq i}}  1_{T_{V_{UXY\tilde{X}Z}}}(\mathbf{u},
\mathbf{x}_i,
\mathbf{y}_j,\mathbf{x}_k,\mathbf{z})W^n\left(\mathbf{z}|\mathbf{u}, \mathbf{x}_i,\mathbf{y}_j \right)\nonumber\\
&= \sum_{\substack{V_{UXY\tilde{X}Z} \in \V^r_{X,n}}}
2^{-n[D(V_{Z|XYU} || W|V_{XYU}) + H_V(Z|XYU)]}  \cdot \Big[
\frac{1}{M_X M_Y} \sum_{\substack{i,j}} 1_{T_{V_{UXY}}}(\mathbf{u},
\mathbf{x}_i, \mathbf{y}_j) \cdot
A^X_{i,j}\left(V_{UXY\tilde{X}Z},C\right)\Big]\label{form1-1},
\end{align}
where
\begin{align}
&A^X_{i,j}\left(V_{UXY\tilde{X}Z},C\right) \triangleq
\big|\{\mathbf{z}:(\mathbf{u},\mathbf{x}_i,\mathbf{y}_j,
\mathbf{x}_k, \mathbf{z}) \in T_{V_{UXY\tilde{X}Z}} \text{for some }
k \neq i\}\big|\nonumber\\
& \V^r_{X,n} \triangleq \left\{V_{UXY\tilde{X}Z} :\alpha(V_{UXYZ})
\geq \alpha(V_{U\tilde{X}YZ}), V_{UX}=V_{U\tilde{X}}= P_{UX},
V_{UY}= P_{UY}  \right\}.\label{VXr}
\end{align}
Note that $\V^r_{X,n}$ is a set of types of resolution $n$,
therefore, we use a subscript $n$ to define it.  Similarly, the
second and third term term on the right side of~\eqref{proof-thm1-1}
can be written as follows:
\begin{align}
&\frac{1}{M_X M_Y}\sum_{\substack{i,j}} W^n(\bigcup_{\substack{l
\neq j}}D_{il}|\mathbf{x}_i,\mathbf{y}_j) \nonumber\\
&= \sum_{\substack{V_{UXY\tilde{Y}Z} \in \V^r_{Y,n}}}
2^{-n[D(V_{Z|XYU} || W|V_{XYU}) + H_V(Z|XYU)]} .\Big[ \frac{1}{M_X
M_Y} \sum_{\substack{i,j}} 1_{T_{V_{UXY}}}(\mathbf{u}, \mathbf{x}_i,
\mathbf{y}_j).A^Y_{i,j}\left(V_{UXY\tilde{Y}Z},C\right)\Big]\label{form1-1},
\end{align}
where
\begin{align}
&A^Y_{i,j}\left(V_{UXY\tilde{Y}Z},C\right) \triangleq
\big|\{\mathbf{z}:(\mathbf{u},\mathbf{x}_i,\mathbf{y}_j,
\mathbf{y}_l, \mathbf{z}) \in T_{V_{UXY\tilde{Y}Z}} \text{for some }
l \neq j\}\big|\nonumber\\
& \V^r_{Y,n} \triangleq \left\{V_{UXY\tilde{Y}Z} :\alpha(V_{UXYZ})
\geq \alpha(V_{UX\tilde{Y}Z}), V_{UX}= P_{UX},
V_{UY}=V_{U\tilde{Y}}=P_{UY}  \right\},\label{VYr}
\end{align}
and,
\begin{align}
&\frac{1}{M_X M_Y} \sum_{\substack{i,j}} W^n(\bigcup_{\substack{k
\neq i\\ l \neq
j}}D_{kl}|\mathbf{x}_i,\mathbf{y}_j)\nonumber\\
&= \sum_{\substack{V_{UXY\tilde{X}\tilde{Y}Z} \in \V^r_{XY,n}}}
2^{-n[D(V_{Z|XYU} || W|V_{XYU}) + H_V(Z|XYU)]} \cdot \Big[
\frac{1}{M_X M_Y} \sum_{\substack{i,j}} 1_{T_{V_{UXY}}}(\mathbf{u},
\mathbf{x}_i,
\mathbf{y}_j).A^{XY}_{i,j}\left(V_{UXY\tilde{X}\tilde{Y}Z},C\right)\Big]\label{form1-1},
\end{align}
where
\begin{align}
&A^{XY}_{i,j}\left(V_{UXY\tilde{X}\tilde{Y}Z},C\right) \triangleq
\big|\{\mathbf{z}:(\mathbf{u},\mathbf{x}_i,\mathbf{y}_j,\mathbf{x}_k,
\mathbf{y}_l, \mathbf{z}) \in T_{V_{UXY\tilde{X}\tilde{Y}Z}}
\text{for some }
k \neq i ,l \neq j\}\big|\nonumber\\
& \V^r_{XY,n} \triangleq \left\{V_{UXY\tilde{X}\tilde{Y}Z}
:\alpha(V_{UXYZ}) \geq \alpha(V_{U\tilde{X}\tilde{Y}Z}),
V_{UX}=V_{U\tilde{X}}= P_{UX}, V_{UY}=V_{U\tilde{Y}}=P_{UY}
\right\}.\label{VXYr}
\end{align}
Clearly, $A^X_{i,j}\left(V_{UXY\tilde{X}Z}\right)$ satisfies
\begin{align}
B^X_{i,j}\left(V_{UXY\tilde{X}Z},C\right) -
C^X_{i,j}\left(V_{UXY\tilde{X}Z},C\right) \leq
A^X_{i,j}\left(V_{UXY\tilde{X}Z},C\right) \leq
B^X_{i,j}\left(V_{UXY\tilde{X}Z},C\right),\label{Upper-LowerX-Random-Coding}
\end{align}
where
\begin{align}
&B^X_{i,j}\left(V_{UXY\tilde{X}Z},C\right) \triangleq \sum_{k \neq
i} 1_{T_{V_{UXY\tilde{X}}}}(\mathbf{u}, \mathbf{x}_i, \mathbf{y}_j,
\mathbf{x}_k) .\big|\{\mathbf{z}:\mathbf{z} \in
T_{V_{Z|UXY\tilde{X}}}(\mathbf{u},\mathbf{x}_i,\mathbf{y}_j,
\mathbf{x}_k\}\big|,  \label{BXij-def}\\
&C^X_{i,j}\left(V_{UXY\tilde{X}Z},C\right) \triangleq  \sum_{k \neq
i} \sum_{k' \neq k,i} 1_{T_{V_{UXY\tilde{X}}}}(\mathbf{u},
\mathbf{x}_i, \mathbf{y}_j, \mathbf{x}_k)
1_{T_{V_{UXY\tilde{X}}}}(\mathbf{u}, \mathbf{x}_i, \mathbf{y}_j,
\mathbf{x}_{k'}) \;\;\;\;\;\;\;\;\;\;\;\;\;\;\;\;\;\;\;\;\;\;\;\;\nonumber\\
&\quad\quad\quad \quad \quad \quad\quad\quad
\quad\quad\quad\quad\quad\quad \cdot \big|\{\mathbf{z}:\mathbf{z}
\in T_{V_{Z|UXY\tilde{X}}}(\mathbf{u},\mathbf{x}_i,\mathbf{y}_j,
\mathbf{x}_k) \cap
T_{V_{Z|UXY\tilde{X}}}(\mathbf{u},\mathbf{x}_i,\mathbf{y}_j,
\mathbf{x}_{k'})\} \big|.\label{CXij-def}
\end{align}
Having related the probability of error and the function
$B^X_{i,j}$, $B^Y_{i,j}$ and $B^{XY}_{i,j}$, our next task is to
provide a simple upper bound on these functions. This is done as
follows.
\begin{align}
&\frac{1}{M_X M_Y} \sum_{\substack{i,j}} 1_{T_{V_{UXY}}}(\mathbf{u},
\mathbf{x}_i,
\mathbf{y}_j)B^X_{i,j}\left(V_{UXY\tilde{X}Z},C\right)\nonumber\\
&\quad\quad= \frac{1}{M_X M_Y}\sum_{\substack{i,j}}\sum_{k \neq i}
1_{T_{V_{UXY\tilde{X}}}}(\mathbf{u}, \mathbf{x}_i, \mathbf{y}_j,
\mathbf{x}_k) \left|\left\{\mathbf{z}:\mathbf{z} \in
T_{V_{Z|UXY\tilde{X}}}(\mathbf{u},\mathbf{x}_i,\mathbf{y}_j,
\mathbf{x}_k)\right\}\right|\nonumber\\
& \quad\quad\leq 2^{nH(Z|UXY\tilde{X})} \frac{1}{M_X
M_Y}\sum_{\substack{i,j}}\sum_{k \neq i}
1_{T_{V_{UXY\tilde{X}}}}(\mathbf{u}, \mathbf{x}_i, \mathbf{y}_j,
\mathbf{x}_k)\nonumber\\
&\quad\quad= 2^{nH(Z|UXY\tilde{X})} N_X(C, V_{UXY\tilde{X}} )
\end{align}
Similarly, we can provide upper bounds for $B^Y_{i,j}$ and
$B^{XY}_{i,j}$. Moreover, we can also provide trivial upper bounds
on $A(\cdot)$ functions as was done in the point-to-point case.
\[
A^X_{i,j}(V_{UXY\tilde{X}Z},C) \leq 2^{nH_V(Z|XYU)}.
\]
The same bound applies to $A^Y$ and $A^{XY}$. Collecting all these
results, we provide the following upper bound on the probability of
error.
\begin{align}
e(C,W) &\leq \sum_{\substack{V_{UXY\tilde{X}Z} \\ \in \V^r_{X,n}}}
2^{-n[D(V_{Z|XYU} || W|V_{XYU})]} \min\left\{ 2^{-nI_V(\tilde{X}
\wedge
  Z|XYU)}  N_X(C,V_{UXY\tilde{X}}),1 \right\}  \nonumber\\
& + \sum_{\substack{V_{UXY\tilde{Y}Z} \\ \in \V^r_{Y,n}}}
2^{-n[D(V_{Z|XYU} || W|V_{XYU})]} \min\left\{ 2^{-n I_V(\tilde{Y}
\wedge Z|XYU)}
N_Y(C,V_{UXY\tilde{Y}}),1 \right\}  \nonumber\\
&  + \sum_{\substack{V_{UXY\tilde{X}\tilde{Y}Z}\\ \in \V^r_{XY,n}}}
2^{-n[D(V_{Z|XYU} || W|V_{XYU})]} \min\left\{ 2^{-n I_V(\tilde{X}
\tilde{Y} \wedge
    Z|XYU)}  N_{XY}(C,V_{UXY\tilde{X}\tilde{Y}}),1 \right\}
\label{upperboundMAC}
\end{align}

Next, we consider lower bounds on $B(\cdot)$ functions and upper
bounds on $C(\cdot)$ functions. One can use a similar argument to
show the following
\begin{align}
\frac{1}{M_X M_Y} \sum_{\substack{i,j}} 1_{T_{V_{UXY}}}(\mathbf{u},
\mathbf{x}_i, \mathbf{y}_j)B^X_{i,j}\left(V_{UXY\tilde{X}Z},C\right)
&\geq 2^{n[H(Z|UXY\tilde{X})-\delta]}  N_X(C, V_{UXY\tilde{X}} ).
\nonumber
\end{align}
Similar lower bounds can be obtained for $B^Y$ and $B^{XY}$.
Moreover, we have the following arguments for bounding from above
the function $C^X$.
\begin{align}
&\frac{1}{M_X M_Y} \sum_{\substack{i,j}}1_{T_{V_{UXY}}}(\mathbf{u},
\mathbf{x}_i,
\mathbf{y}_j) \cdot C^X_{i,j}\left(V_{UXY\tilde{X}Z}\right)\nonumber\\
&= \frac{1}{M_X M_Y} \sum_{\substack{i,j}}
1_{T_{V_{UXY}}}(\mathbf{u}, \mathbf{x}_i, \mathbf{y}_j) \sum_{k \neq
i} \sum_{k' \neq k,i} 1_{T_{V_{UXY\tilde{X}}}}(\mathbf{u},
\mathbf{x}_i, \mathbf{y}_j, \mathbf{x}_k)
1_{T_{V_{UXY\tilde{X}}}}(\mathbf{u}, \mathbf{x}_i, \mathbf{y}_j,
\mathbf{x}_{k'}) \;\;\;\;\;\;\;\nonumber\\
&\quad \quad \quad \quad\quad\quad \quad\quad\quad\quad\quad\quad
\cdot \left|\left\{\mathbf{z}:\mathbf{z} \in
T_{V_{Z|UXY\tilde{X}}}(\mathbf{u},\mathbf{x}_i,\mathbf{y}_j,
\mathbf{x}_k) \cap
T_{V_{Z|UXY\tilde{X}}}(\mathbf{u},\mathbf{x}_i,\mathbf{y}_j,
\mathbf{x}_{k'})\right\}
\right|\nonumber\\
& = \frac{1}{M_X M_Y} \sum_{\substack{i,j}}\sum_{\substack{V_{UXY\tilde{X}\hat{X}Z} :\\
V_{UXY\hat{X}Z}= V_{UXY\tilde{X}Z}}} \sum_{k \neq i} \sum_{k' \neq
k,i}  1_{T_{V_{UXY\tilde{X}\hat{X}}}}(\mathbf{u}, \mathbf{x}_i,
\mathbf{y}_j,
\mathbf{x}_k,\mathbf{x}_{k'})\left|\left\{\mathbf{z}:\mathbf{z} \in
T_{V_{Z|UXY\tilde{X}\hat{X}}}(\mathbf{u},\mathbf{x}_i,\mathbf{y}_j,
\mathbf{x}_k,\mathbf{x}_{k'})\right\}
\right|\nonumber\\
& \leq  \sum_{\substack{V_{UXY\tilde{X}\hat{X}Z} :\\
V_{UXY\hat{X}Z}= V_{UXY\tilde{X}Z}}} 2^{nH(Z|UXY\tilde{X}\hat{X})}
\frac{1}{M_X M_Y} \sum_{\substack{i,j}} \sum_{k \neq i} \sum_{k'
\neq k,i} 1_{T_{V_{UXY\tilde{X}\hat{X}}}}(\mathbf{u},\mathbf{x}_i,
\mathbf{y}_j,\mathbf{x}_k,\mathbf{x}_{k'})  \nonumber\\
&= \sum_{\substack{V_{UXY\tilde{X}\hat{X}Z} :\\
V_{UXY\hat{X}Z}= V_{UXY\tilde{X}Z}}} 2^{nH(Z|UXY\tilde{X}\hat{X})}
\Lambda_X (C_, V_{UXY\tilde{X}\hat{X}}).
\end{align}
Similar relation can be obtained that relate $C^Y$ and $\lambda_Y$,
$C^{XY}$ and $\lambda_{XY}$. Combining the lower bounds on
$B(\cdot)$-functions and upper bounds on $C(\cdot)$-functions, we
have the following lower bound on the probability of decoding error.

\begin{align}
&e(C,W)\nonumber \\
&\geq \sum_{\substack{V_{UXY\tilde{X}Z} \\ \in \V^r_{X,n}}}
2^{-n[D(V_{Z|XYU} || W|V_{XYU}) + I_V(\tilde{X} \wedge
Z|XYU)+\delta]}
\left|N_X- \sum_{\substack{V_{UXY\tilde{X}\hat{X}Z} :\\
V_{UXY\hat{X}Z}= V_{UXY\tilde{X}Z}}} 2^{nI(\hat{X} \wedge
Z|UXY\tilde{X})} \Lambda_X \right|^+
\nonumber\\
&  +\sum_{\substack{V_{UXY\tilde{Y}Z} \\ \in \V^r_{Y,n}}}
2^{-n[D(V_{Z|XYU} || W|V_{XYU}) + I_V(\tilde{Y} \wedge
Z|XYU)+\delta]}
\left|N_Y- \sum_{\substack{V_{UXY\tilde{Y}\hat{Y}Z} :\\
V_{UXY\hat{Y}Z}= V_{UXY\tilde{Y}Z}}} 2^{nI(\hat{Y} \wedge
Z|UXY\tilde{Y})} \Lambda_Y \right|^+
\nonumber\\
& +\sum_{\substack{V_{UXY\tilde{X}\tilde{Y}Z} \\ \in \V^r_{XY,n}}}
2^{-n[D(V_{Z|XYU} || W|V_{XYU}) + I_V(\tilde{X} \tilde{Y}\wedge
Z|XYU)+\delta]}
\left|N_{XY}- \sum_{\substack{V_{UXY\tilde{X}\hat{X}\tilde{Y}\hat{Y}Z} :\\
V_{UXY\hat{X}\hat{Y}Z}= V_{UXY\tilde{X}\tilde{Y}Z}}} 2^{nI(\hat{X}
\hat{Y}\wedge Z|UXY\tilde{X}\tilde{Y})} \Lambda_{XY} \right|^+.
\label{lowerboundMAC}
\end{align}

This completes our task of relating the average probability of error
of any code $C$ in terms of the first and the second order packing
functions. We next proceed toward obtaining lower bounds on the
error exponents. The expressions for the error exponents that we
derive are conceptually very similar to those derived for the
point-to-point channels. However, since we have to deal with a
bigger class of error events, the expressions for the error
exponents become longer. To state our results concisely, in the next
subsection, we define certain functions of information quantities
and transmission rates. We will express our results in terms of
these functions. The reader can skip this subsection, and move to
the next subsection without losing the flow of the exposition. The
reader can  come back to it when we refer to it in the subsequent
discussions.

\subsection{Definition of Information Functions}

In the following, we consider five definitions which are mainly used
for conciseness.
\begin{definition}
For any fix rate pair $R_X, R_Y \geq 0$ , and any distribution
$V_{UXY\tilde{X}\tilde{Y}} \in \P\left(\U \times (\X \times
\Y)^2\right)$, we define
\begin{subequations}
\begin{align}
F_U(V_{UXY}) &\triangleq I(X \wedge Y|U),\label{def1}\;\;\;\;\;\;\;\;\;\;\;\;\;\;\;\;\;\;\;\;\;\;\;\;\;\;\;\;\;\;\;\;\;\;\;\;\\
F_X(V_{UXY\tilde{X}}) &\triangleq I(X \wedge Y|U) + I_V(\tilde{X}
\wedge XY|U)
-R_X,\;\;\;\;\\
F_Y(V_{UXY\tilde{Y}}) &\triangleq I(X \wedge Y|U)
+I(\tilde{Y} \wedge XY|U)-R_Y, \;\;\;\;\\
F_{XY}(V_{UXY\tilde{X}\tilde{Y}}) &\triangleq I(X \wedge Y|U) +
I(\tilde{X} \wedge \tilde{Y}|U) +I(\tilde{X}\tilde{Y} \wedge
XY|U)-R_X -R_Y.\label{def4}
\end{align}
\end{subequations}
Moreover, for any $V_{UXY\tilde{X}\tilde{Y}\hat{X}\hat{Y}} \in
\P\left(\U \times (\X \times \Y)^3\right)$, we define
\begin{subequations}
\begin{align}
&E^X_S(V_{UXY\tilde{X}\hat{X}}) \triangleq I(\hat{X} \wedge
XY\tilde{X}|U) + I(\tilde{X}\wedge XY|U) + I(X \wedge Y|U)-2 R_X,\label{ESX-def}\\
&E^Y_S(V_{UXY\tilde{Y}\hat{Y}}) \triangleq I(\hat{Y} \wedge
XY\tilde{Y}|U) + I(\tilde{Y}\wedge XY|U) + I(X \wedge Y|U)-2 R_Y,\label{ESY-def}\\
&E^{XY}_S(V_{UXY\tilde{X}\tilde{Y}\hat{X}\hat{Y}})
\triangleq \nonumber\\
&\quad I(\hat{X}\hat{Y} \wedge
XY\tilde{X}\tilde{Y}|U)+I(\tilde{X}\tilde{Y} \wedge XY|U)  + I(X
\wedge Y|U) +I(\tilde{X}\wedge \tilde{Y}|U) + I(\hat{X}\wedge
\hat{Y}|U)-2 R_X -2R_Y.\label{ESXY-def}
\end{align}
\end{subequations}
\end{definition}
\begin{definition}
For any given $R_X,R_Y \geq 0$, $P_{XYU} \in \P\left(\X \times \Y
\times \U \right)$, we define the sets of distributions $\V_X^r$,
$\V_Y^r$ and $\V_{XY}^r$ as follows:
\begin{subequations}
\begin{align}
& \V^r_X \triangleq \left\{V_{UXY\tilde{X}Z} :\alpha(V_{UXYZ}) \geq
\alpha(V_{U\tilde{X}YZ}), V_{UX}=V_{U\tilde{X}}= P_{UX}, V_{UY}=
P_{UY}  \right\},\label{VXr-def}\\
& \V^r_Y \triangleq \left\{V_{UXY\tilde{Y}Z} :\alpha(V_{UXYZ}) \geq
\alpha(V_{UX\tilde{Y}Z}), V_{UX}= P_{UX},
V_{UY}=V_{U\tilde{Y}}=P_{UY}  \right\},\label{VYr-def}\\
& \V^r_{XY} \triangleq \left\{V_{UXY\tilde{X}\tilde{Y}Z}
:\alpha(V_{UXYZ}) \geq \alpha(V_{U\tilde{X}\tilde{Y}Z}),
V_{UX}=V_{U\tilde{X}}= P_{UX}, V_{UY}=V_{U\tilde{Y}}=P_{UY}
\right\}.\label{VXYr-def}
\end{align}
\end{subequations} Moreover, $\V^{r,L}_X$, $\V^{r,L}_Y$ and
$\V^{r,L}_{XY}$ are sets of distributions and defined as
\begin{subequations}
\begin{align}
& \V^{r,L}_X \triangleq \left\{V_{UXY\tilde{X}Z} \in \V^{r}_X
: I(\tilde{X} \wedge XYZ|U) \geq R_X  \right\},\label{VXUr-def}\\
& \V^{r,L}_Y \triangleq \left\{V_{UXY\tilde{Y}Z} \in \V^{r}_Y
:I(\tilde{Y} \wedge XYZ|U) \geq R_Y   \right\},\label{VYUr-def}\\
& \V^{r,L}_{XY} \triangleq \left\{V_{UXY\tilde{X}\tilde{Y}Z} \in
\V^{r}_{XY} :I(\tilde{X}\tilde{Y} \wedge XYZ|U) + I(\tilde{X} \wedge
\tilde{Y}) \geq R_X + R_Y \right\}.\label{VXYUr-def}
\end{align}
\end{subequations}
\end{definition}
\begin{definition}
For any given $R_X,R_Y \geq 0$, $P_{XYU} \in \P\left(\X \times \Y
\times \U \right)$, we define the sets of distributions $\V^T_X$,
$\V^T_Y$, and $\V^T_{XY}$ as follows
\begin{subequations}
\begin{align}
&\V^T_X \triangleq \left\{
\begin{array}{lc}
V_{UXY\tilde{X}}: & V_{XU}=V_{\tilde{X}U}=P_{XU}, V_{YU}=P_{YU} \\
 & F_U(V_{UXY}) , F_U(V_{U\tilde{X}Y}) \leq R_X +R_Y\\
 & F_{X}(V_{UXY\tilde{X}}) \leq R_X+R_Y\\
 & \alpha(V_{UXYZ}) \geq \alpha(V_{U\tilde{X}YZ})
\end{array}\right\}\label{VT-X-def}\\
&\V^T_Y \triangleq \left\{
\begin{array}{lc}
V_{UXY\tilde{Y}}: & V_{XU}=P_{XU}, V_{YU}=V_{\tilde{Y}U}=P_{YU} \\
 & F_U(V_{UXY}) , F_U(V_{UX\tilde{Y}}) \leq R_X +R_Y\\
 & F_{Y}(V_{UXY\tilde{Y}}) \leq R_X+R_Y\\
 & \alpha(V_{UXYZ}) \geq
\alpha(V_{UX\tilde{Y}Z})
\end{array}\right\}\label{VT-Y-def}\\
\V^T_{XY} \triangleq &\left\{
\begin{array}{lc}
V_{UXY\tilde{X}\tilde{Y}}: & V_{XU}=V_{\tilde{X}U}=P_{XU}, V_{YU}=V_{\tilde{Y}U}=P_{YU} \\
  & F_U(V_{UXY}) , F_U(V_{U\tilde{X}Y}),F_U(V_{UX\tilde{Y}}) , F_U(V_{U\tilde{X}\tilde{Y}}) \leq R_X +R_Y\\
  & F_{X}(V_{UXY\tilde{X}}), F_{X}(V_{UX\tilde{Y}\tilde{X}}) \leq R_X+R_Y\\
  & F_{Y}(V_{UXY\tilde{Y}}), F_{Y}(V_{U\tilde{X}Y\tilde{Y}}) \leq R_X+R_Y\\
  & F_{XY}(V_{UXY\tilde{X}\tilde{Y}}), F_{XY}(V_{U\tilde{X}YX\tilde{Y}}) \leq
  R_X+R_Y\\
  & \alpha(V_{UXYZ})
\geq \alpha(V_{U\tilde{X}\tilde{Y}Z})
\end{array}\right\}\label{VT-XY-def}
\end{align}
\end{subequations}
Moreover, $\V^{T,L}_X$, $\V^{T,L}_Y$, and $\V^{T,L}_{XY}$ are sets
of distributions and defined as
\begin{subequations}
\begin{align}
& \V^{T,L}_X \triangleq \left\{V_{UXY\tilde{X}Z} \in \V^{T}_X
: I(\tilde{X} \wedge XYZ|U) \geq R_X  \right\},\label{VXUT-def}\\
& \V^{T,L}_Y \triangleq \left\{V_{UXY\tilde{Y}Z} \in \V^{T}_Y
:I(\tilde{Y} \wedge XYZ|U) \geq R_Y   \right\},\label{VYUT-def}\\
& \V^{T,L}_{XY} \triangleq \left\{V_{UXY\tilde{X}\tilde{Y}Z} \in
\V^{T}_{XY} :I(\tilde{X}\tilde{Y} \wedge XYZ|U) + I(\tilde{X} \wedge
\tilde{Y}) \geq R_X + R_Y \right\}.\label{VXYUT-def}
\end{align}
\end{subequations}
\end{definition}
\begin{definition}
For any given $R_X,R_Y \geq 0$, $P_{XYU} \in \P\left(\X \times \Y
\times \U \right)$, we define the sets of distributions $\V^{ex}_X$,
$\V^{ex}_Y$, and $\V^{ex}_{XY}$ as follows
\begin{subequations}
\begin{align}
&\V^{ex}_X \triangleq \left\{
\begin{array}{lc}
V_{UXY\tilde{X}}: & V_{XU}=V_{\tilde{X}U}=P_{XU}, V_{YU}=P_{YU} \\
 & F_U(V_{UXY}) , F_U(V_{U\tilde{X}Y}) \leq \min\{R_X,R_Y\}\\
 & F_{X}(V_{UXY\tilde{X}}) \leq \min\{R_X,R_Y\}\\
 & \alpha(V_{UXYZ}) \geq \alpha(V_{U\tilde{X}YZ})
\end{array}\right\}\label{VEX-X-def}\\
&\V^{ex}_Y \triangleq \left\{
\begin{array}{lc}
V_{UXY\tilde{Y}}: & V_{XU}=P_{XU}, V_{YU}=V_{\tilde{Y}U}=P_{YU} \\
 & F_U(V_{UXY}) , F_U(V_{UX\tilde{Y}}) \leq \min\{R_X,R_Y\}\\
 & F_{Y}(V_{UXY\tilde{Y}}) \leq \min\{R_X,R_Y\}\\
 & \alpha(V_{UXYZ}) \geq
\alpha(V_{UX\tilde{Y}Z})
\end{array}\right\}\label{VEX-Y-def}\\
\V^{ex}_{XY} \triangleq &\left\{
\begin{array}{lc}
V_{UXY\tilde{X}\tilde{Y}}: & V_{XU}=V_{\tilde{X}U}=P_{XU}, V_{YU}=V_{\tilde{Y}U}=P_{YU} \\
  & F_U(V_{UXY}) , F_U(V_{U\tilde{X}Y}),F_U(V_{UX\tilde{Y}}) , F_U(V_{U\tilde{X}\tilde{Y}}) \leq \min\{R_X,R_Y\}\\
  & F_{X}(V_{UXY\tilde{X}}), F_{X}(V_{UX\tilde{Y}\tilde{X}}) \leq \min\{R_X,R_Y\}\\
  & F_{Y}(V_{UXY\tilde{Y}}), F_{Y}(V_{U\tilde{X}Y\tilde{Y}}) \leq \min\{R_X,R_Y\}\\
  & F_{XY}(V_{UXY\tilde{X}\tilde{Y}}), F_{XY}(V_{U\tilde{X}YX\tilde{Y}}) \leq
  \min\{R_X,R_Y\}\\
  & \alpha(V_{UXYZ})
\geq \alpha(V_{U\tilde{X}\tilde{Y}Z})
\end{array}\right\}\label{VEX-XY-def}
\end{align}
\end{subequations}
\end{definition}
\begin{definition}
For any given $R_X,R_Y \geq 0$, $P_{XYU} \in \P\left(\X \times \Y
\times \U \right)$, and $V_{UXY\tilde{X}\tilde{Y}} \in \P\left(\U
\times (\X \times \Y)^2 \right)$, we define the following quantities
\begin{subequations}
\begin{align}
&E_{X}(R_X,R_Y,W,P_{XYU},V_{UXY\tilde{X}}) \triangleq D(V_{Z|XYU} ||
W|V_{XYU}) + I_V(X \wedge Y |U) + |I(\tilde{X} \wedge XYZ |U)
- R_X|^+,\label{EX-def}\\
&E_{Y}(R_X,R_Y,W,P_{XYU},V_{UXY\tilde{Y}}) \triangleq D(V_{Z|XYU} ||
W|V_{XYU}) + I_V(X \wedge Y |U) + |I(\tilde{Y} \wedge XYZ |U)
- R_Y|^+,\\
&E_{XY}(R_X,R_Y,W,P_{XYU},V_{UXY\tilde{X}\tilde{Y}})
\triangleq\nonumber\\ &\quad\quad\quad\quad D(V_{Z|XYU} || W
|V_{XYU}) + I_V(X \wedge Y |U) + |I(\tilde{X}\tilde{Y} \wedge XYZ
|U) + I_V(\tilde{X} \wedge \tilde{Y}| U)- R_X -
R_Y|^+.\label{EXY-def}
\end{align}
\end{subequations}
Moreover, we define
\begin{subequations}
\begin{align}
&E^L_{X}(R_X,R_Y,W,P_{XYU},V_{UXY\tilde{X}}) \triangleq D(V_{Z|XYU}
|| W|V_{XYU}) + I_V(X \wedge Y |U) + I(\tilde{X} \wedge XYZ |U)
- R_X,\label{EXL-def}\\
&E^L_{Y}(R_X,R_Y,W,P_{XYU},V_{UXY\tilde{Y}}) \triangleq D(V_{Z|XYU}
|| W|V_{XYU}) + I_V(X \wedge Y |U) + I(\tilde{Y} \wedge XYZ |U)
- R_Y,\\
&E^L_{XY}(R_X,R_Y,W,P_{XYU},V_{UXY\tilde{X}\tilde{Y}})
\triangleq\nonumber\\ &\quad\quad\quad\quad D(V_{Z|XYU} || W
|V_{XYU}) + I_V(X \wedge Y |U) + I(\tilde{X}\tilde{Y} \wedge XYZ |U)
+ I_V(\tilde{X} \wedge \tilde{Y}| U)- R_X - R_Y,\label{EXY-def}
\end{align}
\end{subequations}
and,
\begin{subequations}
\begin{align}
&E^\alpha_{\beta}(R_X,R_Y,W,P_{XYU},\V_\beta^\alpha) \triangleq \min_{\substack{V_{UXY\tilde{\beta}Z}\in \V^\alpha_\beta}} E_{\beta}(R_X,R_Y,W,P_{XYU},V_{UXY\tilde{\beta}}),\label{E-def}\\
&E^{\alpha,L}_{\beta}(R_X,R_Y,W,P_{XYU},\V_\beta^\alpha) \triangleq
\min_{\substack{V_{UXY\tilde{\beta}Z}\in \V^{\alpha,L}_\beta}}
E^L_{\beta}(R_X,R_Y,W,P_{XYU},V_{UXY\tilde{\beta}}),\label{EL-def}
\end{align}
\end{subequations}
for $\alpha \in \{r, T, ex\}$, and $\beta \in \{ X, Y,XY\}$.
\end{definition}

\subsection{Packing Lemmas}

As we did in the point-to-point case, here we perform random coding
and derive bounds on the packing functions. The results will be
stated as three lemmas, one for the average and one for the typical
performance of the ensemble, and finally one for the expurgated
ensemble. These results will be used in conjunction with the
relation between the packing functions and the probability of error
established in Section~\ref{MAC-relation} to obtain the bounds on
the error exponents.

\begin{lemma}\label{packing1}
Fix a finite set $\U$, $P_{XYU} \in \P_n( \X \times \Y \times \U)$
such that $X-U-Y$, $R_X \geq 0$, $R_Y \geq 0$ , $\delta > 0$,
$2^{n(R_X-\delta)} \leq M_X \leq 2^{nR_X}$, $ 2^{n(R_Y-\delta)} \leq
M_Y \leq 2^{nR_Y}$, and $\mathbf{u} \in T_{P_U}$. Let $X^{M_X}
\triangleq \{X_1,X_2,...,X_{M_X}\} $ and $Y^{M_Y} \triangleq
\{Y_1,Y_2,...,Y_{M_Y}\}$ are independent, and $X_i$s and $Y_j$s are
uniformly distributed over $T_{P_{X|U}}(\mathbf{u})$ and
$T_{P_{Y|U}}(\mathbf{u})$ respectively. For every joint type
$V_{UXY\tilde{X}\tilde{Y}} \in \P_n(\U \times (\X \times \Y)^2)$,
the expectation of the packing functions over the random code
$X^{M_X} \times Y^{M_Y}$ are bounded by
\begin{subequations}\label{77}
\begin{align}
&\;\;\;\;\;\;\;\;\;\;\;2^{-n[F_U(V_{UXY})+ \delta]} \leq
\mathbb{E}\Big[N_{U} (X^{M_X} \times Y^{M_Y},V_{UXY})\Big] \leq
2^{-n[F_U(V_{UXY})- 2 \delta]}, \\
&\;\;\;\;\;\;2^{-n[F_X(V_{UXY\tilde{X}})+ 3 \delta]} \leq
\mathbb{E}\Big[N_{X} (X^{M_X} \times Y^{M_Y},V_{UXY\tilde{X}})\Big]
\leq
2^{-n[F_X(V_{UXY\tilde{X}})- 4 \delta]},\\
&\;\;\;\;\;\;2^{-n[F_Y(V_{UXY\tilde{Y}})+ 3 \delta]} \leq
\mathbb{E}\Big[N_{Y} (X^{M_X} \times Y^{M_Y},V_{UXY\tilde{Y}})\Big]
\leq
2^{-n[F_Y(V_{UXY\tilde{Y}})-4 \delta]},\\
&2^{-n[F_{XY}(V_{UXY\tilde{X}\tilde{Y}}) + 4 \delta]} \leq
\mathbb{E}\Big[N_{XY} (X^{M_X} \times
Y^{M_Y},V_{UXY\tilde{X}\tilde{Y}})\Big] \leq
2^{-n[F_{XY}(V_{UXY\tilde{X}\tilde{Y}})- 4 \delta]},
\end{align}
\end{subequations}
whenever $n \geq n_0(|\U|,|\X|,|\Y|,\delta)$.
Moreover, for any $V_{UXY\tilde{X}\tilde{Y}\hat{X}\hat{Y}} \in
\P_n(\U \times (\X \times \Y)^3)$
\begin{subequations}\label{78}
\begin{align}
\mathbb{E}\Big[\Lambda_X (X^{M_X} \times
Y^{M_Y},V_{UXY\tilde{X}\hat{X}})\Big] &\leq
2^{-n\left(E^{X}_S(V_{UXY\tilde{X}\hat{X}}) - 4
\delta\right)},\label{Random-Lamx}\\
\mathbb{E}\Big[\Lambda_Y (X^{M_X} \times
Y^{M_Y},V_{UXY\tilde{Y}\hat{Y}})\Big] &\leq
2^{-n\left(E^{Y}_S(V_{UXY\tilde{Y}\hat{Y}}) - 4
\delta\right)},\label{random-Lamy}\\
\mathbb{E}\Big[\Lambda_{XY} (X^{M_X} \times
Y^{M_Y},V_{UXY\tilde{X}\tilde{Y}\hat{X}\hat{Y}})\Big] &\leq
2^{-n\left(E^{XY}_S(V_{UXY\tilde{X}\tilde{Y}\hat{X}\hat{Y}}) - 6
\delta \right) },\label{Random-Lamxy}
\end{align}
\end{subequations}
whenever $n \geq n_0(|\U|,|\X|,|\Y|,\delta)$.
\end{lemma}
\begin{proof}
The proof is provided in the Appendix.
\end{proof}

\begin{lemma}\label{packing3}
Fix a finite set $\U$, $P_{XYU} \in \P_n(\X \times \Y \times \U)$
such that $X-U-Y$, $R_X \geq 0$, $R_Y \geq 0$ , $\delta > 0$,
$2^{n(R_X-\delta)} \leq M_X \leq 2^{nR_X}$, $ 2^{n(R_Y-\delta)} \leq
M_Y \leq 2^{nR_Y}$, and $\mathbf{u} \in T_{P_U}$. Almost every
multi-user code $C=C_X \times C_Y$,
$C_X=\{\mathbf{x}_1,\mathbf{x}_2,...,\mathbf{x}_{M_X}\}\subset
T_{P_{X|U}}(\mathbf{u})$ and
$C_Y=\{\mathbf{y}_1,\mathbf{y}_2,...,\mathbf{y}_{M_Y}\}\subset
T_{P_{Y|U}}(\mathbf{u})$, in the constant composition code ensemble,
$\C$, satisfies the following inequalities:
\begin{subequations}
\begin{align}
2^{-n[F_U(V_{UXY})+ 3 \delta]} &\leq N_{U} (C,V_{UXY}) \leq
2^{-n[F_U(V_{UXY})- 3
\delta]},\label{Fu}\\
2^{-n[F_X(V_{UXY\tilde{X}})+ 5 \delta]} &\leq N_{X}
(C,V_{UXY\tilde{X}}) \leq 2^{-n[F_X(V_{UXY\tilde{X}})- 5
\delta]},\label{Fx}\\
2^{-n[F_Y(V_{UXY\tilde{Y}})+ 5\delta]} &\leq N_{Y}
(C,V_{UXY\tilde{Y}}) \leq 2^{-n[F_Y(V_{UXY\tilde{X}\tilde{Y}})- 5
\delta]},\label{Fy}\\
2^{-n[F_{XY}(V_{UXY\tilde{X}\tilde{Y}})+ 5 \delta]} &\leq N_{XY}
(C,V_{UXY\tilde{X}\tilde{Y}}) \leq
2^{-n[F_{XY}(V_{UXY\tilde{X}\tilde{Y}})- 5 \delta]},\label{Fxy}
\end{align}
\end{subequations}
for all $V_{UXY\tilde{X}\tilde{Y}} \in \P_n(\U \times (\X \times
\Y)^2)$, and
\begin{subequations}
\begin{align}
\Lambda_X (C,V_{UXY\tilde{X}\hat{X}}) &\leq
2^{-n\left(E^{X}_S(V_{UXY\tilde{X}\hat{X}}) - 5
\delta\right)},\label{Lamx}\\
\Lambda_Y (C,V_{UXY\tilde{Y}\hat{Y}}) &\leq
2^{-n\left(E^{Y}_S(V_{UXY\tilde{Y}\hat{Y}}) - 5
\delta\right)},\label{Lamy}\\
\Lambda_{XY} (C,V_{UXY\tilde{X}\tilde{Y}\hat{X}\hat{Y}}) &\leq
2^{-n\left(E^{XY}_S(V_{UXY\tilde{X}\hat{X}}) - 7 \delta \right)
}.\label{Lamxy}
\end{align}
\end{subequations}
for all $V_{UXY\tilde{X}\tilde{Y}\hat{X}\hat{Y}} \in \P_n\left(\U
\times (\X \times \Y)^3\right)$, whenever $n \geq
n_0(|\U|,|\X|,|\Y|,\delta)$.
\end{lemma}
\begin{proof}
The proof is provided in the Appendix.
\end{proof}

%
%
\begin{lemma}\label{packing2}
For every finite set $\U$, $P_{XYU} \in \P_n(\X \times \Y \times
\U)$ such that $X-U-Y$, $R_X \geq 0$, $R_Y \geq 0$ , $\delta > 0$,
and $\mathbf{u} \in T_{P_U}$, there exist a multi-user code
$C^*=C^*_X \times C^*_Y$,
$C^*_X=\{\mathbf{x}_1,\mathbf{x}_2,...,\mathbf{x}_{M^*_X}\} \subset
T_{P_{X|U}}(\mathbf{u})$ and
$C^*_Y=\{\mathbf{y}_1,\mathbf{y}_2,...,\mathbf{y}_{M^*_Y}\} \subset
T_{P_{Y|U}}(\mathbf{u})$ with
$M^*_X \geq \frac{2^{n(R_X-\delta)}}{2}$, $M^*_Y \geq
\frac{2^{n(R_Y- \delta)}}{2}$, such that for every joint type
$V_{UXY\tilde{X}\tilde{Y}} \in \P_n(\U \times (\X \times \Y)^2)$,
\begin{subequations}
\begin{align}
N_{U} (C^*,V_{UXY}) &\leq
2^{-n[F_U(V_{UXY})- 6 \delta]}\;\;\;\;\;\;\;\;\;\;\;\;\;\label{expur-lem1}\\
N_{X} (C^*,V_{UXY\tilde{X}}) &\leq 2^{-n[F_{X}(V_{UXY\tilde{X}})- 6
\delta]}\;\;\;\;\;\;\label{expur-lem2}\\
N_{Y} (C^*,V_{UXY\tilde{Y}}) &\leq 2^{-n[F_{Y}(V_{UXY\tilde{Y}})- 6
\delta]}\;\;\;\;\;\;\label{expur-lem3}\\
N_{XY} (C^*,V_{UXY\tilde{X}\tilde{Y}}) &\leq
2^{-n[F_{XY}(V_{UXY\tilde{X}\tilde{Y}})- 6
\delta]}\label{expur-lem4}
\end{align}
\end{subequations}
and for any $1 \leq i \leq M^*_X$, and any $1 \leq j \leq M^*_Y$,
\begin{subequations}
\begin{eqnarray}
 1_{T_{V_{UXY}}}(\mathbf{u},
\mathbf{x}_i, \mathbf{y}_j)
\leq 2^{-n[F_U(V_{UXY}) -\min\{R_X,R_Y\} - 6 \delta]} \;\; \label{randomnew1}\;\;\;\;\;\;\;\;\;\;\;\;\;\;\;\;\;\;\;\\
\sum_{k \neq i} 1_{T_{V_{UXY\tilde{X}}}}(\mathbf{u}, \mathbf{x}_i,
\mathbf{y}_j, \mathbf{x}_k) \leq 2^{-n[F_{X}(V_{UXY\tilde{X}})-
\min\{R_X,R_Y\} - 6
\delta]}\label{randomnew2}\;\;\;\;\;\;\;\;\;\;\;\\
\sum_{l \neq j}  1_{T_{V_{UXY\tilde{Y}}}}(\mathbf{u}, \mathbf{x}_i,
\mathbf{y}_j, \mathbf{y}_l) \leq 2^{-n[F_{Y}(V_{UXY\tilde{Y}})-
\min\{R_X,R_Y\}-6
\delta]}\label{randomnew3}\;\;\;\;\;\;\;\;\;\;\;\\
\sum_{k \neq i} \sum_{l \neq j}
1_{T_{V_{UXY\tilde{X}\tilde{Y}}}}(\mathbf{u}, \mathbf{x}_i,
\mathbf{y}_j, \mathbf{x}_k, \mathbf{y}_l) \leq
2^{-n[F_{XY}(V_{UXY\tilde{X}\tilde{Y}})- \min\{R_X,R_Y\} - 6
\delta]}\label{randomnew4},
\end{eqnarray}
\end{subequations}
whenever
\begin{equation*}
n \geq n_0(|\U|,|\X|,|\Y|,\delta).
\end{equation*}
\end{lemma}
\begin{proof}
The proof is provided in the Appendix.
\end{proof}

As it is shown in the Appendix, the above property is derived by the
method of expurgation.
Unlike the point-to-point case, expurgation in the MAC is not a
trivial procedure. To see that, observe that expurgating bad pairs
of codewords results in a code with correlated messages, which is
hard to analyze.
Instead, what we do is a sort of ``partial'' expurgation. Roughly
speaking, we start with a code whose existence is proved in
Lemma~\ref{packing1} and eliminate some of the bad codewords from
the code with the larger rate (as opposed to codeword pairs).
%
%
By doing that, all messages in the new code are independent, and
such a code is easier to analyze.
%

\subsection{Error exponent bounds}

We can now proceed in a fashion that is similar to the
point-to-point case and derive a series of exponential bounds based
on Lemmas~\ref{packing1}, \ref{packing3}, and \ref{packing2}. In the
following, we present three lower bounds, the random coding, the
typical random coding, and the expurgated bounds. As in the case of
point-to-point channels, here too, all the lower bounds are
expressed in terms of the optimization of a single objective
function under different constraint sets.

\begin{theorem}\label{randomcodingthm}
Fix a finite set $\U$, $P_{XYU} \in \P_n(\X \times \Y \times \U)$
such that $X-U-Y$, $R_X \geq 0$, $R_Y \geq 0$ , $\delta > 0$,
$2^{n(R_X-\delta)} \leq M_X \leq 2^{nR_X}$, $ 2^{n(R_Y-\delta)} \leq
M_Y \leq 2^{nR_Y}$, and $\mathbf{u} \in T_{P_U}$. Consider the
ensemble, $\C$, of multi-user codes consisting of all pair of
codebooks $(C_X,C_Y)$, where
$C_X=\{\mathbf{x}_1,\mathbf{x}_2,...,\mathbf{x}_{M_X}\} \subset
T_{P_{X|U}}(\mathbf{u})$ and
$C_Y=\{\mathbf{y}_1,\mathbf{y}_2,...,\mathbf{y}_{M_Y}\} \subset
T_{P_{Y|U}}(\mathbf{u})$. The expectation of the average probability
of error over $\C$ is bounded by
\begin{align}
2^{-n[E_{rL}(R_X,R_Y,W,P_{XYU})+ 8 \delta]} \leq \bar{P_e} \leq
2^{-n[E_{r}(R_X,R_Y,W,P_{XYU}) - 6 \delta  ]}
\end{align}
whenever $n \geq n_1(|\Z|,|\X|,|\Y|,|\U|, \delta)$, where
\begin{eqnarray}
&E_{r}(R_X,R_Y,W,P_{XYU}) \triangleq
\min_{\substack{\substack{\beta=X,Y,XY}}}
E^r_{\beta}(R_X,R_Y,W,P_{UXY},\V_\beta^r),\\
&E_{rL}(R_X,R_Y,W,P_{XYU}) \triangleq \min_{\substack{\beta=X,Y,XY}}
E^{r,L}_{\beta}(R_X,R_Y,W,P_{UXY},\V_\beta^{r,L}).
\end{eqnarray}

\end{theorem}
\begin{proof}
The proof is provided in the Appendix.
\end{proof}
\begin{corollary}
In the low rate regime,
\begin{align}
E_{rL}(R_X,R_Y,W,P_{XYU}) = E_r(R_X,R_Y,W,P_{XYU}).
\end{align}
We call this rate region as the critical region for $W$.
\end{corollary}
\begin{proof}
The proof is similar to the proof of corollary~\ref{cor:erl_eq_er}
and is omitted.
\end{proof}

\begin{theorem}\label{Typrandomcodingthm}
Fix a finite set $\U$, $P_{XYU} \in \P_n(\X \times \Y \times \U)$
such that $X-U-Y$, $R_X \geq 0$, $R_Y \geq 0$ , $\delta > 0$,
$2^{n(R_X-\delta)} \leq M_X \leq 2^{nR_X}$, $ 2^{n(R_Y-\delta)} \leq
M_Y \leq 2^{nR_Y}$, and $\mathbf{u} \in T_{P_U}$. The average
probability of error for almost all multi-user codes $C=C_X \times
C_Y$, $C_X=\{\mathbf{x}_1,\mathbf{x}_2,...,\mathbf{x}_{M_X}\}
\subset T_{P_{X|U}}(\mathbf{u})$ and
$C_Y=\{\mathbf{y}_1,\mathbf{y}_2,...,\mathbf{y}_{M_Y}\} \subset
T_{P_{Y|U}}(\mathbf{u})$, in ensemble $\C$, satisfies the following
inequalities
\begin{align}
2^{-n[E_{TL}(R_X,R_Y,W,P_{XYU})+ 7 \delta]} \leq e(C,W) \leq
2^{-n[E_{T}(R_X,R_Y,W,P_{XYU}) - 6 \delta  ]}
\end{align}
whenever $n \geq n_1(|\Z|,|\X|,|\Y|,|\U|, \delta)$, where
\begin{eqnarray}
&E_{T}(R_X,R_Y,W,P_{XYU}) \triangleq \min_{\substack{\beta=X,Y,XY}}
E^T_{\beta}(R_X,R_Y,W,P_{UXY},\V_\beta^T)\\
&E_{TL}(R_X,R_Y,W,P_{XYU}) \triangleq \min_{\substack{\beta=X,Y,XY}}
E^{T,L}_{\beta}(R_X,R_Y,W,P_{UXY},\V_\beta^{T,L}).
\end{eqnarray}

\end{theorem}
\begin{proof}
The proof is provided in the Appendix.
\end{proof}

\begin{corollary}\label{MAC-typical-Cor1}
 For every finite set $\U$, $\P_{XYU} \in \P_n(\X \times \Y
\times \U)$ such that $X-U-Y$ , $R_X \geq 0$, $R_Y \geq 0$,
\begin{align}
E_T(R_X,R_Y,P_{XYU},W) \leq E_{av}^T (R_X,R_Y) \leq
E_{TL}(R_X,R_Y,P_{XYU},W).
\end{align}
\end{corollary}
\begin{proof}
The proof is very similar to the proof of
Corollary~\ref{P2P-typical-Cor1}.
\end{proof}

\begin{corollary}
In the low rate regime,
\begin{align}
E_{TL}(R_X,R_Y,P_{XYU},W) = E_T(R_X,R_Y,P_{XYU},W).
\end{align}
\end{corollary}
\begin{proof}
The proof is similar to the proof of Corollary~\ref{cor:erl_eq_er}
and is omitted.
\end{proof}

\begin{theorem}\label{expurgatedthm}
 For every finite set $\U$, $\P_{XYU} \in \P_n(\X \times \Y
\times \U)$ such that $X-U-Y$ , $R_X \geq 0$, $R_Y \geq 0$, $\delta
> 0$, and $\mathbf{u} \in T_{P_U}$, there exists a multi-user code
\begin{equation}
C = \{ (\mathbf{x}_i, \mathbf{y}_j,D_{ij}) : i=1,...M^*_X,
j=1,...M^*_Y \}
\end{equation}
with $\mathbf{x}_i \in T_{P_{X|U}}(\mathbf{u})$, $\mathbf{y}_j \in
T_{P_{Y|U}}(\mathbf{u})$ for all $i$ and $j$, $M^*_X \geq
\frac{2^{n(R_X-\delta)}}{2}$, and $M^*_Y \geq
\frac{2^{n(R_Y-\delta)}}{2}$, such that for every MAC $W:\X \times
\Y \rightarrow \Z$
\begin{equation}
e(C,W) \leq 2^{-n[E_{ex}(R_X,R_Y,W,P_{XYU})-
5\delta]}\label{expur-bnd}
\end{equation}
whenever $n \geq n_1(|\Z|,|\X|,|\Y|,|\U|, \delta)$, where
%
\begin{eqnarray}
&E_{ex}(R_X,R_Y,W,P_{XYU}) \triangleq \min_{\substack{\beta=X,Y,XY}}
E^{ex}_{\beta}(R_X,R_Y,W,P_{UXY},\V_\beta^{ex}).
\end{eqnarray}
\end{theorem}
\begin{proof}
The proof is provided in the Appendix.
\end{proof}

This exponential error bound can be universally obtained for all
MAC's with given input and output alphabets, since the choice of the
codewords does not depend on the channel.

In the following, we show that the bound in
Theorem~\ref{randomcodingthm} is at least as good as the best known
random coding bound, found in~\cite{Liu-RandomCoding}. For this
purpose, let us use the minimum equivocation decoding rule.
\begin{definition}
Given $\mathbf{u}$, for a multiuser code
\begin{equation*}
C = \{ (\mathbf{x}_i, \mathbf{y}_j,D_{ij}) : i=1,...M_X, j=1,...M_Y
\}
\end{equation*}
we say that the $D_{ij}$ are minimum equivocation decoding sets for
$\mathbf{u}$, if $\mathbf{z} \in D_{ij}$ implies
\begin{equation*}
H(\mathbf{x}_i \mathbf{y}_j|\mathbf{z} \mathbf{u}) = \min_{k,l}
H(\mathbf{x}_k \mathbf{y}_l|\mathbf{z} \mathbf{u}).
\end{equation*}
It can be easily observed that these sets are equivalent to
$\alpha$-decoding sets, where $\alpha(\mathbf{u}, \mathbf{x},
\mathbf{y}, \mathbf{z})$ is defined as
\begin{equation}
\alpha(V_{UXYZ}) \triangleq H_V(XY|ZU).
\end{equation}
Here, $V_{UXYZ}$ is the joint empirical distribution of
$(\mathbf{u}, \mathbf{x}, \mathbf{y}, \mathbf{z})$.
\end{definition}
\begin{theorem}\label{finalthm}
 For every finite set
$\U$, $\P_{XYU} \in \P(\X \times \Y \times \U)$ , $R_X \geq 0$, $R_Y
\geq 0$, and $W:\X \times \Y \rightarrow \Z$, and an appropriate
$\alpha$-decoder (minimum equivocation),
\begin{subequations}
\begin{eqnarray}
E^{r}_{\beta}(R_X,R_Y,W,P_{XYU}) \geq
E^{Liu}_{r\beta}(R_X,R_Y,W,P_{XYU})\;\;\;\;\;\;\;\;\;\;\;\;\;\;\;\;\;\;\;\;
\beta=X,Y,XY,\\
E^{T}_{\beta}(R_X,R_Y,W,P_{XYU}) \geq
E^{Liu}_{r\beta}(R_X,R_Y,W,P_{XYU})\;\;\;\;\;\;\;\;\;\;\;\;\;\;\;\;\;\;\;\;
\beta=X,Y,XY,\\
E^{ex}_{\beta}(R_X,R_Y,W,P_{XYU}) \geq
E^{Liu}_{r\beta}(R_X,R_Y,W,P_{XYU})\;\;\;\;\;\;\;\;\;\;\;\;\;\;\;\;\;\;\;\;
\beta=X,Y,XY.\label{betterLiu}
\end{eqnarray}
\end{subequations}
Hence
\begin{subequations}
\begin{eqnarray}
E_{r}(R_X,R_Y,W,P_{XYU}) \geq
E^{Liu}_r(R_X,R_Y,W,P_{XYU}),\label{r-Inequaltiy}\\
E_{T}(R_X,R_Y,W,P_{XYU}) \geq
E^{Liu}_r(R_X,R_Y,W,P_{XYU}),\label{t-Inequaltiy}\\
E_{ex}(R_X,R_Y,W,P_{XYU}) \geq
E^{Liu}_r(R_X,R_Y,W,P_{XYU}),\label{ex-Inequaltiy}
\end{eqnarray}
\end{subequations}
for all $P_{XYU} \in \P(\X \times \Y \times \U)$ satisfying $X-U-Y$.
Here, $E^{Liu}_{r}$ is the random coding exponent
of~\cite{Liu-RandomCoding}. $E^{Liu}_{r\beta}$ are also defined
in~\cite{Liu-RandomCoding} for $\beta=X,Y,XY$.
\end{theorem}
\begin{proof}
The proof is provided in the Appendix.
\end{proof}
We expect our typical random coding and expurgated bound to be
strictly better than the one in~\cite{Liu-RandomCoding} at low
rates. This is so, because all inequalities
in~\eqref{VT-X-def}-\eqref{VT-XY-def}
and~\eqref{VEX-X-def}-\eqref{VEX-XY-def} will be active at zero
rates, and thus (due to continuity) at sufficiently low rates.
Although we have not been able to prove this fact rigorously, in the
next section, we show that this is true by numerically evaluating
the expurgated bound for different rate pairs.

\section{Numerical result}
\label{sec:num} In this section, we calculate the exponent derived
in Theorem~\ref{expurgatedthm}
for a multiple-access channel very similar to the one used
in~\cite{Liu-RandomCoding}. This example shows that strict
inequality can hold in~\eqref{betterLiu}. Consider a discrete
memoryless MAC with $\X = \Y = \Z = \{ 0 , 1 \}$ and the transition
probability given in the following table.

\begin{center}
  \begin{tabular}{ l  c  c  || r  }
  \hline
    $x$ & $y$ & $z$ & $W(z|xy)$ \\ \hline
    0 & 0 & 0 & 0.99\;\;\;\\
    0 & 0 & 1 & 0.01\;\;\;\\
    0 & 1 & 0 & 0.01\;\;\;\\
    0 & 1 & 1 & 0.99\;\;\;\\
    1 & 0 & 0 & 0.01\;\;\;\\
    1 & 0 & 1 & 0.99\;\;\;\\
    1 & 1 & 0 & 0.50\;\;\;\\
    1 & 1 & 1 & 0.50\;\;\;\\ \hline
    \end{tabular}
\end{center}
First, we choose some time-sharing alphabet $\U$ of size $|\U|=4$.
Then some channel input distribution $P_UP_{X|U}P_{Y|U}$ is chosen
randomly. The following table gives numerical values of the random
coding exponent of~\cite{Liu-RandomCoding}, and the expurgated
exponent we have obtained for selected rate pairs.
\begin{center}
  \begin{tabular}{ l  c  || c   r  }
  \hline
    $R_X$ & $R_Y$ & $E_{ex}(R_X,R_Y,W,P_{UXY})$ & $E_{r}^{Liu}(R_X,R_Y,W,P_{UXY})$ \\ \hline
    0.01 & 0.01 & 0.2672 & 0.2330\;\;\;\\
    0.01 & 0.02 & 0.2671 & 0.2330\;\;\;\\
    0.01 & 0.03 & 0.2671 & 0.2330\;\;\;\\
    0.02 & 0.01 & 0.2458 & 0.2230\;\;\;\\
    0.02 & 0.02 & 0.2379 & 0.2230\;\;\;\\
    0.02 & 0.05 & 0.2379 & 0.2230\;\;\;\\
    0.03 & 0.01 & 0.2279 & 0.2130\;\;\;\\
    0.03 & 0.03 & 0.2183 & 0.2130\;\;\;\\
    0.04 & 0.01 & 0.2123 & 0.2030\;\;\;\\
    0.04 & 0.04 & 0.2040 & 0.2030\;\;\;\\
    0.05 & 0.05 & 0.1930 & 0.1930\;\;\;\\
    0.06 & 0.01 & 0.1856 & 0.1830\;\;\;\\
    0.06 & 0.06 & 0.1830 & 0.1830\;\;\;\\
    0.07 & 0.01 & 0.1740 & 0.1730\;\;\;\\
    0.07 & 0.07 & 0.1730 & 0.1730\;\;\;\\ \hline
  \end{tabular}
\end{center}
As we see in the table, in the low rate regime, we have strictly
better results in comparison with the results
of~\cite{Liu-RandomCoding}. For larger rate pairs, the inequalities
containing $\min\{R_X,R_Y\}$ will not be active anymore, thus, we
will end up with result similar to~\cite{Liu-RandomCoding}.


\section{Conclusions}
\label{sec:conclusion} We studied a unified framework to obtain all
known lower bounds (random coding, typical random coding and
expurgated bound) on the reliability function of a point-to-point
discrete memoryless channel. We showed that the typical random
coding bound is the typical performance of the constant composition
code ensemble. By using a similar idea with a two-user discrete
memoryless multiple-access channel, we derived three lower bounds on
the reliability function. The first one (random coding) is identical
to the best known lower bound on the reliability function of DM-MAC.
We also showed that the random coding bound is the average
performance of the constant composition code ensemble. The second
bound (typical random coding) is the typical performance of the
constant composition code ensemble. To derive the third bound
(expurgated), we eliminated some of the codewords from the codebook
with a larger rate. This is the first bound of its type that
explicitly uses the method of expurgation in a multi-user
transmission system. We showed that the exponent of the typical
random coding and expurgated bounds are greater than or equal to the
exponent of the known random coding bounds for all rate pairs. By
numerical evaluation of the random coding and the expurgated bounds
for a simple symmetric MAC, we showed that, at low rates, the
expurgated bound is strictly larger. All these bounds can be
universally obtained for all discrete memoryless MACs with given
input and output alphabets.

\input{appendix.tex}
\bibliographystyle{plain}
\bibliography{ali}

\end{document}

%% file: appendix.tex
\appendix


\subsection*{1. Point to Point Proofs} This section contains the proof
of all lemmas and theorems related to point to point result.
\begin{proof}\textbf{(Lemma~\ref{P2P-random-Packing})}
We use the method of random selection. Define $M$ such that
\begin{eqnarray*}
2^{n(R-\delta)} \leq M \leq 2^{nR}.
\end{eqnarray*}
In the following, we obtain the expectation of the packing function
over the constant composition code ensemble. The expectation of
$\pi(X^M,V_{X\tilde{X}})$ can be obtained as follows:
\begin{align}
\mathbb{E} \left(\pi(X^M,V_{X\tilde{X}})\right) &=  \frac{1}{M}
\sum_{i=1}^{M} \sum_{j \neq i} \mathbb{E}\left(1_{T_{V_{X
\mathcal{X}}}}(X_i,X_j)\right) =  \frac{1}{M} \sum_{i=1}^{M} \sum_{j
\neq i} \mathbb{P}\left(X_j
\in T_{V_{\tilde{X}|X}}(X_i) \right)\nonumber\\
&= (M-1) \mathbb{P}\left( X_2 \in T_{V_{\tilde{X}|X}}(X_1)\right)\
\leq 2^{n(R-I_V(X \wedge \tilde{X})+\delta)}.\label{p2p:
random-packing1}
\end{align}
Similarly, it can be shown that for sufficiently large $n$,
\begin{align}
\mathbb{E} \left(\pi(X^M,V_{X\tilde{X}})\right) \geq 2^{n(R-I_V(X
\wedge \tilde{X})-\delta)}.
\end{align}
The expectation of $\lambda $ over the ensemble can be written as
\begin{align}
\mathbb{E} \left(\lambda (X^M, V_{X\tilde{X}\hat{X}})\right) =
\frac{1}{M} \sum_{i=1}^{M}\sum_{j \neq i} \sum_{k \neq i,j}
\mathbb{P} \left((X_i,X_j, X_k) \in
T_{V_{X\tilde{X}\hat{X}}}\right).
\end{align}
Since
\begin{align}
\frac{2^{n[H(\tilde{X}\hat{X}|X)-\delta]}}{2^{nH(\tilde{X})}2^{nH(\hat{X})}}
\leq \mathbb{P} \left((X_i,X_j, X_k) \in
T_{V_{X\tilde{X}\hat{X}}}\right) \leq
\frac{2^{nH(\tilde{X}\hat{X}|X)}}{2^{n[H(\tilde{X})-\delta]}2^{n[H(\hat{X})-\delta]}},
\end{align}
it can be concluded that
\begin{align}
2^{n[E_S(R,V_{X\tilde{X}\hat{X}}) - 2 \delta]} \leq \mathbb{E}
\left(\lambda (X^M, V_{X\tilde{X}\hat{X}})\right) \leq
2^{n[E_S(R,V_{X\tilde{X}\hat{X}}) +
2\delta]},\label{p2p:upper-expectation-lambda}
\end{align}
where
\begin{align}
E_S(R,V_{X\tilde{X}\hat{X}}) \triangleq 2R - I(X \wedge \tilde{X}) -
I(\hat{X} \wedge \tilde{X}X).
\end{align}
By using~\eqref{p2p: random-packing1} and markov inequality, it can
be concluded that
\begin{align}
\mathbb{P} \left( \pi(X^M,V_{X\tilde{X}}) \geq 2^{n(R-I_V(X \wedge
\tilde{X})+2\delta)} \text{ for some } V_{X\tilde{X}}\right) \leq
\sum_{V_{X\tilde{X}}} \frac{\mathbb{E}
\left(\pi(X^{M},V_{X\tilde{X}})\right)}{2^{n(R-I_V(X \wedge
\tilde{X})+2\delta)}} \leq
2^{-n\frac{\delta}{2}},\label{Random-Packing-Proof-1}
\end{align}
therefore, there exists at least one code, $C^r$, with $M$ codewords
satisfying
\begin{align}
\pi(C^r,V_{X\tilde{X}}) \leq 2^{n(R-I_V(X \wedge
\tilde{X})+2\delta)}. \label{Random-Packing-Proof-2}
\end{align}
\end{proof}

\begin{proof}\textbf{(Lemma~\ref{lm:p2p_typical})}
To prove that a specific property holds for almost all codes, with
certain number of codewords, in the constant composition code
ensemble, we use a second-order argument method. We already have
obtained upper and lower bounds on the expectation of the desired
function over the entire ensemble. In the following, we derive an
upper bound on the variance of the packing function. Finally, by
using the Chebychev's inequality, we prove that the desired property
holds for almost all codes in the ensemble.

To find the variance of the packing function, let us define $U_{ij}
\triangleq 1_{T_{V_{X \tilde{X}}}}(X_i,X_j)$, and $Y_{ij} \triangleq
U_{ij} + U_{ji} $. We can rewrite $\pi(X^M,V_{X\tilde{X}})$ as
\begin{align}
\pi(X^M,V_{X\tilde{X}}) = \frac{1}{M} \sum_{i=1}^{M} \sum_{j \neq i}
U_{ij} = \frac{1}{M} \sum_{i=1}^{M} \sum_{j < i} \left( U_{ij} +
U_{ji}\right) = \frac{1}{M} \sum_{i=1}^{M} \sum_{j < i} Y_{ij}.
\end{align}
It is easy to check that $Y_{ij}$'s are identically distributed
pairwise independent random variables. Therefore, the variance of
$\pi(X^M,V_{X\tilde{X}})$ can be written as
\begin{align}
Var \left( \pi(X^M,V_{X\tilde{X}}) \right) = \frac{1}{M^2}
\sum_{i=1}^{M} \sum_{j < i} Var (Y_{ij} ) = \frac{1}{M^2} {M \choose
2} Var (Y_{21}).
\end{align}
To find the variance of $Y_{21}$, let us consider the following two
cases for $V_{X\tilde{X}}$:
\begin{itemize}
\item  $V_{X\tilde{X}}$ is a symmetric distribution. In this case $U_{12}=U_{21}$, therefore,

$$
Y_{21}= \left\{
\begin{array}{lc}
2 & \text{with probability   } p \leq  2^{-n[I(X \wedge \tilde{X})-\delta]} \\
0 & \;\;\;\;\;\;\text{with probability   } 1-p
\end{array},\right.
$$
and the variance is upper bounded by
\begin{align}
Var(Y_{21}) \leq E(Y_{21}^2) = 4 \times 2^{-n[I(X \wedge \tilde{X})-
\delta]} \label{var1}
\end{align}
\item  $V_{X\tilde{X}}$ is not a symmetric distribution. In this case, if $U_{ij} = 1\Rightarrow U_{ji}=0$. Therefore,
\begin{align}
\mathbb{P} \left(Y_{12} =1\right) = \mathbb{P} \left( U_{12} =1
\text{ or } U_{21} =1\right) = \mathbb{P} \left( U_{12} =1\right) +
\mathbb{P} \left( U_{21} =1\right) \leq 2 \times 2^{-n[I(X \wedge
\tilde{X})- \delta]},
\end{align}
therefore,
\begin{align}
Var(Y_{21}) \leq E(Y_{21}^2) = 2 \times 2^{-n[I(X \wedge \tilde{X})-
\delta]}.\label{var2}
\end{align}
\end{itemize}
By using~\eqref{var1}, and~\eqref{var2}, we have
\begin{align}
Var \left( \pi(X^M,V_{X\tilde{X}}) \right) \leq \frac{1}{M^2} {M
\choose 2} 4 \times 2^{-n[I(X \wedge \tilde{X})-\delta]} \leq 2
\times 2^{-n[I(X \wedge \tilde{X})-\delta]},
\end{align}
for any $V_{X\tilde{X}} \in \P(\X \times \X)$. Now, by using
Chebychev's inequality,
\begin{align}
&\mathbb{P}\left( \left|\pi(X^M,V_{X\tilde{X}}) - \mathbb{E}\left(\pi(X^M,V_{X\tilde{X}})\right)\right|  \geq 2^{n\delta} \text{ for some } V_{X\tilde{X}}\right)\nonumber\\
& \;\;\;\;\;\; \leq \sum_{V_{X\tilde{X}}} \mathbb{P}\left( \left|\pi(X^M,V_{X\tilde{X}}) - \mathbb{E}\left(\pi(X^M,V_{X\tilde{X}})\right)\right|  \geq 2^{n\delta} \right) \nonumber\\
& \;\;\;\;\;\; \leq \sum_{V_{X\tilde{X}}} \frac{Var\left(\pi(X^M,V_{X\tilde{X}}) \right)}{2^{2n\delta}} \leq \sum_{V_{X\tilde{X}}} \frac{2 \times 2^{-n [I(X \wedge \tilde{X})- \delta]}}{2^{2n\delta}} \nonumber\\
& \;\;\;\;\;\; =\sum_{V_{X\tilde{X}}} 2 \times 2^{-n (I(X \wedge
\tilde{X}) + \delta)} \leq 2^{-n\frac{\delta}{2}}, \;\;\;\;\;\;\;
\text{ for sufficiently large $n$.}\label{typical-con1}
\end{align}
Moreover, by using~\eqref{p2p:upper-expectation-lambda} and
Markov's inequality, it can be concluded that
\begin{align}
&\mathbb{P}\left( \lambda (X^M, V_{X\tilde{X}\hat{X}}) \geq
2^{n[E_S(R,V_{X\tilde{X}\hat{X}}) +4 \delta]} \text{ for some }
V_{X\tilde{X}\hat{X}}\right) \leq \sum_{V_{X\tilde{X}\hat{X}}}
 \frac{\mathbb{E} \lambda (X^M,
V_{X\tilde{X}\hat{X}})}{2^{n[E_S(R,V_{X\tilde{X}\hat{X}}) +4
\delta]}} \leq 2^{-n\delta}.\label{typical-con2}
\end{align}
Now, by combining~\eqref{typical-con1} and~\eqref{typical-con2} and using the bound on $\mathbb{E} \left(\pi(X^M, V_{X\tilde{X}})\right)$, we
conclude that for any $V_{X\tilde{X}} \in \P(\X \times \X)$, any
$V_{X\tilde{X}\hat{X}} \in \P(\X \times \X \times \X)$, for
sufficiently large $n$
\begin{align}
2^{n (R - I(X \wedge \tilde{X})- \delta)} \leq
\pi(X^M,V_{X\tilde{X}}) \leq 2^{n (R - I(X \wedge \tilde{X}) +
\delta)}, \label{p2p:packing-T-cond}\nonumber\\
 \lambda (X^M, V_{X\tilde{X}\hat{X}}) \leq 2^{n[E_S(R,V_{X\tilde{X}\hat{X}}) +4
 \delta]},\quad \quad\quad
\end{align}
with probability $> 1 - 2 \times 2^{-n\frac{\delta}{2}}$. We put all
the codebooks satisfying~\eqref{p2p:packing-T-cond} in a set called
$\C^T$.
\end{proof}

\begin{proof}\textbf{(Lemma~\ref{lm:p2p_expurgated})}
Consider the code $C^r \triangleq \left\{\mathbf{x}_1,
\mathbf{x}_2,..., \mathbf{x}_M\right\}$ whose existence is asserted
in random coding packing lemma.
 Let
us define
\begin{align}
\Pi(C^r)\triangleq \sum_{V_{X\tilde{X}}}
 2^{-n(R -I_V(X \wedge
\tilde{X}) +3 \delta)} \pi(C^r,V_{X\tilde{X}}).
\end{align}
Note that using Lemma~\ref{P2P-random-Packing} and using the fact
that $\Pi(C^r)= \frac{1}{M} \sum_{i=1}^{M}
\left\{\sum_{V_{X\tilde{X}}} |T_{V_{\tilde{X}|X}}(\mathbf{x}_i) \cap
C^r| 2^{-n(R -I_V(X \wedge \tilde{X}) +3 \delta)} \right\}$, it can
be concluded that
\begin{align}
\Pi(C^r) \leq \sum_{V_{X\tilde{X}}} 2^{-n(R -I_V(X \wedge \tilde{X})
+3 \delta)} 2^{n(R -I_V(X \wedge \tilde{X}) +2 \delta)} <
\frac{1}{2}.
\end{align}
As a result, it can be concluded that there exists $M^{*} \geq
\frac{M}{2}$ codewords in $C^r$ satisfying
\begin{align}
\sum_{V_{X\tilde{X}}} |T_{V_{\tilde{X}|X}}(\mathbf{x}_i) \cap C^r|
2^{-n(R -I_V(X \wedge \tilde{X}) +3 \delta)} < 1.
\label{P2P-Ex-proof1}
\end{align}
Let us call this subset of the code as $C^{ex}$. Without loss of
generality, we assume $C^{ex}$ contains the first $M^*$ sequences of
$C^r$, i.e., $C^{ex}=\{\mathbf{x}_1,
\mathbf{x}_2,...,\mathbf{x}_{M^{*}}\}$. Since
\begin{align}
|T_{V_{\tilde{X}|X}}(\mathbf{x}_i) \cap C^{ex}| \leq
|T_{V_{\tilde{X}|X}}(\mathbf{x}_i) \cap C^r| \;\;\; \forall
\mathbf{x}_i \in C^{ex},
\end{align}
it can be concluded that for all $\mathbf{x}_i \in C^{ex}$,
\begin{align}
\sum_{V_{X\tilde{X}}} |T_{V_{\tilde{X}|X}}(\mathbf{x}_i) \cap
C^{ex}| 2^{-n(R -I_V(X \wedge \tilde{X}) +3 \delta)} < 1.
\end{align}
Since all the terms in the summation are non-negative terms, we
conclude that
\begin{align}\label{P2P-proof-Ex1}
|T_{V_{\tilde{X}|X}}(\mathbf{x}_i) \cap C^{ex}|  < 2^{n(R -I_V(X
\wedge \tilde{X}) +3 \delta)},
\end{align}
for all $V_{X\tilde{X}} \in \P (\X \times \X)$, and all
$\mathbf{x}_i \in C^{ex}$. Also, by~\eqref{P2P-proof-Ex1}, it can be
concluded that for all $V_{X\tilde{X}} \in \P (\X \times \X)$
\begin{align}
\pi(C^{ex},V_{X\tilde{X}}) = \frac{1}{M^*} \sum_{i=1}^{M^*}
|T_{V_{\tilde{X}|X}}(\mathbf{x}_i) \cap C^{ex}| \leq 2^{n(R -I_V(X
\wedge \tilde{X}) +3 \delta)}.
\end{align}
\end{proof}
\begin{proof}\textbf{(Fact~\ref{fc:p2p_random_coding})} We will use
the result of Lemma~\ref{P2P-random-Packing} and the relation
between
 the probability of error and the packing functions.
Let $X^{M} \triangleq \Big(X_1,X_2,...,X_{M} \Big)$
be independent sequences of independent random variable, where
$X_i$s are uniformly distributed on $T_{P}$.

\textbf{(Upper Bound):}
Taking expectation on both sides of (\ref{upperboundA}), using
Lemma~\ref{P2P-random-Packing}
and using the continuity of
information measures, it can be concluded that
\begin{align}
\mathbb{E}\left(e(X^M,W)\right) &\leq
\sum_{\substack{V_{X\tilde{X}Y} \in \P_n^r}} 2^{-n[D(V_{Y|X}||W|P)+
|I(\tilde{X} \wedge XY) - R|^+  - \delta]}\nonumber\\
&\leq 2^{-n[E_r(R,P,W)- 2\delta]}\label{PP-avg-rand-upper}
\end{align}
whenever $n \geq n_1(|\X|,|\Y|,\delta)$, where
\begin{equation}
E_r(R,P,W) \triangleq \min_{\substack{V_{X\tilde{X}Y}} \in \P^r}
D(V_{Y|X} || W |P) + |I_V(XY \wedge \tilde{X})-R|^+,
\label{P2P-Er-Def}
\end{equation}
and $\P^r$ is defined in~\eqref{P^r-def}.\\ \textbf{(Lower Bound):}
Taking expectation on both sides of (\ref{lowerboundA}), and using
Lemma~\ref{P2P-random-Packing} we have

\begin{align}
\bar{P_e} &= \mathbb{E}e(X^M,W) \geq \sum_{\substack{V_{X\tilde{X}Y}
\in \P_n^r}} 2^{-n[D(V_{Y|X}||W|P)+I_V(\tilde{X} \wedge
Y|X)+\delta]}
\left| 2^{n(R-I(X \wedge \tilde{X}) -\delta)}  - \right. \hspace{1in} \nonumber\\
&  \hspace{2.5in} \sum_{\substack{V_{X\tilde{X}\hat{X}Y}:\\V_{X\hat{X}Y}
=V_{X\tilde{X}Y}}} 2^{-n[I_V(\hat{X} \wedge Y|X \tilde{X})]}
 2^{n(2R-I(X \wedge \tilde{X})-I(\hat{X} \wedge  X \tilde{X}) +4 \delta)}  \Big|^+  \\
&= \sum_{\substack{V_{X\tilde{X}Y} \in \P_n^r}}
2^{-n[D(V_{Y|X}||W|P)+I_V(\tilde{X} \wedge XY)-R+2\delta]} \left| 1
- \sum_{\substack{V_{X\tilde{X}\hat{X}Y}:\\V_{X\hat{X}Y}
=V_{X\tilde{X}Y}}} 2^{-n[I_V(\hat{X} \wedge X \tilde{X} Y)-R
+3\delta ]}   \right|^+ \label{beforelemma}
\end{align}
Toward further simplification of this expression, we have the
following lemma.
\begin{lemma}
\label{intermmediatelemma}
\begin{align}
\min_{\substack{V_{X\tilde{X}\hat{X}Y}:\\V_{X\hat{X}Y}
=V_{X\tilde{X}Y} }} I(\hat{X} \wedge X\tilde{X}Y)= I(\tilde{X}
\wedge XY).
\end{align}
\end{lemma}
\begin{proof}
Note that, for any $V_{X\tilde{X}\hat{X}Y}$,
\begin{align}
I(\hat{X} \wedge X\tilde{X}Y) = I(\hat{X} \wedge XY) + I(\hat{X}
\wedge \tilde{X}|XY) \geq I(\hat{X} \wedge XY),
\end{align}
therefore,
\begin{align}
\min_{\substack{V_{X\tilde{X}\hat{X}Y}:\\V_{X\hat{X}Y} =
V_{X\tilde{X}Y} }} I(\hat{X} \wedge X\tilde{X}Y) \geq  I(\hat{X}
\wedge XY) = I(\tilde{X} \wedge XY).\label{XXXY1>}
\end{align}
Now, consider $V^*_{X\tilde{X}\hat{X}Y}$ defined as
\begin{align}
V^*_{X\tilde{X}\hat{X}Y}(x, \tilde{x}, \hat{x}, y) =
V_{\tilde{X}|XY}(\tilde{x}|x,y) V_{\tilde{X}|XY}(\hat{x}|x,y)
V_{XY}(x,y).
\end{align}
Note that $V^*_{X\hat{X}Y} = V^*_{X\tilde{X}Y}$, and $\tilde{X} -
(X,Y) - \hat{X}$. Therefore,
\begin{align}
I_{V^*}(\hat{X} \wedge X\tilde{X}Y) = I_V(\hat{X} \wedge XY) =
I_V(\tilde{X} \wedge XY).\label{lemma-XXXY1}
\end{align}
By combining~\eqref{XXXY1>} and~\eqref{lemma-XXXY1}, the proof is
complete.
\end{proof}

Therefore, using the above lemma, ~\eqref{beforelemma} can be rewritten as
\begin{align}
\bar{P_e} &\geq
\sum_{\substack{V_{X\tilde{X}Y} \in \P_n^r \\
I(\tilde{X} \wedge XY)
> R + 3\delta}} 2^{-n[D(V_{Y|X}||W|P)+I_V(\tilde{X} \wedge XY) - R +3\delta
]}.
\end{align}
By using the continuity of information measures, it can be concluded
that
\begin{align} & \mathbb{E}\left(e(X^M,W)\right) \geq
2^{-n[E_L(R,P,W)+4 \delta]},\qquad \qquad \text{for sufficient large
 }\;n
\end{align}
where
\begin{equation}
E_L(R,P,W) \triangleq \min_{\substack{V_{X\tilde{X}Y} \in \P^r
\\ I(\tilde{X} \wedge XY)
\geq  R}} D(V_{Y|X} || W |P) + I_V(XY \wedge \tilde{X})-R.
\end{equation}
Now, by using Markov inequality and~\eqref{PP-avg-rand-upper}, we
conclude that
\begin{align}
\mathbb{P} \left(e(X^M,W) \geq 2^{-n[E_r(R,P,W)- 3\delta]}\right)
\leq \frac{\mathbb{E}\left(e(X^M,W)\right)}{2^{-n[E_r(R,P,W)-
3\delta]}} \leq 2^{-n\delta}.
\end{align}
Therefore, with probability greater than $1- 2^{-n\delta}$, any
selected code with $M$ codewords form the constant composition code
ensemble satisfies the desired property. Let us call one of these
codebooks as $C^r$.
\end{proof}

\begin{proof}\textbf{(Corollary~\ref{cor:erl_eq_er})}
Consider the input distribution $P^* \in \P(\X)$ maximizing the
random coding bound, i.e.,
\begin{align}
P^* \triangleq \arg\max_{P \in \P(\X)} E_r(R,P,W).
\end{align}
Let us define
\begin{align}
V^*_{X\tilde{X}Y} \triangleq \arg\min_{V_{X\tilde{X}Y}}
E_r(R,P^*,W).
\end{align}
For any $R \leq R_{crit}$, the random coding bound is a straight
line with slope $-1$, and the term in $|\cdot|^+$ is active.
Therefore,
\begin{align}
E_r(R,P^*,W) = D(V^*_{Y|X} || W |P^*) + I_{V^*}(\tilde{X} \wedge
XY)-R.
\end{align}
Here, $ I_{V^*}(\tilde{X} \wedge XY) \geq R$. It is clear that
$V^*_{X\tilde{X}Y}$ is the minimizing distribution in
$E_{rL}(R,P^*,W)$, and as a result
\begin{align}
E_{rL}(R,P^*,W) = E_r(R,P^*,W).
\end{align}
\end{proof}

\begin{proof}\textbf{(Theorem~\ref{th:p2p_typical})}
In the proof of Fact~\ref{fc:p2p_random_coding}, we used the lower
and upper bounds on the expected value of he first-order packing
functions and an upper bound on the expected value of the
second-order packing functions. In the following, we use similar
techniques on the packing function of almost every codebook in the
ensemble by using the bounds obtained in Lemma~\ref{lm:p2p_typical}.
Consider the code $C$ whose existence is asserted in the typical
random coding packing lemma. For all $V_{X\tilde{X}} \in \P (\X
\times \X)$, we have
\begin{align}\label{lemma2-formula1}
\frac{1}{M} \sum_{i=1}^{M} |T_{V_{\tilde{X}|X}}(\mathbf{x}_i) \cap
C| \leq 2^{n(R-I_V(X \wedge \tilde{X})+2 \delta)}.
\end{align}
By multiplying both sides of inequality~\eqref{lemma2-formula1} by
$M$, and using the proper upper bound on the number of sequences in
$C$, we conclude that
\begin{equation}\label{lemma2-formula2}
|T_{V_{\tilde{X}|X}}(\mathbf{x}_i) \cap C| \leq 2^{n(2R-I_V(X \wedge
\tilde{X})+2\delta)}\;\;\;\;\;\;\;\;\;\forall i=1,...,M,
\end{equation}
for all $V_{X\tilde{X}} \in \P (\X \times \X)$. We will obtain a
higher error exponent for almost all codes by removing certain types
from the constraint set $\P_n^r$. Consider any $V_{X\tilde{X}} \in
\P (\X \times \X)$ satisfying $I_V(X \wedge \tilde{X})
> 2(R +\delta)$. By~\eqref{lemma2-formula2},
\begin{align}\label{typical-condition2}
|T_{V_{\tilde{X}|X}}(\mathbf{x}_i) \cap C| = 0 \text{  for all $i$
}\Rightarrow \pi(C,V_{X\tilde{X}})= 0.
\end{align}

\noindent \textbf{Upper bound:} Hence, by using (\ref{upperboundA})
on $C$, and by using the result of Lemma \ref{lm:p2p_typical}, we
have
\begin{align}
e(C,W) &\leq \sum_{\substack{V_{X\tilde{X}Y} \in \P_n^T(\delta)}}
2^{-n[D(V_{Y|X}||W|P) + |I_V(XY \wedge \tilde{X}) - R|^+ -
2\delta]}\nonumber
\end{align}
where
\begin{eqnarray}
\P_n^T(\delta) \triangleq \big\{V_{X\tilde{X}Y} \in \P_n(\X \times \X \times
\Y): \;V_X=V_{\tilde{X}}=P,\; I_V(X \wedge \tilde{X}) \leq 2R+2 \delta
\;,\;\alpha(P ,V_{Y|\tilde{X}}) \leq \alpha(P ,V_{Y|X})\big\}.
\end{eqnarray}
Using the continuity of
information measures, the upper bound as given by the theorem
follows.

\noindent \textbf{Lower bound:} Using (\ref{lowerboundA}) on $C$ and
using Lemma \ref{lm:p2p_typical}, we have
\begin{align}
e(C,W) &\geq \sum_{\substack{V_{X\tilde{X}Y} \in \P_n^r}}
2^{-n[D(V_{Y|X}||W|P)+I_V(\tilde{X} \wedge Y|X)+\delta]} \left|
\pi(C,V_{X\tilde{X}}) -
\sum_{\substack{V_{X\tilde{X}\hat{X}Y}:\\V_{X\hat{X}Y}
=V_{X\tilde{X}Y}}} 2^{-n[I_V(\hat{X} \wedge Y|X \tilde{X})]}
 \lambda(C,V_{X\tilde{X}\hat{X}}) \right|^+  \nonumber \\
&\geq \sum_{\substack{V_{X\tilde{X}Y} \in \P_n^T(\delta)}}
2^{-n[D(V_{Y|X}||W|P)+I_V(\tilde{X} \wedge Y|X)+\delta]}
\left| 2^{n(R-I(X \wedge \tilde{X}) -\delta)}  - \right. \hspace{1in} \nonumber\\
&  \hspace{2.5in} \sum_{\substack{V_{X\tilde{X}\hat{X}Y}:\\V_{X\hat{X}Y}
=V_{X\tilde{X}Y}}} 2^{-n[I_V(\hat{X} \wedge Y|X \tilde{X})]}
 2^{n(2R-I(X \wedge \tilde{X})-I(\hat{X} \wedge  X \tilde{X}) +2 \delta)}  \Big|^+  \\
&= \sum_{\substack{V_{X\tilde{X}Y} \in \P_n^T(\delta)}}
2^{-n[D(V_{Y|X}||W|P)+I_V(\tilde{X} \wedge XY)-R+2\delta]} \left| 1
- \sum_{\substack{V_{X\tilde{X}\hat{X}Y}:\\V_{X\hat{X}Y}
=V_{X\tilde{X}Y}}} 2^{-n[I_V(\hat{X} \wedge X \tilde{X} Y)-R -3
\delta
]}   \right|^+  \\
&\geq \sum_{\substack{V_{X\tilde{X}Y} \in \P_n^T(\delta) \\
I(\tilde{X} \wedge XY)
> R + 5\delta}} 2^{-n[D(V_{Y|X}||W|P)+I_V(\tilde{X} \wedge XY) - R +3\delta
]},
\end{align}
Here, the last inequality follows
 from Lemma \ref{intermmediatelemma}.

By using the continuity argument, and for sufficient large $n$,
\begin{align}
& e(C,W) \geq 2^{-n[E_{LT}(R,P,W)+4 \delta]},
\end{align}
where
\begin{equation}
E_{LT}(R,P,W) \triangleq \min_{\substack{V_{X\tilde{X}Y} \in \P^T
\\ I(\tilde{X} \wedge XY)
\geq R}} D(V_{Y|X} || W |P) + I_V(XY \wedge \tilde{X})-R.
\end{equation}
\end{proof}

\begin{proof}(\textbf{Corollary~\ref{P2P-typical-Cor1}})
Fix $R \geq 0$, $\delta > 0$. By the result of
Theorem~\ref{th:p2p_typical} and for sufficiently large $n$, there
exists a collection of codes, $\C^*$, with length $n$ and rate $R$,
such that
\begin{itemize}
\item $\mathbb{P} \left( \C^* \right) \geq 1 - \delta$,
\item $2^{-n[E_{TL}(R,P,W)+ 4\delta]} \leq e(C,W) \leq 2^{-n[E_T(R,P,W)-
3\delta]} \text{    for all } C \in \C^*$.
\end{itemize}
Note that
\begin{align}
\max_{ \tilde{\C}: \mathbb{P}(\tilde{\C} ) > 1- \delta}\;\; \min_{
\substack{ C \in \tilde{\C}}} -\frac{1}{n} \log{e(C,W)}  \geq \min_{
\substack{ C \in \C^*}} -\frac{1}{n} \log{e(C,W)} \geq E_T(R,P,W)-
3\delta.\label{P2P-typical-Cor1-proof-Ineq1}
\end{align}
Now, consider any high probability collection of codes with length
$n$ and rate $R$. Let us call this collection as $\hat{\C}$. Note
that
\begin{align}
 \left.
\begin{array}{lc}
&\mathbb{P} \left( \C^* \right) \geq 1 - \delta \\
& \mathbb{P}( \hat{\C} ) \geq 1 - \delta
\end{array}\right\} \Rightarrow \mathbb{P} ( \C^*  \cap  \hat{\C} ) \geq 1
-2 \delta \Rightarrow  \C^*  \cap  \hat{\C} \neq \phi.
\end{align}
Consider a code $C(\hat{\C}) \in \C^*  \cap  \hat{\C}$. It can be
concluded that
\begin{align}
\max_{ \tilde{\C}: \mathbb{P}(\tilde{\C} ) > 1- \delta}\;\; \min_{
\substack{ C \in \tilde{\C}}} -\frac{1}{n} \log{e(C,W)}  \leq \max_{
\tilde{\C}: \mathbb{P}(\tilde{\C} ) > 1- \delta}\;\; -\frac{1}{n}
\log{e(C(\tilde{\C}),W)} \leq E_{LT}(R,P,W)+ 4\delta.
\label{P2P-typical-Cor1-proof-Ineq2}
\end{align}
The last inequality follows from the fact that $C(\hat{\C}) \in
\C^*$. By combining~\eqref{P2P-typical-Cor1-proof-Ineq1}
and~\eqref{P2P-typical-Cor1-proof-Ineq2}, and by letting $\delta$
goes to zero and $n$ goes to infinity, it can be concluded that
\begin{align}
E_T(R,P,W) \leq E_{av}^T (R) \leq E_{TL}(R,P,W).
\end{align}
\end{proof}


\begin{proof}\textbf{(Fact~\ref{fc:p2p_expurgated})} First, we prove
the following lemma.
\begin{lemma}
Let $C^{ex}$ be the collection of the codewords whose existence is
asserted in Lemma~\ref{lm:p2p_expurgated}. For any distribution
$V_{X\tilde{X}} \in \P_n(\X \times \X)$, satisfying $I_V(X \wedge
\tilde{X}) > R+\delta$, the following holds:
\begin{equation}
\pi(C^{ex},V_{X\tilde{X}} )= 0.
\end{equation}
\end{lemma}
\begin{proof}
By~\eqref{condition3-2},
\begin{equation}\label{P2P-expurgated-random-bound-1}
|T_{V_{\tilde{X}|X}}(\mathbf{x}_i) \cap C^{ex}| \leq 2^{n(R-I_V(X
\wedge \tilde{X})+2\delta)},
\end{equation}
for every $\mathbf{x}_i \in C^{ex}$. Since $I_V(X \wedge \tilde{X})
> R+2\delta$, it can be concluded that
\begin{align}
|T_{V_{\tilde{X}|X}}(\mathbf{x}_i) \cap C^{ex}| =0 \text{ for every
} \mathbf{x}_i \in C^{ex} \Rightarrow \pi(C^{ex},V_{X\tilde{X}} )= 0
\end{align}
\end{proof}
The rest of the proof is identical to the proof of random coding
bound.
\end{proof}


\subsection*{2. MAC Proofs}
\begin{proof} \textbf{(Lemma~\ref{packing1})}
In this proof, we use a similar random coding argument that Pokorny
and Wallmeier used in \cite{Pokorney}. The main difference is that
our lemma uses a different code ensemble which results in a tighter
bound. Instead of choosing our sequences from $T_{P_X}$ and
$T_{P_Y}$,  we choose our random sequences uniformly from
$T_{P_{X|U}}(\mathbf{u})$, and $T_{P_{Y|U}}(\mathbf{u})$ for a given
$\mathbf{u} \in T_{P_U}$. In \cite{Liu-RandomCoding}, we see a
similar random code ensemble, however, their packing lemma
incorporates the channel output $\mathbf{z}$ into the packing
inequalities. One can easily show that, by using this packing lemma
and considering the minimum equivocation decoding rule, we would end
up with the random coding bound derived in~\cite{Liu-RandomCoding}.

Fix any $\U$, $P_{XYU} \in \P_n(\U \times \X \times \Y)$ such that
$X-U-Y$, $R_X \geq 0$, $R_Y \geq 0$ , $\delta > 0$, and $\mathbf{u}
\in T_{P_U}$. Define $M_X$, $M_Y$ such that
\begin{eqnarray*}
2^{n(R_X-\delta)} \leq M_X \leq 2^{nR_X},\qquad 2^{n(R_Y-\delta)}
\leq M_Y \leq 2^{nR_Y}.
\end{eqnarray*}

First, we find upper bounds on the expectations of packing functions
for a fixed $\alpha$ and $V_{UXY\tilde{X}\tilde{Y}}$, with respect
to the random variables $X_i$ and $Y_j$.
Since $X_i$s and $Y_j$s are i.i.d random sequences, we have
\begin{align}
\mathbb{E}\Big[N_{U} (X^{M_X} \times Y^{M_Y},V_{UXY})\Big]
&\triangleq \mathbb{E}\Big[ \frac{1}{M_X M_Y} \sum_{i,j}
1_{T_{V_{UXY}}} (\mathbf{u},
X_i,Y_j)\Big]\nonumber\\
&= \mathbb{E} \big[ 1_{T_{V_{UXY}}} (\mathbf{u}, X_1,Y_1)\big]\nonumber\\
&= \sum_{\mathbf{x},\mathbf{y}}
1_{T_{V_{XY|U}}}(\mathbf{x},\mathbf{y}|\mathbf{u})
\mathbb{P}(X_1=\mathbf{x}|\mathbf{u}) \mathbb{P}(Y_1=\mathbf{y}|\mathbf{u})\nonumber\\
& \leq \sum_{\substack{(\mathbf{x},\mathbf{y})\in T_{V_{XY|U}}}
(\mathbf{u})} 2^{-n[ H_V(X|U)-\delta]} 2^{-n[ H_V(Y|U)-\delta]}\nonumber\\
& \leq 2^{nH_V(XY|U)} 2^{-n[ H_V(X|U)-\delta]} 2^{-n[ H_V(Y|U)-\delta]}\nonumber\\
&=  2^{-n[I_V(X \wedge Y|U) - 2 \delta]}= 2^{-n[F_U(V_{UXY})- 2
\delta]}.\label{alpha1}
\end{align}
On the other hand,
\begin{align}
\mathbb{E}\Big[N_{U} (X^{M_X} \times Y^{M_Y},V_{UXY})\Big] &=
\sum_{\mathbf{x},\mathbf{y}}
1_{T_{V_{XY|U}}}(\mathbf{x},\mathbf{y}|\mathbf{u})
\mathbb{P}(X_1=\mathbf{x}|\mathbf{u}) \mathbb{P}(Y_1=\mathbf{y}|\mathbf{u})\nonumber\\
& \geq \sum_{\substack{(\mathbf{x},\mathbf{y})\in T_{V_{XY|U}}}
(\mathbf{u})} 2^{-n H_V(X|U)} 2^{-n H_V(Y|U)}\nonumber\\
& \geq 2^{n[H_V(XY|U)-\delta]} 2^{-n H_V(X|U)} 2^{-nH_V(Y|U)}\nonumber\\
& = 2^{-n[I_V(X \wedge Y|U) + \delta]}= 2^{-n[F_U(V_{UXY})+
\delta]}.\label{alpha5}
\end{align}
Therefore, by~\eqref{alpha1} and~\eqref{alpha5},
\begin{align}
 2^{-n[F_U(V_{UXY})+
\delta]} \leq \mathbb{E}\Big[N_{U} (X^{M_X} \times
Y^{M_Y},V_{UXY})\Big] \leq 2^{-n[F_U(V_{UXY})- 2 \delta]}.
\end{align}
By using a similar argument,
\begin{align}
\mathbb{E}\Big[N_{X} (X^{M_X} \times Y^{M_Y},V_{UXY\tilde{X}})\Big]
\geq 2^{-n[F_X(V_{UXY\tilde{X}})-4 \delta]}.\label{alpha2}
\end{align}
On the other hand,
\begin{align}
\mathbb{E}\Big[N_{X} (X^{M_X} \times Y^{M_Y},V_{UXY\tilde{X}})\Big]
& \geq \left(M_X-1\right) \mathbb{E}\Big[1_{T_{V_{UXY}}}
(\mathbf{u}, X_1,Y_1)
1_{T_{V_{UXY\tilde{X}}}}(\mathbf{u}, X_1,Y_1,X_2)\Big]\nonumber\\
& = \left(M_X-1\right) \sum_{\mathbf{x},\mathbf{y}}
\mathbb{P}(X_1=\mathbf{x}|\mathbf{u})
\mathbb{P}(Y_1=\mathbf{y}|\mathbf{u})1_{T_{V_{UXY}}} (\mathbf{u},
\mathbf{x},\mathbf{y})\nonumber\\
& \;\;\;\;\;\;\;\;\;\;\;\;\;\;\;\;\;\;\; \cdot
\sum_{\tilde{\mathbf{x}}} \mathbb{P}(X_2=\tilde{\mathbf{x}}|\mathbf{u})
1_{T_{V_{UXY\tilde{X}}}}(\mathbf{u}, \mathbf{x},\mathbf{y},
\tilde{\mathbf{x}})\nonumber\\
& \geq \left(M_X-1\right) \sum_{\mathbf{x},\mathbf{y} \in
T_{V_{XY|U}}(\mathbf{u})}
2^{-nH_V(X|U)} 2^{-nH_V(Y|U)}\nonumber\\
&\;\;\;\;\;\;\;\; \sum_{\tilde{\mathbf{x}} \in
T_{V_{\tilde{X}|UXY}}(\mathbf{u},\mathbf{x},\mathbf{y})} 2^{-n
H_V(\tilde{X}|U)}\nonumber\\
& \geq \left(M_X-1\right) 2^{n[H(XY|U)-\delta]} 2^{-nH_V(X|U)} 2^{-n
H_V(Y|U)}\nonumber\\
&\;\;\;\;\;\;\;\;\;\;\;\;\;\;\;\;\;\; \cdot 2^{n
[H_V(\tilde{X}|UXY)-\delta]}
2^{-n H_V(\tilde{X}|U)}\nonumber\\
& \geq 2^{-n[I_V(X \wedge Y|U) + I_V(\tilde{X} \wedge Y|U)
+I_V(\tilde{X} \wedge X|UY)-R_X + 3 \delta] }\nonumber\\
& = 2^{-n[F_X(V_{UXY\tilde{X}})+3 \delta]}.\label{alpha6}
\end{align}
Therefore, by~\eqref{alpha2} and~\eqref{alpha6},
\begin{align}
 2^{-n[F_X(V_{UXY\tilde{X}})+ 3
\delta]} \leq \mathbb{E}\Big[N_{X} (X^{M_X} \times
Y^{M_Y},V_{UXY\tilde{X}})\Big] \leq 2^{-n[F_X(V_{UXY\tilde{X}})- 4
\delta]}.
\end{align}
By using a similar argument for $ N_{Y} (X^{M_X} \times
Y^{M_Y},V_{UXY\tilde{X}}) $ and $ N_{XY} (X^{M_X} \times
Y^{M_Y},V_{UXY\tilde{X}\tilde{Y}}) $, we can show that
\begin{align}
 2^{-n[F_Y(V_{UXY\tilde{Y}})+ 3
\delta]} \leq \mathbb{E}\Big[N_{Y} (X^{M_X} \times
Y^{M_Y},V_{UXY\tilde{Y}})\Big] \leq
2^{-n[F_Y(V_{UXY\tilde{Y}})-4 \delta]},\qquad\label{alpha3}\\
2^{-n[F_{XY}(V_{UXY\tilde{X}\tilde{Y}}) + 4 \delta]} \leq
\mathbb{E}\Big[N_{XY} (X^{M_X} \times
Y^{M_Y},V_{UXY\tilde{X}\tilde{Y}})\Big] \leq
2^{-n[F_{XY}(V_{UXY\tilde{X}\tilde{Y}})- 4 \delta]}.
\end{align}
We can obtain an upper bound for $\mathbb{E} \left[ \Lambda_{XY}
(X^{M_X} \times
Y^{M_Y},V_{UXY\tilde{X}\tilde{Y}\hat{X}\hat{Y}})\right] $ as follows
\begin{align}
&\mathbb{E} \left[ \Lambda_{XY}
(X^{M_X},Y^{M_Y},V_{UXY\tilde{X}\tilde{Y}\hat{X}\hat{Y}}) \right]\quad\quad\quad\quad\quad\quad\quad\quad\quad\quad\quad\quad\quad\quad\quad\quad\quad\quad\quad\quad\quad\quad\quad\quad\quad\quad\quad\quad\quad\quad\quad\quad\quad\quad\quad\quad\quad\quad\quad\quad\quad\quad\quad\quad\quad\quad\quad\quad\quad\quad\quad\quad\quad\quad\quad\quad\quad\quad\quad\quad \nonumber\\
&\quad\quad\quad\quad\quad= \mathbb{E}
\left[\frac{1}{M_X M_Y} \sum_{i,j} \sum_{\substack{k \neq i\\
l \neq j}}\sum_{\substack{k' \neq i,k\\ l' \neq j,l}}
1_{T_{V_{UXY\tilde{X}\tilde{Y}\hat{X}\hat{Y}}}}(\mathbf{u},X_i,Y_j,X_k,Y_l,
X_{k'},Y_{l'})\right]\nonumber\\
& \quad\quad\quad\quad\quad\leq M^2_X M^2_Y
\mathbb{E}\Big[1_{T_{V_{UXY}}} (\mathbf{u}, X_1,Y_1)
1_{T_{V_{UXY\tilde{X}\tilde{Y}\hat{X}\hat{Y}}}}(\mathbf{u},X_1,Y_1,X_2,Y_2,
X_{3},Y_{3})\Big]\nonumber\\
& \quad\quad\quad\quad\quad= M^2_X M^2_Y
\sum_{\mathbf{x},\mathbf{y},\tilde{\mathbf{x}},\tilde{\mathbf{y}},\hat{\mathbf{x}},\hat{\mathbf{y}}}
\mathbb{P}(X_1=\mathbf{x},Y_1= \mathbf{y}, X_2=
\tilde{\mathbf{x}},Y_2= \tilde{\mathbf{y}},X_3=
\hat{\mathbf{x}},Y_3=
\hat{\mathbf{y}}|\mathbf{u})\nonumber\\
&
\quad\quad\quad\quad\quad\quad\quad\quad\quad\quad\quad\quad\quad\quad\quad\;\;\;\;\;\;\;\;\;\;\;\;\;\;\;\;\;\;\;
\cdot 1_{T_{V_{UXY}}} (\mathbf{u}, \mathbf{x},\mathbf{y}).
1_{T_{V_{UXY\tilde{X}\tilde{Y}\hat{X}\hat{Y}}}}(\mathbf{u},
\mathbf{x},\mathbf{y},
\tilde{\mathbf{x}},\tilde{\mathbf{y}},\hat{\mathbf{x}},\hat{\mathbf{x}})\nonumber\\
& \quad\quad\quad\quad\quad= M^2_X M^2_Y\sum_{\mathbf{x},\mathbf{y}}
\mathbb{P}(X_1=\mathbf{x}|\mathbf{u}) Pr(Y_1=\mathbf{y}|\mathbf{u})
\cdot 1_{T_{V_{UXY}}} (\mathbf{u},
\mathbf{x},\mathbf{y})\nonumber\\
& \;\;\;\;\;\;\;\;\;\;\;\;\;\;\;\;\;\;\; \quad
\cdot\sum_{\tilde{\mathbf{x}}}
\mathbb{P}(X_2=\tilde{\mathbf{x}}|\mathbf{u})
1_{T_{V_{UXY\tilde{X}}}}(\mathbf{u}, \mathbf{x},\mathbf{y},
\tilde{\mathbf{x}})  \sum_{\tilde{\mathbf{y}}}
\mathbb{P}(Y_2=\tilde{\mathbf{y}}|\mathbf{u})
1_{T_{V_{UXY\tilde{X}\tilde{Y}}}}(\mathbf{u}, \mathbf{x},\mathbf{y},
\tilde{\mathbf{x}},\tilde{\mathbf{y}})\nonumber\\
& \;\;\;\;\;\;\;\;\;\;\;\;\;\;\;\;\;\;\; \quad
\cdot\sum_{\hat{\mathbf{x}}}
\mathbb{P}(X_3=\hat{\mathbf{x}}|\mathbf{u}) 1_{T_{V_{UXY\tilde{X}
\tilde{Y}\hat{X}}}}(\mathbf{u}, \mathbf{x},\mathbf{y},
\tilde{\mathbf{x}},\tilde{\mathbf{y}},\hat{\mathbf{x}})
\sum_{\hat{\mathbf{y}}} \mathbb{P}(Y_3=\hat{\mathbf{y}}|\mathbf{u})
1_{T_{V_{UXY\tilde{X} \tilde{Y}\hat{X}\hat{Y}}}}(\mathbf{u},
\mathbf{x},\mathbf{y},
\tilde{\mathbf{x}},\tilde{\mathbf{y}},\hat{\mathbf{x}},\hat{\mathbf{y}})  \nonumber\\
& \quad\quad\quad\quad\quad\leq M^2_X M^2_Y
\sum_{\mathbf{x},\mathbf{y} \in T_{V_{XY|U}}(\mathbf{u})} 2^{-n[
H_V(X|U)-\delta]} 2^{-n[ H_V(Y|U)-\delta]} \sum_{\tilde{\mathbf{x}}
\in T_{V_{\tilde{X}|UXY}}(\mathbf{u},\mathbf{x},\mathbf{y})} 2^{-n[
H_V(\tilde{X}|U)-\delta]}\nonumber\\
&\quad\quad\quad\quad\quad\quad\quad \cdot \sum_{\tilde{\mathbf{y}}
\in
T_{V_{\tilde{Y}|UXY\tilde{X}}}(\mathbf{u},\mathbf{x},\mathbf{y},\tilde{\mathbf{x}})}
2^{-n[ H_V(\tilde{Y}|U)-\delta]}  \sum_{\hat{\mathbf{x}} \in
T_{V_{\hat{X}|UXY\tilde{X}\tilde{Y}}}(\mathbf{u},\mathbf{x},\mathbf{y},\tilde{\mathbf{x}},\tilde{\mathbf{y}})}
2^{-n[
H_V(\hat{X}|U)-\delta]}\nonumber\\
&\quad\quad\quad\quad\quad\quad\quad \cdot  \sum_{\hat{\mathbf{y}}
\in
T_{V_{\hat{Y}|UXY\tilde{X}\tilde{Y}\hat{X}}}}(\mathbf{u},\mathbf{x},\mathbf{y},\tilde{\mathbf{x}},\tilde{\mathbf{y}},\hat{\mathbf{x}})
2^{-n[
H_V(\hat{Y}|U)-\delta]}  \nonumber\\
&\quad\quad\quad\quad\quad \leq M^2_X M^2_Y \cdot 2^{nH(XY|U)}
2^{-n[ H_V(X|U)-\delta]}  2^{-n[ H_V(Y|U)-\delta]} 2^{n
H_V(\tilde{X}|UXY)} 2^{-n[ H_V(\tilde{X}|U)-\delta]}
2^{n H_V(\tilde{Y}|UXY\tilde{X})}\nonumber\\
&\quad\quad\quad\quad\quad\;\;\;\;\;\;\;\;\;\;\;\;\;\;\;\;\;\; \cdot
2^{-n[ H_V(\tilde{Y}|U)-\delta]} 2^{n
H_V(\hat{X}|UXY\tilde{X}\tilde{Y})} 2^{-n[ H_V(\hat{X}|U)-\delta]}
2^{n H_V(\hat{Y}|UXY\tilde{X}\tilde{Y}\hat{X})}
2^{-n[ H_V(\hat{Y}|U)-\delta]}\nonumber\\
& \quad\quad\quad\quad\quad\leq 2^{-n[I(\tilde{X}\tilde{Y} \wedge
XY|U) + I(\hat{X}\hat{Y} \wedge XY\tilde{X}\tilde{Y}|U) +
 I(X \wedge Y|U) +I(\tilde{X}\wedge \tilde{Y}|U) + I(\hat{X}\wedge \hat{Y}|U)-2 R_X -2R_Y- 6 \delta]
 }\nonumber\\
&
\quad\quad\quad\quad\quad=2^{-n[E^{XY}_S(V_{UXY\tilde{X}\tilde{Y}\hat{X}\hat{Y}})-
6 \delta] }.
\end{align}
By using a similar argument, we can obtain the following bounds
\begin{align}
\mathbb{E} \left[ \Lambda_X (X^{M_X} \times
Y^{M_Y},V_{UXY\tilde{X}\hat{X}}) \right] \leq
2^{-n[E^{X}_S(V_{UXY\tilde{X}\hat{X}})- 4 \delta] }\\
\mathbb{E} \left[ \Lambda_Y (X^{M_X} \times
Y^{M_Y},V_{UXY\tilde{Y}\hat{Y}}) \right] \leq
2^{-n[E^{Y}_S(V_{UXY\tilde{Y}\hat{Y}})- 4 \delta] }
\end{align}
Here, $E^{X}_S$, $E^{Y}_S$ and $E^{XY}_S$ are defined
in~\eqref{ESX-def}-\eqref{ESXY-def}.

 By
using Markov inequality, it can be concluded that
\begin{align}
&\mathbb{P} \left( N_{U} (X^{M_X} \times Y^{M_Y},V_{UXY}) \geq
2^{-n[F_U(V_{UXY})- 3 \delta]} \text{ for some } V_{UXY} \right)
\nonumber\\
&\quad\quad\quad\quad\leq
\sum_{\substack{V_{UXY}:\\V_{UX}=P_{UX}\\V_{UY}=P_{UY}}}
\frac{\mathbb{E}\left( N_{U} (X^{M_X} \times Y^{M_Y},V_{UXY})
\right)}{2^{-n[F_U(V_{UXY})- 3 \delta]}} \leq
\sum_{\substack{V_{UXY}:\\V_{UX}=P_{UX}\\V_{UY}=P_{UY}}}
2^{-n\delta} \leq 2^{-n \frac{\delta}{2}}\label{formula1}
\end{align}
Similarly, it can be shown that
\begin{align}
&\quad\;\mathbb{P} \left( N_{X} (X^{M_X} \times
Y^{M_Y},V_{UXY\tilde{X}}) \geq 2^{-n[F_X(V_{UXY\tilde{X}})- 5
\delta]} \text{ for some } V_{UXY\tilde{X}} \right)
\leq 2^{-n \frac{\delta}{2}},\label{formula2}\\
&\quad\;\mathbb{P} \left( N_{Y} (X^{M_X} \times
Y^{M_Y},V_{UXY\tilde{Y}}) \geq 2^{-n[F_Y(V_{UXY\tilde{Y}})- 5
\delta]} \text{ for some } V_{UXY\tilde{Y}} \right)
\leq 2^{-n \frac{\delta}{2}},\label{formula3}\\
&\mathbb{P} \left( N_{XY} (X^{M_X} \times
 Y^{M_Y},V_{UXY\tilde{X}\tilde{Y}}) \geq
2^{-n[F_{XY}(V_{UXY\tilde{X}\tilde{Y}})- 5 \delta]} \text{ for some
} V_{UXY\tilde{X}\tilde{Y}} \right) \leq 2^{-n
\frac{\delta}{2}}.\label{formula4}
\end{align}
Now, by combining~\eqref{formula1}-\eqref{formula4}, and using the
union bound, it can be concluded that
\begin{align}
\mathbb{P} \Big( &N_{U} (X^{M_X} \times Y^{M_Y},V_{UXY}) \geq
2^{-n[F_U(V_{UXY})- 3 \delta]} \text{ for some } V_{UXY} \text{ or }
\nonumber\\&N_{X} (X^{M_X} \times Y^{M_Y},V_{UXY\tilde{X}}) \geq
2^{-n[F_X(V_{UXY\tilde{X}})- 5 \delta]} \text{ for some }
V_{UXY\tilde{X}}
\text{ or }\nonumber\\
&N_{Y} (X^{M_X} \times Y^{M_Y},V_{UXY\tilde{Y}}) \geq
2^{-n[F_Y(V_{UXY\tilde{Y}})- 5
\delta]} \text{ for some } V_{UXY\tilde{Y}} \text{ or }\nonumber\\
& N_{XY} (X^{M_X} \times Y^{M_Y},V_{UXY\tilde{X}\tilde{Y}}) \geq
2^{-n[F_{XY}(V_{UXY\tilde{X}\tilde{Y}})- 5 \delta]} \text{ for some
} V_{UXY\tilde{X}\tilde{Y}}
 \Big) \leq 4 \times 2^{-n \frac{\delta}{2}},
\end{align}
therefore, there exists at least a multi-user code with the desired
properties mentioned in~\eqref{77}-\eqref{78}.

\end{proof}

\begin{proof} \textbf{(Lemma~\ref{packing3})}
To prove that a specific property holds for almost all codes, with
certain number of codewords, in the constant composition code
ensemble, we use a second order argument method. We already have
obtained upper and lower bounds on the expectation of the desired
function over the entire ensemble. In the following, we derive an
upper bound on the variance of the packing function. Finally, by
using the Chebychev's inequality, we prove that the desired property
holds for almost all codes in the ensemble. To find the variance of
$N_{U} (X^{M_X} \times Y^{M_Y},V_{UXY})$, let us define $W_{ij}
\triangleq 1_{T_{V_{UXY}}} (\mathbf{u}, X_i,Y_j)$. Therefore, the
variance of $N_{U} (X^{M_X} \times Y^{M_Y},V_{UXY})$ can be written
as
\begin{align}
Var\left(N_{U} (X^{M_X} \times Y^{M_Y},V_{UXY})\right) &= Var \left(
\frac{1}{M_X M_Y} \sum_{i,j} 1_{T_{V_{UXY}}} (\mathbf{u}, X_i,Y_j)
\right)\nonumber\\
& = \frac{1}{M_X^2 M_Y^2}  Var \left( \sum_{i,j}
W_{ij}\right).\label{Var-N_U1}
\end{align}
Since $W_{ij}$'s are pairwise independent random
variables,~\eqref{Var-N_U1} can be written as
\begin{align}
Var\left(N_{U} (X^{M_X} \times Y^{M_Y},V_{UXY})\right) &=
\frac{1}{M_X^2 M_Y^2} \sum_{i,j} Var \left( W_{ij}
\right)\nonumber\\
& \leq \frac{1}{M_X^2 M_Y^2} \sum_{i,j} \mathbb{E} \left( W_{ij}
\right)\nonumber\\
& \leq \frac{1}{M_X M_Y} \cdot 2^{-n[F_U(V_{UXY}) - 2\delta]} \leq
2^{-n[F_U(V_{UXY}) + R_X +R_Y -2\delta]}.
\end{align}
By defining $Q^j_{ik} \triangleq1_{T_{V_{UXY\tilde{X}}}}(\mathbf{u},
X_i,Y_j,X_k)$, the variance of $N_{X} (X^{M_X} \times
Y^{M_Y},V_{UXY\tilde{X}})$ can be upper-bounded as follows
\begin{align}
Var \left( N_{X} (X^{M_X} \times Y^{M_Y},V_{UXY\tilde{X}}) \right)
&=Var \left( \frac{1}{M_X M_Y} \sum_{i,j} \sum_{k \neq i}
1_{T_{V_{UXY\tilde{X}}}}(\mathbf{u}, X_i,Y_j,X_k)\right) \nonumber\\
& = \frac{1}{M^2_X M^2_Y} Var \left( \sum_{i,j} \sum_{k \neq i}
1_{T_{V_{UXY\tilde{X}}}}(\mathbf{u}, X_i,Y_j,X_k)\right)\nonumber\\
& = \frac{1}{M^2_X M^2_Y} Var \left( \sum_{j} \sum_{i}\sum_{k \neq
i} Q^j_{ik}\right)\nonumber\\
&= \frac{1}{M^2_X M^2_Y} Var \left( \sum_{j} \sum_{i}\sum_{k < i}
Q^j_{ik} + Q^j_{ki}\right) = \frac{1}{M^2_X M^2_Y} Var \left(
\sum_{j} \sum_{i}\sum_{k < i} J^j_{i,k} \right),\label{Var-N_X1}
\end{align}
where $J^j_{i,k} \triangleq Q^j_{ik} + Q^j_{ki}$,  $k < i$. One can
show that $J^j_{i,k}$'s are identically pairwise independent random
variables. Therefore, the $Var \left( N_{X} (X^{M_X} \times
Y^{M_Y},V_{UXY\tilde{X}}) \right)$ can be written as
\begin{align}
Var \left( N_{X} (X^{M_X} \times Y^{M_Y},V_{UXY\tilde{X}}) \right) =
\frac{1}{M^2_X M^2_Y} \sum_{j} \sum_{i}\sum_{k < i} Var
\left(J^j_{i,k} \right) \leq \frac{1}{2 M_Y} Var  \left(J^1_{2,1}
\right). \label{Var-N_X4}
\end{align}
To find the variance of $J^1_{2,1}$, let us consider the following
two cases for $V_{UXY\tilde{X}}$:
\begin{itemize}
\item $V_{UXY\tilde{X}}$ is a symmetric distribution, i.e., $V_{UXY\tilde{X}} = V_{U\tilde{X}YX}$. In this case $Q^1_{12}=Q^1_{21}$, therefore,

$$
J^1_{2,1}= \left\{
\begin{array}{lc}
2 & \text{with probability  } p \approx 2^{-n[I_V(X \wedge Y|U) +
I_V(\tilde{X} \wedge XY|U)
]}\\
0 & \;\;\;\;\;\;\text{with probability  } 1 - p
\end{array},\right.
$$
and the variance is upper bounded by
\begin{align}
Var(J^1_{2,1}) \leq E({J^1_{2,1}}^2) = 4 \times 2^{-n[I_V(X \wedge
Y|U) + I_V(\tilde{X} \wedge XY|Y)]},\label{Var-N_X2}
\end{align}
\item $V_{UXY\tilde{X}}$ is not a symmetric distribution. In this case, if $Q^j_{ik} = 1\Rightarrow Q^j_{ki}=0$. Therefore,
\begin{align}
\mathbb{P} \left(J^1_{2,1} =1\right) = \mathbb{P} \left( Q^1_{12} =1
\text{ or } Q^1_{21} =1\right) &= \mathbb{P} \left( Q^1_{12}
=1\right) + \mathbb{P} \left( Q^1_{21} =1\right) \nonumber\\&\leq 2
\times
 2^{-n[I_V(X \wedge
Y|U) + I_V(\tilde{X} \wedge XY|U)]}, \label{Var-N_X3}
\end{align}
therefore,
\begin{align}
Var(J^1_{2,1}) \leq E({J^1_{2,1}}^2) = 2 \times 2^{-n[I_V(X \wedge
Y|U) + I_V(\tilde{X} \wedge XY|U)]}.
\end{align}
\end{itemize}
By combining the results in~\eqref{Var-N_X4}-\eqref{Var-N_X3}, it
can be concluded that
\begin{align}
Var \left( N_{X} (X^{M_X} \times Y^{M_Y},V_{UXY\tilde{X}} \right)
\leq 2^{-n[I_V(X \wedge Y|U) +  I_V(\tilde{X} \wedge XY|U) + R_Y - 3
\delta ]}.
\end{align}
Similarly, it can be shown that
\begin{align}
Var \left( N_{Y} (X^{M_X} \times Y^{M_Y},V_{UXY\tilde{Y}}) \right)
\leq 2^{-n[I_V(X \wedge Y|U) +I_V(\tilde{Y} \wedge YX|U)+ R_X - 3
\delta ]}.
\end{align}
By defining $R^{jl}_{ik}
\triangleq1_{T_{V_{UXY\tilde{X}}}}(\mathbf{u}, X_i,Y_j,X_k,Y_l)$,
the variance of $N_{XY} (X^{M_X} \times
Y^{M_Y},V_{UXY\tilde{X}\tilde{Y}})$ can be upper-bounded as follows
\begin{align}
Var \left( N_{XY} (X^{M_X} \times Y^{M_Y},V_{UXY\tilde{X}\tilde{Y}})
\right)
&=Var \left( \frac{1}{M_X M_Y} \sum_{i,j} \sum_{\substack{k \neq i\\
l \neq j}}
1_{T_{V_{UXY\tilde{X}\tilde{Y}}}}(\mathbf{u}, X_i,Y_j,X_k,Y_l)\right) \nonumber\\
& = \frac{1}{M^2_X M^2_Y} Var \left( \sum_{i,j} \sum_{\substack{k
\neq i\\ l \neq j}}
1_{T_{V_{UXY\tilde{X}\tilde{Y}}}}(\mathbf{u}, X_i,Y_j,X_k,Y_l)\right)\nonumber\\
& = \frac{1}{M^2_X M^2_Y} Var \left( \sum_{i} \sum_{j}\sum_{\substack{k \neq i}} \sum_{l \neq j} R^{jl}_{ik}\right)\nonumber\\
&= \frac{1}{M^2_X M^2_Y} Var \left( \sum_{i} \sum_{j}\sum_{k < i}
\sum_{l \neq j} \left\{R^{jl}_{ik} + R^{jl}_{ki} + R^{lj}_{ik}
+R^{lj}_{ki}\right\}\right) \nonumber\\&= \frac{1}{M^2_X M^2_Y} Var
\left( \sum_{i} \sum_{j}\sum_{k < i} \sum_{j <l}S^{j,l}_{i,k}
\right),\label{Var-N_X1}
\end{align}
where $S^{j,l}_{i,k} \triangleq R^{jl}_{ik} + R^{jl}_{ki} +
R^{lj}_{ik} +R^{lj}_{ki}$,  $k < i$, $l < j$. It is easy to check
that $S^{j,l}_{i,k}$'s are identically pairwise independent random
variables. Therefore, the $Var \left( N_{XY} (X^{M_X} \times
Y^{M_Y},V_{UXY\tilde{X}\tilde{Y}}) \right)$ can be written as
\begin{align}
Var \left( N_{XY} (X^{M_X} \times Y^{M_Y},V_{UXY\tilde{X}\tilde{Y}})
\right) = \frac{1}{M^2_X M^2_Y} \sum_{i} \sum_{j}\sum_{k < i}
\sum_{l < j} Var \left(S^{j,l}_{i,k} \right) \leq \frac{1}{4} Var
\left(S^{1,2}_{1,2} \right). \label{Var-N_XY4}
\end{align}
By using a similar argument to~\eqref{Var-N_X2}-\eqref{Var-N_X3},
the variance of To find the variance of $S^{1,2}_{1,2}$,can be upper
bounded by
\begin{align}
Var \left(S^{1,2}_{1,2}\right) \leq 16 \times 2^{-n[I_V(X \wedge
Y|U) + I_V(\tilde{X} \wedge \tilde{Y}|U) +I_V(\tilde{X}\tilde{Y}
\wedge XY|U)- 4 \delta] },
\end{align}
and therefore,
\begin{align}
Var \left(N_{XY} (X^{M_X} \times
Y^{M_Y},V_{UXY\tilde{X}\tilde{Y}})\right) \leq 4 \times 2^{-n[I_V(X
\wedge Y|U) + I_V(\tilde{X} \wedge \tilde{Y}|U)
+I_V(\tilde{X}\tilde{Y} \wedge XY|U)- 4 \delta] }.
\end{align}
Now, by using the Chebychev's inequality, we can obtain the
following
\begin{align}
&\mathbb{P} \left( \left|N_{U} (X^{M_X} \times Y^{M_Y},V_{UXY})-
\mathbb{E} \left(N_{U} (X^{M_X} \times Y^{M_Y},V_{UXY})\right)
\right| \geq 2^{2n\delta}
\text{ for some } V_{UXY}\right) \nonumber \\
&\quad \quad \quad \quad \quad \quad \quad \quad \quad \quad \quad
\quad \leq \sum_{V_{UXY}} \mathbb{P} \left( \left|N_{U} (X^{M_X}
\times Y^{M_Y},V_{UXY})- \mathbb{E} \left(N_{U} (X^{M_X} \times
Y^{M_Y},V_{UXY})\right) \right| \geq
2^{2n\delta}\right)\nonumber\\
& \quad \quad \quad \quad \quad \quad \quad \quad \quad \quad \quad
\quad \leq \sum_{V_{UXY}} \frac{Var\left( N_{U}
(X^{M_X} \times  Y^{M_Y},V_{UXY})\right)}{2^{4n\delta}}\nonumber\\
& \quad \quad \quad \quad \quad \quad \quad \quad \quad \quad \quad
\quad \leq \sum_{V} 2^{-n[F_U(V) + R_X +R_Y +  2\delta]} \leq
2^{-n\delta}.\label{NU}
\end{align}
Similarly, it can be shown that
\begin{align}
&\mathbb{P} \left( \left|N_{X} (X^{M_X} \times
Y^{M_Y},V_{UXY\tilde{X}})- \mathbb{E} \left(N_{X} (X^{M_X} \times
Y^{M_Y},V_{UXY\tilde{X}})\right) \right| \geq 2^{2n\delta}
\text{ for some } V_{UXY\tilde{X}}\right) \leq 2^{-n\delta}\label{NX}\\
&\mathbb{P} \left( \left|N_{Y} (X^{M_X} \times
Y^{M_Y},V_{UXY\tilde{Y}})- \mathbb{E} \left(N_{Y} (X^{M_X} \times
Y^{M_Y},V_{UXY\tilde{Y}})\right) \right| \geq 2^{2n\delta}
\text{ for some } V_{UXY\tilde{Y}}\right) \leq 2^{-n\delta}\label{NY}\\
&\mathbb{P} \left( \left|N_{XY} (X^{M_X} \times
Y^{M_Y},V_{UXY\tilde{X}\tilde{Y}})- \mathbb{E} \left(N_{XY} (X^{M_X}
\times Y^{M_Y},V_{UXY\tilde{X}\tilde{Y}})\right) \right| \geq
2^{2n\delta} \text{ for some } V_{UXY\tilde{X}\tilde{Y}}\right) \leq
2^{-n\delta}.\label{NXY}
\end{align}
Now, by using the result of Lemma~\ref{packing1} and Markov's
inequality, it can be concluded that
\begin{align}
&\mathbb{P} \left(\Lambda_X (X^{M_X} \times
^{M_Y},V_{UXY\tilde{X}\hat{X}}) \geq
2^{-n\left(E^{X}_S(V_{UXY\tilde{X}\hat{X}}) - 5 \delta \right)
}\text{
for some } V_{UXY\tilde{X}\hat{X}}  \right) \quad \quad \quad \quad \quad\quad \quad \quad \quad \quad\quad \quad \quad \quad \quad\quad \quad \quad \quad \quad\nonumber\\
&  \quad \quad \quad \quad \quad\quad \quad \quad \quad \quad \leq
\sum_{V_{UXY\tilde{X}\hat{X}}} \mathbb{P} \left(\Lambda_X (X^{M_X}
\times ^{M_Y},V_{UXY\tilde{X}\hat{X}})  \geq
2^{-n\left(E^{X}_S(V_{UXY\tilde{Y}\hat{Y}}) - 5 \delta \right)}
\right)\nonumber\\
&  \quad \quad \quad \quad \quad\quad \quad \quad \quad \quad \leq
\sum_{V_{UXY\tilde{X}\hat{X}}} \frac{\mathbb{E}\left(\Lambda_X
(X^{M_X} \times
^{M_Y},V_{UXY\tilde{X}\hat{X}})\right)}{2^{-n\left(E^{X}_S(V_{UXY\tilde{X}\hat{X}})
- 5 \delta \right)}} \leq
\sum_{V_{UXY\tilde{X}\tilde{Y}\hat{X}\hat{Y}}} 2^{-n \delta} \leq
2^{-n\frac{\delta}{2}}.\label{LambdaX}
\end{align}
Similarly,
\begin{align}
&\mathbb{P} \left(\Lambda_Y (X^{M_X} \times
Y^{M_Y},V_{UXY\tilde{Y}\hat{Y}}) \geq
2^{-n\left(E^{Y}_S(V_{UXY\tilde{Y}\hat{Y}}) - 5 \delta \right)
}\text{ for some } V_{UXY\tilde{Y}\hat{Y}}  \right) \leq
2^{-n\frac{\delta}{2}},\label{LambdaY}
\end{align}
and
\begin{align}
&\mathbb{P} \left(\Lambda_{XY} (X^{M_X} \times
Y^{M_Y},V_{UXY\tilde{X}\tilde{Y}\hat{X}\hat{Y}}) \geq
2^{-n\left(E^{XY}_S(V_{UXY\tilde{X}\tilde{Y}\hat{X}\hat{Y}}) - 7
\delta \right) }\text{ for some }
V_{UXY\tilde{X}\tilde{Y}\hat{X}\hat{Y}}  \right) \leq
2^{-n\frac{\delta}{2}}.\label{LambdaXY}
\end{align}
Therefore, with probability $> 1- 7 \times 2^{-n\frac{\delta}{2}}$,
a code $C =C_X \times C_Y$ from random code ensemble satisfies the
conditions given in the lemma.
\end{proof}

\begin{proof} \textbf{(Lemma~\ref{packing2})}
Let  $C^r_X=\{\mathbf{x}_1,\mathbf{x}_2,...,\mathbf{x}_{M_X}\}$ and
$C^r_Y=\{\mathbf{y}_1,\mathbf{y}_2,...,\mathbf{y}_{M_Y}\}$ be the
collections of codewords whose existence is asserted in
Lemma~\ref{packing1}. Let us define
\begin{align}
\Pi(C^r_X \times C^r_Y) &\triangleq
\sum_{\substack{V_{UXY\tilde{X}\tilde{Y}}}} \Big\{ N_{U} (C^r_X
\times C^r_Y,V_{UXY}) 2^{n[F_{U}(V_{UXY})-6\delta]}
\nonumber\\
&\quad \quad \quad \quad \;+ N_{X} (C^r_X \times
C^r_Y,V_{UXY\tilde{X}})
2^{n[F_{X}(V_{UXY\tilde{X}})-6\delta]} \nonumber\\
&\quad \quad \quad \quad \;+ N_{Y} (C^r_X \times
C^r_Y,V_{UXY\tilde{Y}})
2^{n[F_{Y}(V_{UXY\tilde{Y}})-6\delta]} \nonumber\\
&\quad \quad \quad \quad \;+ N_{XY} (C^r_X \times
C^r_Y,V_{UXY\tilde{X}\tilde{Y}})
2^{n[F_{XY}(V_{UXY\tilde{X}\tilde{Y}})-6\delta]}\Big\}\label{120}\\
&\leq \sum_{\substack{V_{UXY\tilde{X}\tilde{Y}}}} 4 \times
2^{-n\delta} <
\frac{1}{2}\quad\quad\quad\quad\quad\quad\quad\;\;\;\label{121}
\end{align}
%
%
For $C^r= C^r_X \times C^r_Y$, and the sequence $\mathbf{u}$ defined
in random coding packing lemma, we define
\begin{eqnarray}
 L_U(C^r,V_{UXY}, i,j) &\triangleq& 1_{T_{V_{UXY}}}(\mathbf{u},
\mathbf{x}_i,
\mathbf{y}_j),\\
 L_X(C^r,V_{UXY\tilde{X}}, i,j) &\triangleq& \sum_{k \neq i}
1_{T_{V_{UXY\tilde{X}}}}(\mathbf{u}, \mathbf{x}_i, \mathbf{y}_j,
\mathbf{x}_k), \\
 L_Y(C^r,V_{UXY\tilde{Y}}, i,j) &\triangleq& \sum_{l \neq j}
1_{T_{V_{UXY\tilde{Y}}}}(\mathbf{u}, \mathbf{x}_i, \mathbf{y}_j,
\mathbf{y}_l),\\
 L_{XY} (C^r,V_{UXY\tilde{X}\tilde{Y}}, i,j) &\triangleq& \sum_{k
\neq i} \sum_{l \neq j}
1_{T_{V_{UXY\tilde{X}\tilde{Y}}}}(\mathbf{u}, \mathbf{x}_i,
\mathbf{y}_j, \mathbf{x}_k, \mathbf{y}_l).
\end{eqnarray}
By definition of $N_{\alpha}$,~\eqref{120} can be written as
\begin{align}
\Pi(C^r) = \frac{1}{M_X} \sum_{i=1}^{M_X} G(i) , \quad\quad\quad
\text{ for } \alpha= U,X,Y,XY,
\end{align}
where $G(i)$ is defined as follows:
\begin{align}
G(i) \triangleq \frac{1}{M_Y} \sum_{j=1}^{M_Y} \sum_{V_{UXY\tilde{X}\tilde{Y}}} \Big\{& L_U(C^r,V_{UXY}, i,j) 2^{n[F_{U}(V_{UXY})- 6 \delta]}\nonumber\\
+&L_X(C^r,V_{UXY\tilde{X}}, i,j) 2^{n[F_{X}(V_{UXY\tilde{X}})- 6 \delta]} \nonumber\\
+&L_X(C^r,V_{UXY\tilde{Y}}, i,j) 2^{n[F_{Y}(V_{UXY\tilde{Y}})- 6 \delta]} \nonumber\\
+&L_{XY}(C^r,V_{UXY\tilde{X}\tilde{Y}}, i,j)
2^{n[F_{XY}(V_{UXY\tilde{X}\tilde{Y}})- 6 \delta]} \Big\}.
\end{align}
By using~\eqref{121}, we see that the average of $G(i)$ over $C_X^r$
is upper bounded by $\frac{1}{2}$, therefore, there must exist
$\hat{M}_X \geq \frac{M_X}{2}$ codewords, $\mathbf{x}_i \in C_X^r$,
for which
\begin{align}
G(i) < 1.\label{MAC-proof-ex2}
\end{align}
Let us call this set of codewords as $C_X^{ex}$. Without loss of
generality, we assume $C^{ex}_X$ contains the first $\hat{M}_X$
sequences of $C^r_X$, i.e., $C^{ex}_X=\{\mathbf{x}_1,
\mathbf{x}_2,...,\mathbf{x}_{\hat{M}_X}\}$.  Consider the multiuser
code $C^{ex}_1 \triangleq C_X^{ex} \times C_Y$. By definition of
$L_{\alpha}$, $\alpha=U,X,Y,XY$,
\begin{align}
L_{\alpha} (C^{ex}_1,V,i,j) \leq L_{\alpha} (C^r,V,i,j)\;\;\;\;
\forall\;(\mathbf{x}_i.\mathbf{y}_j) \in
C^{ex}_1.\label{MAC-proof-ex3}
\end{align}
By combining~\eqref{MAC-proof-ex2} and~\eqref{MAC-proof-ex3}, we
conclude that for all $i \in \{1,2,...,\hat{M}_X\}$
\begin{align}
\frac{1}{M_Y} \sum_{j=1}^{M_Y} \sum_{V_{UXY\tilde{X}\tilde{Y}}} \Big\{& L_U(C_1^{ex},V_{UXY}, i,j) 2^{n[F_{U}(V_{UXY})- 6 \delta]}\nonumber\\
+&L_X(C_1^{ex},V_{UXY\tilde{X}}, i,j) 2^{n[F_{X}(V_{UXY\tilde{X}})- 6 \delta]} \nonumber\\
+&L_X(C_1^{ex},V_{UXY\tilde{Y}}, i,j) 2^{n[F_{Y}(V_{UXY\tilde{Y}})- 6 \delta]} \nonumber\\
+&L_{XY}(C_1^{ex},V_{UXY\tilde{X}\tilde{Y}}, i,j)
2^{n[F_{XY}(V_{UXY\tilde{X}\tilde{Y}})- 6 \delta]} \Big\} < 1,
\end{align}
which results in
\begin{align}
\sum_{j=1}^{M_Y} \sum_{V_{UXY\tilde{X}\tilde{Y}}} \Big\{& L_U(C_1^{ex},V_{UXY}, i,j) 2^{n[F_{U}(V_{UXY})- R_Y-6 \delta]}\nonumber\\
+&L_X(C_1^{ex},V_{UXY\tilde{X}}, i,j) 2^{n[F_{X}(V_{UXY\tilde{X}})- R_Y- 6 \delta]} \nonumber\\
+&L_X(C_1^{ex},V_{UXY\tilde{Y}}, i,j) 2^{n[F_{Y}(V_{UXY\tilde{Y}})- R_Y- 6 \delta]} \nonumber\\
+&L_{XY}(C_1^{ex},V_{UXY\tilde{X}\tilde{Y}}, i,j)
2^{n[F_{XY}(V_{UXY\tilde{X}\tilde{Y}})- R_Y- 6 \delta]} \Big\} < 1.
\end{align}
Since all terms in the summation are non-negative, we conclude that
\begin{align}
L_{\alpha} (C^{ex}_1,V,i,j) 2^{-n[F_{\alpha}(V)- R_Y - 6\delta]} < 1
\end{align}
for all $i \in \{1,2,...,\hat{M}_X\}$, $j \in \{1,2,...,M_Y\}$, all
$V \in \P(\U \times \X \times \Y  \times \X  \times \Y )$, and all
$\alpha =U,X,Y,XY$. Therefore,
\begin{align}
L_{\alpha} (C^{ex}_1,V,i,j) < 2^{-n[F_{\alpha}(V)- R_Y -
6\delta]}.\label{MAC-ex-final1}
\end{align}
On the other hand,~\eqref{121} can also be written as
\begin{align}
\Pi(C^r) = \frac{1}{M_Y} \sum_{j=1}^{M_Y} H(j) , \quad\quad\quad
\text{ for } \alpha= U,X,Y,XY,
\end{align}
where $H(j)$ is defined as
\begin{align}
H(j) \triangleq \frac{1}{M_X} \sum_{i=1}^{M_X} \sum_{V_{UXY\tilde{X}\tilde{Y}}} \Big\{& L_U(C^r,V_{UXY}, i,j) 2^{n[F_{U}(V_{UXY})- 6 \delta]}\nonumber\\
+&L_X(C^r,V_{UXY\tilde{X}}, i,j) 2^{n[F_{X}(V_{UXY\tilde{X}})- 6 \delta]} \nonumber\\
+&L_X(C^r,V_{UXY\tilde{Y}}, i,j) 2^{n[F_{Y}(V_{UXY\tilde{Y}})- 6 \delta]} \nonumber\\
+&L_{XY}(C^r,V_{UXY\tilde{X}\tilde{Y}}, i,j)
2^{n[F_{XY}(V_{UXY\tilde{X}\tilde{Y}})- 6 \delta]} \Big\}.
\end{align}
By a similar argument as we did before, we can show that there exist
$\hat{M}_Y \geq \frac{M_Y}{2}$ codewords,  $\mathbf{y}_j \in C_Y^r$,
for which
\begin{align}
H(j) < 1.\label{MAC-proof-ex5}
\end{align}
Let us call this set of codewords as $C_Y^{ex}$. Without loss of
generality, we assume $C^{ex}_Y$ contains the first $\hat{M}_Y$
sequences of $C^r_Y$, i.e., $C^{ex}_Y=\{\mathbf{y}_1,
\mathbf{y}_2,...,\mathbf{y}_{\hat{M}_Y}\}$.  Consider the multiuser
code $C^{ex}_2 \triangleq C_X \times C_Y^{ex}$. By definition of
$L_{\alpha}$, $\alpha=U,X,Y,XY$, we have
\begin{align}
L_{\alpha} (C^{ex}_2,V,i,j) \leq L_{\alpha} (C^r,V,i,j)\;\;\;\;
\forall\;(\mathbf{x}_i.\mathbf{y}_j) \in
C^{ex}_2.\label{MAC-proof-ex6}
\end{align}
By a similar argument as we did before, we can show that
\begin{align}
L_{\alpha} (C^{ex}_2,V,i,j) < 2^{-n[F_{\alpha}(V)- R_X -
6\delta]}.\label{MAC-ex-final2}
\end{align}
for all $i \in \{1,2,...,{M}_X\}$, $j \in \{1,2,...,\hat{M}_Y\}$,
all $V \in \P(\U \times \X  \times \Y  \times \X  \times \Y )$, and
all $\alpha =U,X,Y,XY$.

By combining~\eqref{MAC-ex-final1} and~\eqref{MAC-ex-final2}, we
conclude that, there exists a multiuser code $C^{ex}=C_X^* \times
C_Y^*$ with $M_X^* \times M_Y^*$ messages
\begin{align}
 M_X^* \geq \frac{2^{n(R_X- \delta)}}{2},\qquad
 M_Y^* \geq \frac{2^{n(R_Y- \delta)}}{2},\qquad
 M_X^* \times M_Y^* \geq \frac{2^{n(R_X+R_Y- 2\delta)}}{2}
\end{align}
such that for any pair of messages $(\mathbf{x}_i ,\mathbf{y}_j) \in
C^{ex}$, all $V \in \P(\U \times \X  \times \Y  \times \X  \times \Y
)$, and all $\alpha =U,X,Y,XY$,
\begin{align}
L_{\alpha} (C^{ex},V,i,j) < 2^{-n[F_{\alpha}(V)- \min\{R_X,R_Y\} -
6\delta]}\label{MAC-ex-final3}.
\end{align}
It is easy to check that
\begin{align}
\Pi (C^{ex}) \leq 2 \times \Pi(C^r) < 1,
\end{align}
therefore, $C^{ex}$, satisfies all the constraints
in~\eqref{expur-lem1}-\eqref{expur-lem4}.
Here, by method of expurgation, we end up with a code with a similar
average bound as we had for the original code. However, all pairs of
codewords in the new code also satisfy
\eqref{randomnew1}-\eqref{randomnew4}. Therefore, we did not lose
anything in terms of average performance, however, as we will see in
Theorem~\ref{randomcodingthm}, we would end up with a tighter bound
since we have more constraints on any particular pair of codewords
in our codebook pair.
\end{proof}

\begin{proof}\textbf{(Theorem~\ref{randomcodingthm})}
Let us do random coding. Fix any $\U$, $P_{XYU} \in \P_n(\X \times
\Y \times \U)$ such that $X-U-Y$, $R_X \geq 0$, $R_Y \geq 0$ ,
$\delta
> 0$, and $\mathbf{u} \in T_{P_U}$. Define $M_X$, $M_Y$ such that
\begin{eqnarray*}
2^{n(R_X-\delta)} \leq M_X \leq 2^{nR_X}\qquad\qquad
2^{n(R_Y-\delta)} \leq M_Y \leq 2^{nR_Y}
\end{eqnarray*}
Let $X^{M_X} \triangleq \Big(X_1,X_2,...,X_{M_X} \Big)$ and $
Y^{M_Y} \triangleq \Big(Y_1,Y_2,...,Y_{M_Y} \Big)$ be independent
random variables, where $X_i$s are uniformly distributed on
$T_{P_{X|U}}(\mathbf{u})$, and $Y_j$s are uniformly distributed on
$T_{P_{Y|U}}(\mathbf{u})$.

\noindent \textbf{Upper bound:} By taking expectation over
(\ref{upperboundMAC}), applying Lemma \ref{packing1}, and using the
continuity of information measures, we get the desired upper bound.

\noindent \textbf{Lower bound:} By taking expectation over
(\ref{lowerboundMAC}), applying Lemma \ref{packing1}, we get

\begin{align}
\mathbb{E}e(C,W) &\geq
\sum_{\substack{V_{UXY\tilde{X}Z} \\ \in \V^r_{X,n}}}
2^{-n(E_X^L+4\delta)}
\left[1- \sum_{\substack{V_{UXY\tilde{X}\hat{X}Z} :\\
V_{UXY\hat{X}Z}= V_{UXY\tilde{X}Z}}}
2^{-n(I_V(\hat{X} \wedge XY \tilde{X}Z|U)-R_x -7 \delta)}  \right]
\nonumber\\
&  +\sum_{\substack{V_{UXY\tilde{Y}Z} \\ \in \V^r_{Y,n}}}
2^{-n(E_Y^L+4\delta)}
\left[1   - \sum_{\substack{V_{UXY\tilde{Y}\hat{Y}Z} :\\
V_{UXY\hat{Y}Z}= V_{UXY\tilde{Y}Z}}}
2^{-n(I_V(\hat{Y} \wedge XY \tilde{Y}Z|U)-R_Y -7 \delta)}  \right]
\nonumber\\
& +\sum_{\substack{V_{UXY\tilde{X}\tilde{Y}Z} \\ \in \V^r_{XY,n}}}
2^{-n(E_{XY}^L+4\delta)}
\left[1   -
\sum_{\substack{V_{UXY\tilde{X}\hat{X}\tilde{Y}\hat{Y}Z} :\\
V_{UXY\hat{X}\hat{Y}Z}= V_{UXY\tilde{X}\tilde{Y}Z}}}
2^{-n(I_V(\hat{X}\hat{Y} \wedge XY \tilde{X}\tilde{Y}Z|U)-R_X-R_Y -7
\delta)}  \right]
\end{align}

Toward further simplification of this expression, we use the
following lemma.
\begin{lemma}\label{Opt-Lemma}
\begin{align}
\min_{\substack{V_{UXY\tilde{X}\hat{X}Z} :\\
V_{UXY\hat{X}Z}= V_{UXY\tilde{X}Z}}}  I_V(\hat{X} \wedge
XY\tilde{X}Z|U) = I_V(\tilde{X} \wedge XYZ|U)
\end{align}
\end{lemma}
\begin{proof}
Note that, for any $V_{UXY\tilde{X}\hat{X}Z}$,
\begin{align}
 I_V(\hat{X} \wedge
XY\tilde{X}Z|U) = I_V(\hat{X} \wedge XYZ|U) + I_V(\tilde{X} \wedge
\hat{X}| UXYZ),
\end{align}
therefore,
\begin{align}
\min_{\substack{V_{UXY\tilde{X}\hat{X}Z} :\\
V_{UXY\hat{X}Z}= V_{UXY\tilde{X}Z}}}  I_V(\hat{X} \wedge
XY\tilde{X}Z|U) \geq I_V(\hat{X} \wedge XYZ|U) = I_V(\tilde{X}
\wedge XYZ|U) .\label{XXXY>}
\end{align}
Now, consider $V^*_{UX\tilde{X}\hat{X}YZ}$ defined as
\begin{align}
V^*_{UXY\tilde{X}\hat{X}Z}(u,x,y, \tilde{x}, \hat{x}, z) =
V_{\tilde{X}|UXYZ}(\tilde{x}|u,x,y,z)
V_{\tilde{X}|UXYZ}(\hat{x}|u,x,y,z) V_{UXYZ}(u,x,y,z).
\end{align}
Note that $V^*_{UX\hat{X}YZ} = V^*_{UX\tilde{X}YZ}$, and $\tilde{X}
- (U,X,Y,Z) - \hat{X}$. Therefore,
\begin{align}
I_{V^*}(\hat{X} \wedge XY\tilde{X}Z|U) = I_V(\hat{X} \wedge XYZ|U) =
I_V(\tilde{X} \wedge XYZ|U).\label{lemma-XXXY}
\end{align}
By combining~\eqref{XXXY>} and~\eqref{lemma-XXXY}, the proof is
complete.
\end{proof}

Using the above lemma, the average probability of error can be bounded
from below as
\begin{align}
\bar{P_e} &\geq  \sum_{\substack{V_{UXY\tilde{X}Z} \in \V^r_{X,n} \\
 I(\tilde{X}
\wedge XYZ |U) > R_X + 12 \delta}} 2^{-nE_X^L} +
\sum_{\substack{V_{UXY\tilde{Y}Z} \in \V^r_{Y,n}\\ I(\tilde{Y}
\wedge XYZ |U) > R_Y  + 12 \delta}} 2^{-nE_Y^L}
+\sum_{\substack{V_{UXY\tilde{X}\tilde{Y}Z} \in
\V^r_{XY,n}\\I_V(\tilde{X}\tilde{Y} \wedge XY|U)+I_V(\tilde{X}
\wedge \tilde{Y}| U)> \\R_X + R_Y + 14 \delta}} 2^{-nE_{XY}^L}
\nonumber
\end{align}
Using the continuity argument, the lower bound on the average error
probability follows.
\end{proof}
\begin{proof}\textbf{(Theorem~\ref{Typrandomcodingthm})}
As was done in Theorem~\ref{th:p2p_typical} for the point-to-point
case, here, we will obtain higher error exponents for almost all
codes by removing certain types from the constraint sets $\V_X^r$,
$\V_Y^r$ and $\V_{XY}^r$. Let us define the sets of $n$-types
$\V^t_X$, $\V^t_X$ and $\V^t_{XY}$ as follows:
\begin{align}
&\V^t_{X,n} \triangleq \left\{
\begin{array}{lc}
V_{UXY\tilde{X}}: & V_{XU}=V_{\tilde{X}U}=P_{XU}, V_{YU}=P_{YU} \\
 & F_U(V_{UXY}) , F_U(V_{U\tilde{X}Y}) \leq R_X +R_Y\\
 & F_{X}(V_{UXY\tilde{X}}) \leq R_X+R_Y
\end{array}\right\}\\
&\V^t_{Y,n} \triangleq \left\{
\begin{array}{lc}
V_{UXY\tilde{Y}}: & V_{XU}=P_{XU}, V_{YU}=V_{\tilde{Y}U}=P_{YU} \\
 & F_U(V_{UXY}) , F_U(V_{UX\tilde{Y}}) \leq R_X +R_Y\\
 & F_{Y}(V_{UXY\tilde{Y}}) \leq R_X+R_Y
\end{array}\right\}\\
\V^t_{XY,n} \triangleq &\left\{
\begin{array}{lc}
V_{UXY\tilde{X}\tilde{Y}}: & V_{UXY\tilde{X}}, V_{UX\tilde{Y}\tilde{X}} \in \V^t_X,\;\;\; V_{UXY\tilde{Y}}, V_{U\tilde{X}Y\tilde{Y}} \in \V^t_Y \\
  & F_{XY}(V_{UXY\tilde{X}\tilde{Y}}), F_{XY}(V_{U\tilde{X}YX\tilde{Y}}) \leq R_X+R_Y
\end{array}\right\}
\end{align}
\begin{lemma}\label{typical-N-Prop-lem}
Let $C=C_X \times C_Y$ be one of the multiuser codes whose existence
is asserted in the Typical random coding packing lemma. The
following hold:
\begin{align}
\text{If } V_{UXY\tilde{X}} \in (\V^t_{X,n})^c &\Rightarrow N_X(C,V_{UXY\tilde{X}})= 0,\\
\text{If } V_{UXY\tilde{Y}} \in (\V^t_{Y,n})^c &\Rightarrow N_Y(C,V_{UXY\tilde{Y}})= 0,\\
\text{If } V_{UXY\tilde{X}\tilde{Y}} \in (\V^t_{XY,n})^c
&\Rightarrow N_{XY}(C,V_{UXY\tilde{X}\tilde{Y}})= 0.
\end{align}
\end{lemma}
\begin{proof} Consider $V_{UXY\tilde{X}} \in (\V^t_{X,n})^c$. If $V_{XU}
\neq P_{XU}$ or $V_{\tilde{X}U} \neq P_{XU}$ or $V_{YU} \neq
P_{YU}$, it is clear that
\begin{align}
N_X(C,V_{UXY\tilde{X}})= 0.
\end{align}
Now, let us assume $F_U(V_{UXY}) > R_X + R_Y + 3 \delta$. In this
case, by using \eqref{Fu}, we conclude that
\begin{align}
N_U(C,V_{UXY}) < 2^{-n(R_X +R_Y)} \Rightarrow \sum_{i=1}^{M_X}
\sum_{j=1}^{M_Y} 1_{T_{V_{UXY}}}(\mathbf{u}, \mathbf{x}_i,
\mathbf{y}_j) <1 \Rightarrow \sum_{i=1}^{M_X} \sum_{j=1}^{M_Y}
1_{T_{V_{UXY}}}(\mathbf{u}, \mathbf{x}_i, \mathbf{y}_j)=0,
\end{align}
and as a result, $N_U(C,V_{UXY})=0$. Now, note that
\begin{align}
N_X(C,V_{UXY\tilde{X}}) &= \frac{1}{M_X M_Y} \sum_{i=1}^{M_X} \sum_{j=1}^{M_Y} \sum_{k \neq i} 1_{T_{V_{UXY\tilde{X}}}}(\mathbf{u},\mathbf{x}_i,\mathbf{y}_j,\mathbf{x}_k)\nonumber\\
& \leq \frac{1}{M_X M_Y} \sum_{i=1}^{M_X} \sum_{j=1}^{M_Y} \sum_{k \neq i} 1_{T_{V_{UXY}}}(\mathbf{u},\mathbf{x}_i,\mathbf{y}_j)\nonumber\\
& = 2^{nR_X} N_U(C,V_{UXY}) =0,
\end{align}
therefore, $N_X(C,V_{UXY\tilde{X}})= 0$. Similarly, if
$F_U(V_{U\tilde{X}Y}) > R_X + R_Y + 3 \delta$,
\begin{align}
N_U(C,V_{U\tilde{X}Y}) < 2^{-n(R_X +R_Y)} \Rightarrow
\sum_{i=1}^{M_X} \sum_{j=1}^{M_Y}
1_{T_{V_{U\tilde{X}Y}}}(\mathbf{u}, \mathbf{x}_i, \mathbf{y}_j) <1
\Rightarrow \sum_{i=1}^{M_X} \sum_{j=1}^{M_Y}
1_{T_{V_{U\tilde{X}Y}}}(\mathbf{u}, \mathbf{x}_i, \mathbf{y}_j)=0,
\end{align}
and as a result, $N_U(C_X,C_Y,V_{U\tilde{X}Y})=0$. Also, note that
\begin{align}
N_X(C,V_{UXY\tilde{X}}) &= \frac{1}{M_X M_Y} \sum_{i=1}^{M_X}
\sum_{j=1}^{M_Y} \sum_{k \neq i}
1_{T_{V_{UXY\tilde{X}}}}(\mathbf{u},\mathbf{x}_i,\mathbf{y}_j,\mathbf{x}_k)\nonumber\\
& \leq \frac{1}{M_X M_Y} \sum_{i=1}^{M_X} \sum_{j=1}^{M_Y} \sum_{k
\neq i}
1_{T_{V_{U\tilde{X}Y}}}(\mathbf{u},\mathbf{x}_k,\mathbf{y}_j) =0,
\end{align}
therefore, $N_X(C,V_{UXY\tilde{X}})= 0$. If $F_X(V_{UXY\tilde{X}}) >
R_X + R_Y + 5 \delta$, by the property of the code derived in
Lemma~\ref{packing3}, we observe that
$N_X(C_X,C_Y,V_{UXY\tilde{X}})= 0$. Similarly, by doing a similar
argument, it can be concluded that
\begin{align}
\text{If } V_{UXY\tilde{Y}} \in (\V^t_{Y,n})^c &\Rightarrow
N_Y(C,V_{UXY\tilde{Y}})= 0,
\end{align}
and
\begin{align}
\text{If } V_{UXY\tilde{X}\tilde{Y}} \in (\V^t_{XY,n})^c
&\Rightarrow N_{XY}(C,V_{UXY\tilde{X}\tilde{Y}})= 0.
\end{align}
\end{proof}

\noindent \textbf{Upper bound:} We will follow the techniques used
in Theorem~\ref{randomcodingthm} to provide lower and upper bounds
on the average probability of error of almost all codes in the
random coding ensemble. For this, we will use the results of
Lemma~\ref{packing2}. Consider any typical two-user code $C= C_X
\times C_Y$ whose existence was established in Lemma~\ref{packing3}.
Applying (\ref{upperboundMAC}) on $C $, and using the continuity
argument, we conclude that
\begin{align}
e(C,W) &\leq \sum_{\substack{V_{UXY\tilde{X}Z} \in \V^r_{X,n} \cap
\V^t_{X,n}}} 2^{-n[D(V_{Z|XYU} || W|V_{XYU}) +  I_V(X \wedge Y|U) +
|I_V(\tilde{X} \wedge XYZ|U) -R_X|^+ - 5 \delta]} \nonumber\\
&  \quad + \sum_{\substack{V_{UXY\tilde{Y}Z} \in \V^r_{Y,n} \cap
\V^t_{Y,n}}}
2^{-n[D(V_{Z|XYU} || W|V_{XYU}) + I_V(X \wedge Y|U) +
|I_V(\tilde{Y} \wedge XYZ|U) -R_Y|^+ - 5 \delta]} \nonumber\\
&  \quad + \sum_{\substack{V_{UXY\tilde{X}\tilde{Y}Z}\\ \in
\V^r_{XY,n} \cap \V^t_{XY,n}}} 2^{-n[D(V_{Z|XYU} || W|V_{XYU}) +I_V(X
\wedge Y|U) + |I_V(\tilde{X} \wedge \tilde{Y}|U)
+I_V(\tilde{X}\tilde{Y} \wedge XYZ|U)-R_X -R_Y|^+ -5\delta
]}\nonumber\\
&\leq 2^{-n\left[E_T(R_X,R_Y,W,P_{UXY})-6 \delta\right]}
\end{align}
whenever $n \geq n_1(|\Z|,|\X|,|\Y|,|\U|, \delta)$, where
$E_{T}(R_X,R_Y,W,P_{XYU})$ is defined in the statement of the
theorem.

\noindent \textbf{Lower bound:} In the following, we obtain a lower
bound on the average error probability of code $C=C_X \times C_Y$.
Applying (\ref{lowerboundMAC}) on $C$, then using (a) Lemma
\ref{packing3} and (b) the fact that for $V \notin V^t_{X,n}$, we
have $A^X_{i,j} \geq 0$, and similar such facts about $A^Y$ and
$A^{XY}$, we get
\begin{align}
e(C,W) &\geq
\sum_{\substack{V_{UXY\tilde{X}Z} \\ \in \V^r_{X,n} \cap V^t_{X,n}}}
2^{-n(E_X^L+4\delta)}
\left|1- \sum_{\substack{V_{UXY\tilde{X}\hat{X}Z} :\\
V_{UXY\hat{X}Z}= V_{UXY\tilde{X}Z}}} 2^{-n(I_V(\hat{X} \wedge XY
\tilde{X}Z|U)-R_x -7 \delta)}  \right|^+
\nonumber\\
&  +\sum_{\substack{V_{UXY\tilde{Y}Z} \\ \in \V^r_{Y,n} \cap V^t_{Y,n}}}
2^{-n(E_Y^L+4\delta)}
\left|1   - \sum_{\substack{V_{UXY\tilde{Y}\hat{Y}Z} :\\
V_{UXY\hat{Y}Z}= V_{UXY\tilde{Y}Z}}} 2^{-n(I_V(\hat{Y} \wedge XY
\tilde{Y}Z|U)-R_Y -7 \delta)}  \right|^+
\nonumber\\
& +\sum_{\substack{V_{UXY\tilde{X}\tilde{Y}Z} \\ \in \V^r_{XY,n}
\cap V^t_{XY,n}}} 2^{-n(E_{XY}^L+4\delta)} \left|1   -
\sum_{\substack{V_{UXY\tilde{X}\hat{X}\tilde{Y}\hat{Y}Z} :\\
V_{UXY\hat{X}\hat{Y}Z}= V_{UXY\tilde{X}\tilde{Y}Z}}}
2^{-n(I_V(\hat{X}\hat{Y} \wedge XY \tilde{X}\tilde{Y}Z|U)-R_X-R_Y -7
\delta)}  \right|^+
\end{align}
This expression can be simplified as follows.
\begin{align}
e(C,W) &\geq  \sum_{\substack{V_{UXY\tilde{X}Z} \in \V^r_{X,n}  \cap
 V^t_{X,n} \\
 I(\tilde{X}
\wedge XYZ |U) > R_X + 12 \delta}} 2^{-nE_X^L}
+ \sum_{\substack{V_{UXY\tilde{Y}Z} \in \V^r_{Y,n} \cap
 V^t_{Y,n} \\ I(\tilde{Y}
\wedge XYZ |U) > R_Y  + 12 \delta}} 2^{-nE_Y^L}
+\sum_{\substack{V_{UXY\tilde{X}\tilde{Y}Z} \in
\V^r_{XY,n} \cap
 V^t_{XY,n} \\ I_V(\tilde{X}\tilde{Y} \wedge XY|U)+I_V(\tilde{X} \wedge \tilde{Y}| U)> \\R_X + R_Y + 14
\delta}} 2^{-nE_{XY}^L} \nonumber
\end{align}
Using the continuity argument, the lower bound on the average error
probability follows.
\end{proof}

\begin{proof}\textbf{(Theorem~\ref{expurgatedthm})}
Fix $\U$, $\P_{XYU} \in \P_n(\X \times \Y \times \U)$ with $X-U-Y$,
$R_X \geq 0$, $R_Y \geq 0$, $\delta > 0$, and $\mathbf{u} \in
T_{P_U}$. Let $C^*= C^*_X \times C^*_Y$ be the multiuser code whose
existence is asserted in Lemma~\ref{packing2}.
Taking into account the given $\mathbf{u}$, the $\alpha$-decoding
yields the decoding sets
\begin{eqnarray}
D_{ij}= \{\mathbf{z}: \alpha(\mathbf{u},\mathbf{x}_i,\mathbf{y}_j,
\mathbf{z}) \leq \alpha(\mathbf{u},\mathbf{x}_k,\mathbf{y}_l,
\mathbf{z}) \text{ for all } (k,l) \neq (i,j)\}\nonumber.
\end{eqnarray}
Let us define the collection of $n$-types $\V^x_{X,n}$, $\V^x_{Y,n}$
and $\V^x_{XY,n}$ as follows:
\begin{align}
&\V^x_{X,n} \triangleq \left\{
\begin{array}{lc}
V_{UXY\tilde{X}}: & V_{XU}=V_{\tilde{X}U}=P_{XU}, V_{YU}=P_{YU} \\
 & F_U(V_{UXY}) , F_U(V_{U\tilde{X}Y}) \leq \min\{R_X, R_Y\}\\
 & F_{X}(V_{UXY\tilde{X}}) \leq \min\{R_X, R_Y\}
\end{array}\right\}\\
&\V^x_{Y,n} \triangleq \left\{
\begin{array}{lc}
V_{UXY\tilde{Y}}: & V_{XU}=P_{XU}, V_{YU}=V_{\tilde{Y}U}=P_{YU} \\
 & F_U(V_{UXY}) , F_U(V_{UX\tilde{Y}}) \leq \min\{R_X, R_Y\}\\
 & F_{Y}(V_{UXY\tilde{Y}}) \leq \min\{R_X, R_Y\}
\end{array}\right\}\\
\V^x_{XY,n} \triangleq &\left\{
\begin{array}{lc}
V_{UXY\tilde{X}\tilde{Y}}: & V_{UXY\tilde{X}}, V_{UX\tilde{Y}\tilde{X}} \in \V^x_X,\;\;\; V_{UXY\tilde{Y}}, V_{U\tilde{X}Y\tilde{Y}} \in \V^x_Y \\
  & F_{XY}(V_{UXY\tilde{X}\tilde{Y}}), F_{XY}(V_{U\tilde{X}YX\tilde{Y}}) \leq \min\{R_X, R_Y\}
\end{array}\right\}
\end{align}
\begin{lemma}\label{EXP-N-Prop-lem}
For the multiuser code $C^*=C^*_X \times C^*_Y$, the following
holds:
\begin{align}
\text{If } V_{UXY\tilde{X}} \in (\V^{x}_{X,n})^c &\Rightarrow N_X(C^*,V_{UXY\tilde{X}})= 0,\\
\text{If } V_{UXY\tilde{Y}} \in (\V^{x}_{Y,n})^c &\Rightarrow N_Y(C^*,V_{UXY\tilde{Y}})= 0,\\
\text{If } V_{UXY\tilde{X}\tilde{Y}} \in (\V^{x}_{XY,n})^c
&\Rightarrow N_{XY}(C^*,V_{UXY\tilde{X}\tilde{Y}})= 0.
\end{align}
\end{lemma}
\begin{proof}
The proof is very similar to the proof of
lemma~\ref{typical-N-Prop-lem}.
\end{proof}
The average error probability of $C^*$ can be obtained as follows in
a similar way that used in the proof of
Theorem~\ref{randomcodingthm} and Theorem~\ref{Typrandomcodingthm}.
\begin{align}
e(C^*,W) &\leq \sum_{\substack{V_{UXY\tilde{X}Z} \in \V^r_{X,n} \cap
\V^x_{X,n}}} 2^{-n[D(V_{Z|XYU}
|| W|V_{XYU}) +  I_V(X \wedge Y|U)  - 3 \delta]} \nonumber\\
&  \quad\quad\quad + \sum_{\substack{V_{UXY\tilde{Y}Z} \in
\V^r_{Y,n} \cap \V^x_{Y,n}}}
2^{-n[D(V_{Z|XYU} || W|V_{XYU}) + I_V(X \wedge Y|U) - 3 \delta]} \nonumber\\
&  \quad\quad\quad + \sum_{\substack{V_{UXY\tilde{X}\tilde{Y}Z}\\
\in \V^r_{XY,n} \cap \V^x_{XY,n} }} 2^{-n[D(V_{Z|XYU} || W|V_{XYU})
+I_V(X \wedge Y|U)
 -3\delta ]}. \label{TYP-W}
\end{align}
Now using the continuity argument  the statement of the theorem follows.
\end{proof}

\begin{proof}\textbf{(Theorem~\ref{finalthm})}
For any $V_{UXY\tilde{X}Z} \in \V^r_X$,
\begin{eqnarray}
H_V(XY|ZU) \geq H_V(\tilde{X}Y|ZU),\label{equivocation-rule}
\end{eqnarray}
therefore, by subtracting $H_{V}(Y|ZU)$ form both sides of
\eqref{equivocation-rule}, we can conclude that
\begin{eqnarray}
H_V(X|U)- I_V(X \wedge YZ|U) \geq H_V(\tilde{X}|U)- I_V(\tilde{X}
\wedge YZ|U), \label{Gen-Compare-Liu-X1}
\end{eqnarray}
Since $V_{XU}=V_{\tilde{X}U}=P_{XU}$, the last inequality is
equivalent to
\begin{equation}
I_V(X \wedge YZ|U) \leq I_V(\tilde{X} \wedge YZ|U).
\end{equation}
Since $I_V(\tilde{X} \wedge XYZ|U) \geq I_V(\tilde{X} \wedge YZ|U)$,
it can be seen that for any $V_{UXY\tilde{X}Z} \in \V^r_X$
\begin{equation}
I_V(\tilde{X} \wedge XYZ|U) \geq I_V(X \wedge YZ|U).
\end{equation}
Moreover, since
\begin{eqnarray}
\V^r_X \subseteq \left\{V_{UXY\tilde{X}Z}: V_{UXYZ} \in \V(P_{UXY})
\right\}
\end{eqnarray}
 it can be easily concluded that
\begin{equation*}
E^r_{X}(R_X,R_Y,W,P_{XYU}) \geq
E^{Liu}_{rX}(R_X,R_Y,W,P_{XYU}).\label{result1-X}
\end{equation*}
Similarly, for any $V_{UXY\tilde{Y}Z} \in \V^r_Y$,
\begin{eqnarray}
H_V(XY|ZU) \geq H_V(X\tilde{Y}|ZU).
\end{eqnarray}
By using the fact that, $V_{YU}=V_{\tilde{Y}U}=P_{YU}$, it can be
concluded that
\begin{equation}
I_V(\tilde{Y} \wedge XYZ|U) \geq I_V(Y \wedge XZ|U).
\end{equation}
Since
\begin{eqnarray}
\V^r_Y \subseteq \left\{V_{UXY\tilde{Y}Z}: V_{UXYZ} \in \V(P_{UXY})
\right\},
\end{eqnarray}
we conclude that
\begin{equation}
E^r_{Y}(R_X,R_Y,W,P_{XYU}) \geq
E^{Liu}_{rY}(R_X,R_Y,W,P_{XYU}).\label{result1-Y}
\end{equation}
Similarly, we can conclude that, for any $V_{UXY\tilde{X}\tilde{Y}Z}
\in \V^r_{XY}$,
\begin{equation}
I_V(\tilde{X}\tilde{Y} \wedge XYZ|U) + I(\tilde{X} \wedge
\tilde{Y}|U) \geq I_V(XY \wedge Z|U) + I(X \wedge Y|U).
\end{equation}
Since
\begin{eqnarray}
\V^r_{XY} \subseteq \left\{V_{UXY\tilde{X}\tilde{Y}Z}: V_{UXYZ} \in
\V(P_{UXY})  \right\},
\end{eqnarray}
it can be concluded that
\begin{equation}
E^r_{XY}(R_X,R_Y,W,P_{XYU}) \geq
E^{Liu}_{rXY}(R_X,R_Y,W,P_{XYU}).\label{result1-XY}
\end{equation}
By combining~\eqref{result1-X},~\eqref{result1-Y}
and~\eqref{result1-XY}, we conclude that~\eqref{r-Inequaltiy} holds.
Similarly, we can prove that~\eqref{t-Inequaltiy}
and~\eqref{ex-Inequaltiy} hold.
\end{proof}